\documentclass[11pt]{article}

\RequirePackage[OT1]{fontenc}
\usepackage{amsthm,amsmath}
\usepackage{amsfonts,amssymb}
\RequirePackage[colorlinks,citecolor=blue,urlcolor=blue]{hyperref}

\usepackage[dvipsnames]{xcolor}
\usepackage{xcolor}
\usepackage{color}
\usepackage{natbib}

\newcommand{\beginsupplement}{%
        \setcounter{table}{0}
        \renewcommand{\thetable}{S\arabic{table}}%
        \setcounter{figure}{0}
        \renewcommand{\thefigure}{S\arabic{figure}}%
        \setcounter{page}{1}
     }
     
\usepackage[figurename=Figure]{caption}
\usepackage[tablename=Table]{caption}

\usepackage{csquotes}
\usepackage{verbatim}

\usepackage[export]{adjustbox}

\newcommand{\bbS}{\ensuremath\mathbb{S}} 
\newcommand{\bbR}{\ensuremath\mathbb{R}} 

\theoremstyle{plain}

\newtheorem{prop}{Proposition}

\theoremstyle{definition}

\newtheorem{defi}{Definition}

\theoremstyle{remark}

\newcount\pagebak

\usepackage{upgreek}
\usepackage{amssymb}
\usepackage{systeme}
\usepackage{spalign}
\DeclareMathOperator{\xbold}{\mathbf{x}}
\DeclareMathOperator{\ybold}{\mathbf{y}}
\DeclareMathOperator{\zbold}{\mathbf{z}}
\DeclareMathOperator{\ebold}{\mathbf{e}}
\DeclareMathOperator{\dbold}{\mathbf{d}}
\DeclareMathOperator{\ubold}{\mathbf{u}}

\DeclareMathOperator{\wbold}{\mathbf{w}}

\DeclareMathOperator{\Xbold}{\mathbf{X}}
\DeclareMathOperator{\Zbold}{\mathbf{Z}}

\DeclareMathOperator{\1bold}{\mathbf{1}}
\DeclareMathOperator{\Ibold}{\mathbf{I}}
\DeclareMathOperator{\Hbold}{\mathbf{H}}

\DeclareMathOperator{\Jbold}{\mathbf{J}}
\DeclareMathOperator{\Gbold}{\mathbf{G}}
\DeclareMathOperator{\Tbold}{\mathbf{T}}

\DeclareMathOperator{\Phibold}{\mathbf{\Phi}}
\DeclareMathOperator{\Xibold}{\mathbf{\Xi}}

\DeclareMathOperator{\clr}{\text{clr}}
\DeclareMathOperator{\alr}{\text{alr}}
\DeclareMathOperator{\ilr}{\text{ilr}}

\DeclareMathOperator{\Cov}{\text{Cov}}

\newcommand{\norm}[1]{\left\lVert#1\right\rVert}

\title{A new class of $\alpha$-transformations for the spatial analysis of Compositional Data}

\author{Lucia Clarotto$^{a,*}$, Denis Allard$^b$, Alessandra Menafoglio$^c$ \\
\small{$^a$MINES ParisTech, PSL University, Centre de Geosciences, 77300 Fontainebleau, France}\\
\small{$^b$Biostatistics and Spatial Processes (BioSP), INRAE, 84914 Avignon, France}\\
\small{$^c$MOX, Department of Mathematics, Politecnico Milano, Milano, Italy}\\
\small{$^*$Corresponding author: lucia.clarotto@mines-paristech.fr}
}

\date{}

\begin{document}
	\maketitle
	
\begin{center}
\begin{minipage}{0.8\linewidth}

\noindent {\textbf{Abstract}}

Georeferenced compositional data are prominent in many scientific fields and in spatial statistics. This work addresses the problem of proposing models and methods to analyze and predict, through kriging, this type of data. To this purpose, a novel class of transformations, named the Isometric $\alpha$-transformation ($\alpha$-IT), is proposed, which encompasses the traditional Isometric Log-Ratio (ILR) transformation. It is shown that the ILR is the limit case of the $\alpha$-IT as $\alpha$ tends to 0 and that $\alpha=1$ corresponds to a linear transformation of the data. Unlike the ILR, the proposed transformation accepts 0s in the compositions when $\alpha>0$.  Maximum likelihood estimation of the parameter $\alpha$ is established.  Prediction using kriging on $\alpha$-IT transformed data is validated  on synthetic spatial compositional data, using prediction scores computed either in the geometry induced by the $\alpha$-IT, or in the simplex. Application to land cover data shows that the relative superiority of the various approaches w.r.t. a prediction objective depends on whether the compositions contained any zero component. When all components are positive, the limit cases (ILR or linear transformations) are optimal for none of the considered metrics.  An intermediate geometry,  corresponding to the $\alpha$-IT with maximum likelihood estimate, better describes the dataset in a geostatistical setting. When the amount of compositions with 0s is not negligible, some side-effects of the transformation gets amplified as $\alpha$ decreases, entailing poor kriging performances both within the $\alpha$-IT geometry and for metrics in the simplex.

\medskip

\end{minipage}
\end{center}
	
\newpage

\section{Introduction}\label{sec:intro}

In statistics, compositional data are quantitative descriptions of the parts of some whole, conveying relative information. Mathematically, compositional data are represented by points in a simplex, i.e. vectors with non-negative coordinates whose sum is constant. In this paper, we are interested in geostatistics or spatial statistics for compositional data, which is the area of statistics developing methods to analyze and predict, through kriging, compositional data associated with spatial or spatio-temporal phenomena. 

Georeferenced compositional datasets arise in varied fields of research, from geology to economics to chemistry to sociology, although most studies were historically concerned with topics in the geosciences.  In this setting, it took a long time to find a solution to the problem of how to perform a proper statistical analysis of \textit{closed} data -- namely data with a constant-sum constraint -- by taking into account the consequence of compositional constraints on correlations. Because some standard statistical techniques lose their applicability and classical interpretation when applied to compositional data, new techniques were needed. The first theoretically sound solution was proposed in the 1980's, when John Aitchison \citep{aitchison} built a consistent theory based on log-ratio transformations of compositional data. Later developments have shown that the mathematical foundation of a proper statistical analysis for this type of data is based on the definition of a specific geometry on the simplex, referred to as the \textit{Aitchison geometry}. Based on it and on the Principles of Compositional Data Analysis \citep{ilr-egozcue}, a relatively large body of literature has established a complete mathematical framework for statistical analysis which is nowadays widely accepted by the statistical community.  Spatial statistics has also been widely treated in the compositional community \citep{pawlo-olea, tolosana-missing, tolosana-06, tolosana-08,  tolosana-10, tolosana-11, tolosana-boogart}, by following the methods developed in the Aitchison geometry, i.e., by applying the log-ratio transformations that allow one to respect the Principles of Compositional Data Analysis. We refer to \citet{historical-review} for a complete review of the historical evolution of spatial analysis of compositional data through the Aitchison geometry. 

Even though the Aitchison's approach is nowadays mainstream in the analysis of compositional data, a number of authors have pushed forward alternative viewpoints, arguing that the choice of the appropriate method for the statistical data analysis should not be determined \textit{a priori} from a set of mathematical principles, but that this choice should rather depend, at least in part,  on the data \citep{scealy}. In this vein, \citet{tsagris-alpha} proposed a new family of transformations, called $\alpha$-transformations, parameterized by a constant $\alpha$. This parameter allows one to control the \textit{degree of transformation} applied to the data, ranging from a linear transformation ($\alpha=1$) to a log-ratio transformation ($\alpha=0$). In this setting, the parameter $\alpha$ is chosen in a data-driven manner, thus allowing one to render the approach application-specific. Note that the use of $\alpha$-transformations also enables one to deal with the presence of 0s in the compositions, unlike the log-ratio approach which is only suitable for strictly positive compositions. Beside these aspects, the approach based on $\alpha$-transformations proved effective in real studies, both in classification \citep{tsagris-classification} and in regression \citep{tsagris-regression}.

This paper follows the line pioneered by \citet{tsagris-alpha} and aims to establish a methodological framework for the statistical analysis of spatial compositional data, while finding a balance between the data-driven approach of \citet{tsagris-alpha} and that of the compositional community. In this vein, a novel set of transformations is considered, based on the concept of $\alpha$-contrasts -- which generalizes that of log-contrasts upon which the Aitchison geometry is based. We show that, similarly to \citet{tsagris-alpha}, our approach coincides with that of \citet{ilr-egozcue} for specific choices of the parameter $\alpha$, while attaining in general a more flexible framework than that of the Aitchison simplex. Besides, we shall also establish an explicit link between the covariance structure induced by the proposed class of transformations and that defined under the Aitchison geometry, opening broad perspectives of potential application on a wide range of covariance-based methods for exploratory and inferential data analyses.
In fact, the approach we propose is very general and not limited to spatial datasets. However, for the sake of brevity, this work will mainly focus on the problem of spatial analysis of compositional data with the lens of spatial prediction (kriging), which shall drive the formulation of the transformation, and particularly the choice of the “best" parameterization.

The remaining of this work is organized as follows: Section \ref{sec:state_art} introduces the theoretical concepts of compositional data, focusing on the log-ratio transformations, the spatial statistics methodologies applied in this field and the $\alpha$-transformations proposed by \citet{tsagris-alpha}. Section \ref{sec:alpha_ct_it} explores the new class of transformations considered in this paper, i.e., the Centered and Isometric $\alpha$-transformations ($\alpha$-CT and $\alpha$-IT), by discussing their properties, especially in a geostatistical setting. A maximum likelihood estimation method is proposed which maximizes the Gaussianity of the transformed data –    and, as a consequence, the  kriging  performances. Section \ref{sec:simulation_study} is interested in the application of the $\alpha$-IT to a simulated spatial dataset, in order to evaluate the improvement of this transformation over the classical Aitchison transformations. Section \ref{sec:copernicus} conducts a geostatistical analysis of land cover data, following the approach analyzed throughout the previous sections. Special attention is given to the analysis of data in presence of 0-parts in the compositions. Results show that when all parts are positive, the limit cases (ILR or linear transformations) are optimal for none of the considered metrics.  An intermediate geometry,  corresponding to the $\alpha$-IT with maximum likelihood estimate better describes the dataset in a geostatistical setting. When the amount of compositions with 0s is not negligible, some side-effects of the transformation gets amplified as $\alpha$ decreases, entailing poor kriging performances both within the $\alpha$-IT geometry and for metrics in the simplex. Finally, Section \ref{sec:discussion} reviews the main points of the paper.

\section{State-of-the-art analysis of spatial compositions}
\label{sec:state_art}
\subsection{Compositional Data analysis in the Aitchison simplex}\label{sec:coda}
\label{sec:usualCODA}
In this section we provide a brief overview of the key concepts underlying the analysis of compositional data through the Aitchison geometry. We refer the reader to, e.g.,  \citet{pawlo-egozcue} for a deeper account on the subject.

A column vector, $\xbold=[x_1,\dots,x_D]^\top$, is defined as a $D$-part composition when all its components are positive real numbers carrying only relative information. The sample space of compositional data is the simplex, defined as
\begin{equation}
\bbS^D=\{\xbold=[x_1,\dots,x_D]^\top \ | \  x_i>0, \ i =1,\dots,D; \ \sum_{i=1}^Dx_i=\kappa\},
\end{equation}
\noindent where $\kappa$ is an arbitrary constant which can be set to 1 without loss of generality. A $D$-part compositional vector $\zbold \in (0,\infty)^D$ can be represented as a point of the simplex $\bbS^D$ by use of the closure operator
\begin{equation}
    C: C(\zbold) = \xbold, \quad x_i = \frac{z_i}{\sum_{i=1}^D z_i} \quad \forall i=1,\dots,D.
    \label{eq:closure}
\end{equation}
The perturbation of a composition $\xbold \in \bbS^D$ by a composition $\ybold \in \bbS^D$ is
\begin{equation*}
\xbold \oplus \ybold = C \left( [x_1y_1, \dots, x_Dy_D]^\top \right).
\end{equation*}
The power transformation of a composition $\xbold \in \bbS^D$ by a constant $\alpha \in \bbR$ is
\begin{equation*}
\alpha \odot \xbold = C \left( [x_1^{\alpha}, \dots, x_D^{\alpha}]^\top \right).
\end{equation*}
In particular, $(-1) \odot \xbold = C \left( [1/x_1, \dots, 1/x_D]^\top \right)$. The neg-perturbation is then defined in the following way: $\xbold \ominus \ybold = \xbold \oplus (-1 \odot \ybold)$.
The simplex $(\bbS^D,\oplus,\odot)$, equipped with the perturbation operation and the power transformation, is a vector space \citep{aitchison, billheimer-01, pawlowsky-01}. This implies that all properties of translation and scalar multiplication hold. Moreover, by adding an inner product, a norm and a distance on the vector space, the simplex becomes an Euclidean metric space (i.e., a finite-dimensional Hilbert space). To refer to the properties of $(\bbS^D,\oplus,\odot)$ as an Euclidean metric space, we call it \textit{Aitchison geometry} on the simplex. We report here some definitions of the Aitchison geometry that will be useful for this work. The inner product of $\xbold, \ybold \in \bbS^D$ is
\begin{equation}
    \langle \xbold,\ybold \rangle_a = \frac{1}{2D}\sum_{i=1}^D \sum_{j=1}^D \ln\frac{x_i}{x_j}\ln\frac{y_i}{y_j},
    \label{eq:inner_product}
\end{equation}
and the norm of $\xbold \in \bbS^D$ is $\norm{\xbold}_a = \langle \xbold,\xbold \rangle_a$. The Aitchison distance between $\xbold$ and $\ybold \in \bbS^D$ is 
\begin{equation}
    d_a(\xbold, \ybold) = \norm{\xbold \ominus \ybold}_a = \sqrt{\frac{1}{2D}\sum_{i=1}^D \sum_{j=1}^D \left(\ln\frac{x_i}{x_j} - \ln\frac{y_i}{y_j}\right)^2}.
    \label{eq:aitchison_dist}
\end{equation}

\citet{aitchison} used the fact that for compositional data absolute quantities and units are irrelevant -- as interest lies in relative proportions of the components measured -- to introduce transformations based on log-ratios. We refer the reader to \citet{aitchison,principle-coord,buccianti} and references therein for a detailed account. Here, we focus on the  
Centered Log-Ratio transform (CLR) and on the Isometric Log-Ratio transform (ILR) which will be necessary for the sequel. 

The Centered Log-Ratio transformation (CLR) is defined as
\begin{equation}
\begin{aligned}
&\clr: \bbS^D \rightarrow \mathbb{H} \subset \bbR^D\\
\clr(\xbold)&=\left[\ln {\frac {x_1}{g(\xbold)}}, \cdots, \ln {\frac {x_D}{g(\xbold)}}\right]^\top,\label{eq:clr}
\end{aligned}
\end{equation}
where $g(\xbold)$ is the geometric mean of $\xbold$, $g(\xbold)=\left(\prod_{i=1}^{D} x_i \right)^{1/D}$. 

The CLR transform is an injective transformation between $\bbS^D$ and $\bbR^D$ which preserves distances, in the sense that $d_a(\xbold, \ybold)= d_e(\clr(\xbold), \clr(\ybold))$, with $d_e$ the Euclidean distance. However, the CLR is characterized by a constraint on the transformed sample, as the sum of the components of $\clr(\xbold)$ is 0 by definition.  The transformed sample thus lie on a hyper-plane, $\mathbb{H}$, which goes through the origin of $\bbR^D$ and is orthogonal to the vector of ones $[1, \dots, 1]^\top$. 
The CLR transformation cannot be directly associated with an orthogonal coordinate system in the simplex. For this reason, \citet{ilr-egozcue} introduced a new transformation, called Isometric Log-Ratio (ILR), defined as follows. 

Any vector $\xbold \in \bbS^D$ can be written as
\begin{equation*}
\xbold = \bigoplus_{i=1}^D\ln{x_i \odot \wbold_i},
\label{eq:generating_clr}
\end{equation*} 
where $\{\wbold_1,\dots,\wbold_D\}$ are the generator vectors $\wbold_i = C (\exp(\dbold_i))$, $i = 1,  \dots,D$, where $\dbold_i$ is the unit  vector of $\bbR^D$ associated to the $i$th coordinate. Omitting one vector of the generating system, a basis is obtained, e.g. $\{\wbold_1,\dots,\wbold_{D-1}\}$, which is not orthonormal. However, a new basis, orthonormal with respect to the inner product, can be obtained using the Gram-Schmidt procedure, so that it is possible to express a composition $\xbold \in \bbS^D$ as
\begin{equation*}
\xbold = \bigoplus_{i=1}^{D-1} x_i^{*} \odot \ebold_i, \quad x_i^{*} = \langle \xbold, \ebold_i \rangle_{a},
\end{equation*}
where $\{\ebold_1,\dots,\ebold_{D-1}\}$ is a generic orthonormal basis of the simplex $\bbS^D$ and $\langle \cdot, \cdot \rangle_{a}$ is the inner product in the Aitchison geometry of Equation \eqref{eq:inner_product}. 

The Isometric Log-Ratio (ILR) of a composition $\xbold$ is defined as
\begin{equation*}
\begin{aligned}
&\ilr: \bbS^D\rightarrow \bbR^{D-1}\\
\ilr(\xbold)= \xbold^* &= \left[\langle \xbold,\ebold_1\rangle_{a} ,\dots ,\langle \xbold,\ebold_{D-1}\rangle_{a} \right]^\top.
\end{aligned}
\end{equation*}
The ILR is an isomorphism between $(\bbS^D, d_a)$ and $(\bbR^{D-1}, d_e)$ which preserves distances (i.e., it is an isometry), and it can be retrieved from the CLR transform with the following equality:
\begin{equation*}
\ilr(\xbold) = \Hbold_D \clr(\xbold), 
\end{equation*}
where $\Hbold_D$ is the $(D-1,D)$ Helmert matrix whose rows are $\clr(\ebold_i)$. $\Hbold_D$ satisfies $\Hbold_D \Hbold_D^\top = \Ibold_{D-1}$ and $\Hbold_D^\top \Hbold_D = \Gbold_D$, where $\Gbold_D= (\Ibold_D - D^{-1}\Jbold_D)$ is the $D$-dimensional centering matrix and $\Jbold_D$ is the $(D,D)$ matrix of ones.

\subsection{Geostatistics for Compositional Data}\label{sec:geostat}

Throughout this section, we consider a $D$-part composition $\Xbold(s)$ with $X_i(s)>0, \ \text{for} \ i =1,\dots,D,\ \forall s \in \mathcal{D} \subset \bbR^d$. We further assume that it is second-order stationary, i.e. the expected values of the pairwise log-ratios at each point $s$ exist and do not depend on $s$, and the cross-covariances between every pairwise log-ratio at two different points $s_1$ and $s_2$ exist and depend only on the directional vector $h=s_2-s_1$. Such a compositional random field is said to be \textit{second-order Log-Ratio (LR) stationary}.  Let $\{s_j: j=1,\dots,n\}$ be a set of $n$ spatial locations in a spatial domain $\mathcal{D}$ containing georeferenced compositional data $\{\Xbold(s_j)\}_{j=1}^n=\{[X_1(s_j),\dots, X_D(s_j)]^\top\}_{j=1}^n$ in $\bbS^D$. Now, let $\Xbold(s_0)$ be the unobserved vector of compositional data at the prediction location $s_0$. The task of predicting $\Xbold(s_0)$ using linear combinations of $\{\Xbold(s_j)\}_{j=1}^n$ is known as cokriging in the geostatistics literature. We refer to \citet{wackernagel2003multivariate} for a general introduction to multivariate geostatistics and to cokriging.

Direct implementation of cokriging on $\Xbold(s_0)$ is  impractical for several reasons. For instance, the associated covariance matrix must satisfy the closure relations induced by the constant sum constraint, $\sum_{j=1}^D C_{ij}(h) = 0$ for all $h \in \bbR^d$ and all $i=1,\dots,D$,  which result in singular cokriging matrices. One way around would be to use cokriging on $(D-1)$ parts and deduce the last part from the closure property, but this approach also shows drawbacks. Predicted and simulated compositional vectors are expected to be elements of the simplex $\bbS^D$. Generally speaking, direct kriging or cokriging of the proportions, if achievable, cannot guarantee this property in all generality since kriging does not impose  non-linear constraints such as positivity \citep{walvoort2001compositional}.

In an ordinary cokriging setting, closure to 1 can be enforced by imposing identical kriging weights for all parts, since in this case
$$\sum_{i=1}^D X_i^\star(s_0)= \sum_{i=1}^D \sum_{j=1}^n \lambda_j X_i(s_0) = \sum_{j=1}^n \lambda_j \sum_{i=1}^D  X_i(s_0) = \sum_{j=1}^n \lambda_j =1,$$
where $[\lambda_1,\dots,\lambda_n]^\top$ is the unique vector of weights and $X_i^\star(s_0)$ is the cokriging of $X_i(s_0)$.
\citet{allard2018means} have shown that identical weights are obtained for all $s_0$ and all $\{\Xbold(s_j)\}_{j=1}^n$  if and only if the multivariate covariance model for $\Xbold(s)$ is proportional, i.e. it is the product of a covariance matrix and a single spatial covariance function. On the other hand, imposing positivity is possible, but requires quadratic programming \citep{walvoort2001compositional}.

Moreover, working directly on $\Xbold(s_0)$ does not account for the relative nature of parts which, following the above developments, should be properly acknowledged. To sum up, even though direct cokriging is in theory possible, it is not guaranteed to be optimal.
In \citet{tolosana-06}, an approach for geostatistics of compositional data based on the application of the principle of working in coordinates using Isometric Log-Ratio representations is presented in detail. For a recent review, see \citet{historical-review}. Following this approach, compositional data belonging to $\bbS^D$ are first transformed to a set of $(D-1)$ unbounded scores by means of the ILR. Then, the multivariate geostatistical techniques (e.g., covariance modeling, kriging, stochastic simulation) are applied to the scores. Finally, the resulting interpolated or simulated scores are back-transformed to obtain values in $\bbS^D$. Since this technique applies a transformation to the compositional dataset before any geostatistical method,  the spatial structure must be defined on the transformed dataset \citep{pawlo-olea}, as we describe in the following.

Let $\Xbold(s)$ be a second-order LR stationary compositional random field. The \textit{variation matrix}, $\Tbold(h)=[\tau_{ij}(h)]_{i,j=1}^D$ of $\Xbold(s)$ is the $(D, D)$ matrix whose elements are the autocovariances of the corresponding log-ratios
\begin{equation*}
\tau_{ij}(h)  = \Cov \left[\ln \left(\frac{X_i(s)}{X_j(s)}\right), \ln \left(\frac{X_i(s + h)}{X_j(s + h)}\right) \right].
\end{equation*}
The variation matrix is symmetric in the indices and in $h$ for any two components of $\Xbold(s)$. 

The covariance matrix of the Centered Log-Ratio (CLR) transformed random field $\Zbold_{\clr}(s)=\clr(\Xbold(s))$, called the \textit{CLR cross-covariance matrix}, is the $(D, D)$ matrix $\Xibold(h) = [\xi_{ij}(h)]_{i,j=1}^D$ which is the covariance between the elements of the CLR transformed data and the same elements lagged by $h$
\begin{equation*}
\xi_{ij}(h)= \Cov \left[\Zbold_{\clr}(s), \Zbold_{\clr}(s+h) \right]_{ij}= \Cov\left[\ln\left(\frac{X_i(s)}{g(\Xbold(s))}\right), \ln\left(\frac{X_j(s+h)}{g(\Xbold(s+h))}\right)\right].
\end{equation*}

The covariance matrix of the Isometric Log-Ratio (ILR) transformed data $\Zbold_{\ilr}(s)=\ilr(\Xbold(s))$, called the \textit{ILR cross-covariance matrix} (or \textit{coordinate cross-covariance matrix} in \citet{tolosana-06}), is the $(D-1, D-1)$ matrix $\Phibold(h)= [\phi_{ij}(h)]_{i,j=1}^{D-1}$ with
\begin{equation*}
 \phi_{ij}(h) = \Cov \left[\Zbold_{\ilr}(s), \Zbold_{\ilr}(s+h) \right]_{ij}.
\end{equation*} 

These cross-covariance matrices are related to each other in the following way 
\citep{pawlo-olea}:
\begin{equation}
\Xibold(h) + \Xibold^\top(h)  = - \Gbold_D \Tbold(h) \Gbold^\top_D,
\label{trans_xi_t}
\end{equation}
where $\Gbold_D$ is the centering matrix of dimension $D$ defined in Section \ref{sec:usualCODA}.
Moreover, the following relation between $\Xibold(h)$ and $\Phibold(h)$ can be proven by linearity: the ILR cross-covariance matrix $\Phibold(h)$ and the CLR cross-covariance matrix $\Xibold(h)$ satisfy
\begin{equation*}
\Phibold(h) = \Hbold_D \Xibold(h)  \Hbold_D^\top.
\end{equation*}

\subsection{The $\alpha$-transformation}\label{sec:alpha_trans}

In the last few years, \citet{tsagris-alpha} have proposed a new class of transformations, called $\alpha$-\textit{transformations}, that encompasses the Aitchison transformation in the sense that it allows to deal with 0-values in the compositions and that it retrieves the Aitchison geometry when $\alpha$ tends to 0. The $\alpha$-transformation, for $\alpha \neq 0$, of any compositional vector $\xbold \in \bbS^D$ is the mapping 
\begin{equation}
\begin{aligned}
&A_{\alpha} :\bbS^D \rightarrow \bbR^{D-1}\\
\zbold = A_{\alpha}&(\xbold) = \alpha^{-1} \Hbold_D \left(D\ubold_\alpha(\xbold) - \boldsymbol{1}_D \right),  \qquad \ubold_{\alpha}(\xbold) = C(\xbold^\alpha),
\end{aligned}
\label{eq:alpha_trans_tsagris}
\end{equation}
\noindent where $C(\cdot)$ is the closure operator defined in Equation \eqref{eq:closure} and the powering is applied component-wise. 

Detailed analyses of this transformation are available in \citet{tsagris-alpha}, \citet{tsagris-regression} and \citet{tsagris-classification}. Two important benefits of the $\alpha$-transformation are that it is well-defined for any $\alpha>0$ for compositions containing 0s and that it tends to the ILR transform as $\alpha$ approaches 0.

The $\alpha$-transformation was shown to yield good results when used to analyze compositional datasets, especially for classification problems \citep{tsagris-classification}. However, when analyzing the peculiarities of this transformation, \citet{tsagris-stewart} clarified that the images of the $\alpha$-transformation lie in a codomain that has a simplex shape, except for $\alpha=0$; in this case the codomain is the entire space $\bbR^{D-1}$. The simplex shape is due to the fact that a closure operator is introduced in the transformation, which allows for a closed form for the inverse transform. Beside this convenient property, the presence of the closure operation \eqref{eq:closure} does not seem to be mathematically necessary. Moreover, when embedded in a geostatistical setting, the spatial covariance structure is no longer explicitly connected to the variation matrix as in Equation \eqref{trans_xi_t}, which introduces some difficulties for the modeling of the spatial covariance that is needed for spatial interpolation. In order to overcome these limitations, we introduce a new family of $\alpha$-transformations in the next Section, that we name Centered and Isometric $\alpha$-transformations.

\section{The Centered and Isometric $\alpha$-transformations for spatial Compositional Data}
\label{sec:alpha_ct_it}
\subsection{The Centered and Isometric $\alpha$-transformations}\label{sec:alpha_it}

We here introduce a new class of transformations called \textit{Centered} and \textit{Isometric} $\alpha$\textit{-transformations} ($\alpha$-CT and $\alpha$-IT). Let us define
$\bbS_0^D=\{\xbold=[x_1,\dots,x_D]^\top \ | \  x_i \geq 0, \ i =1,\dots,D; \ \sum_{i=1}^Dx_i=1\}$ the simplex that admits one or more 0-values in its components.

\begin{defi}
Let $\alpha > 0$. The Centered $\alpha$-transformation ($\alpha$-CT) of a compositional vector $\xbold \in \bbS_0^D$ is the mapping $A_{\alpha-CT}  :  \bbS_0^D \rightarrow \bbR^{D}$
\begin{equation}
\ubold_{\alpha-CT} = A_{\alpha-CT}(\xbold)  =  \alpha^{-1} \Gbold_D \xbold^\alpha, 
\label{eq:alfa_ct}
\end{equation}
where $\boldsymbol{G}_D$ is the $(D,D)$ centering matrix defined above. The Isometric $\alpha$-transformation ($\alpha$-IT) of a compositional vector $\xbold \in \bbS_0^D$ is the mapping $A_{\alpha-IT} :\bbS_0^D \rightarrow \bbR^{D-1}$
\begin{equation}
\zbold_{\alpha-IT} = A_{\alpha-IT}(\xbold) 
= \alpha^{-1} \Hbold_D \Gbold_D \xbold^\alpha = \alpha^{-1} \Hbold_D \xbold^\alpha,
\label{eq:alfa_it}
\end{equation}
where $\Hbold_D$ is the $(D-1, D)$ Helmert matrix defined in Section \ref{sec:coda}.
\end{defi}

To better understand the $\alpha$-CT, it is useful to make the link between Equation \eqref{eq:alfa_ct} and the CLR transformation \eqref{eq:clr} more explicit. Indeed, the $\alpha$-CT of a compositional vector $\xbold$ in $\bbS^D$ reads
\begin{equation}
\ubold_{\alpha-CT} =  \left[\alpha^{-1} \left( x_1^\alpha-\frac{1}{D}\sum_{i=1}^Dx_i^\alpha\right),\dots, \alpha^{-1} \left( x_D^\alpha-\frac{1}{D}\sum_{i=1}^Dx_i^\alpha\right)
\right]^\top,
\label{eq:alfa_ct_bis}
\end{equation}
whereas the CLR transform of $\xbold$ in \eqref{eq:clr} is equivalent to 
\begin{equation}
\clr(\xbold)  =  \left[\left( \ln(x_1)-\frac{1}{D}\sum_{i=1}^D\ln(x_i)\right), \dots, \left( \ln(x_D)-\frac{1}{D}\sum_{i=1}^D\ln(x_i)\right)
\right]^\top.
\label{eq:clr_bis}
\end{equation}
In fact, similarly as for the CLR transform, the $\alpha$-CT of $\xbold$ operates a centering of a (power) transformation of $\xbold$ with respect to its average value, thus yielding a transformed vector $\ubold_{\alpha-CT}$ characterized by a zero-sum constraint. Hence, both $\ubold_{\alpha-CT}$ and $\clr(\xbold)$ lie on the same hyper-plane $\mathbb{H}\subset\bbR^D$ defined in Section \ref{sec:coda}. Analogously as for the ILR transform, the $\alpha$-IT then identifies a set of $(D-1)$ coordinates over an orthonormal basis of $\mathbb{H}$. The following Proposition sheds further light on the link between the proposed class of transformations and those in use under the Aitchison geometry. 

\begin{prop}
The CLR transform is retrieved from the Centered $\alpha$-transformation when $\alpha \to 0$ and the ILR transform is retrieved from the Isometric $\alpha$-transformation when $\alpha \to 0$, provided that $\xbold \in \bbS^D$.
\end{prop}

\begin{proof}
Since both the ILR transform and the $\alpha$-IT are defined as the multiplication of the Helmert matrix $\Hbold_D$ by the CLR transform and the $\alpha$-CT respectively, it is sufficient to prove that the $\alpha$-CT tends to the CLR transform when $\alpha \to 0$. From the definition of the geometric mean, we get $\ln g(\xbold) = D^{-1} \sum_{i=1}^D \ln x_i$, which does exist since $\xbold \in \bbS^D$. Notice that $\Gbold_D \xbold^\alpha = \xbold^\alpha - m_\alpha(\xbold)\boldsymbol{1}_D$, where $\boldsymbol{1}_D$ is a $D$-vector of ones and where the $\alpha$-mean $m_\alpha(\xbold)$ is 
$m_\alpha(\xbold)  = 1/D \sum_{i=1}^D x_i^\alpha$.
From \eqref{eq:alfa_ct} we get
\begin{equation*}
\begin{aligned}
\lim_{\alpha \to 0} u_{i, \alpha-CT} &=
\lim_{\alpha \to 0} \frac{x_i^\alpha - m_\alpha}{\alpha} = 
\lim_{\alpha \to 0} \left(\frac{x_i^\alpha - 1}{\alpha} - \frac{m_\alpha(\xbold) -1}{\alpha}\right) \\
&  = \ln(x_i) - \ln(g(\xbold)) = \ln\left(\frac{x_i}{g(\xbold)}\right) = \clr(\xbold)_i.
\end{aligned}
\end{equation*}
Hence, we have 
\begin{equation*}
\lim_{\alpha \to 0} \ubold_{\alpha-CT} = \clr(\xbold)
\end{equation*}
and
\begin{equation*}
\lim_{\alpha \to 0}\zbold_{\alpha-IT} = \lim_{\alpha \to 0}\Hbold_D \ubold_{\alpha-CT} = \Hbold_D \clr(\xbold) = \ilr(\xbold).
\end{equation*}
\end{proof}

This result establishes that the $\alpha$-CT and the $\alpha$-IT are a generalization of the CLR and ILR transform, respectively. Notice that, when $\alpha=1$, these transformations boil down to linear transformations of the compositions. On $[0,1]$, the parameter $\alpha$ thus offers a modeling flexibility that allows to interpolate smoothly between linear and log-transformations and therefore, hopefully, better adapt to the data. Other positive values of $\alpha$ are also possible. Negative values of $\alpha$ can also be considered for the $\alpha$-IT and $\alpha$-CT, provided that the compositions are in $\bbS^D$, i.e., that no 0-values occur. From now on, we will restrict the use of the $\alpha$-CT and $\alpha$-IT to $\alpha\geq0$. 

It is worth emphasizing that the $\alpha$-CT and the $\alpha$-IT (with $\alpha > 0$) can be applied to any composition of $\bbS_0^D$, including those lying on the border of the simplex. This point is an important improvement with respect to the usual CLR and ILR transforms that cannot be applied to compositions with one or more 0s. This point will be further discussed when analyzing the Copernicus Land Cover dataset in Section \ref{sec:spatial_with0}.

\subsection{Inverse of the Isometric $\alpha$-transformation}
\label{sec:back-transform}
Let us consider a vector $\zbold$ belonging to the codomain of $A_{\alpha-IT}$, with $\alpha > 0$.  In order to obtain the composition $\xbold$ such that $\zbold = A_{\alpha-IT}(\xbold)$, one must solve Equation \eqref{eq:alfa_it}. Multiplying  both sides of the equation by $\Hbold_D^\top$, we get that the composition $\xbold$ solves
\begin{equation}
\Hbold_D^\top \zbold = \alpha^{-1} \Gbold_D \xbold^\alpha,
\label{eq:numerical_alfainv}
\end{equation}
owing to the fact that $\Hbold_D^\top \Hbold_D = \Gbold_D$. Since $\Gbold_D$ is not invertible, Equation \eqref{eq:numerical_alfainv} cannot be solved directly. Notice that the $\alpha$-transformation in \eqref{eq:alpha_trans_tsagris} can easily be inverted because $\ubold_\alpha(\xbold)$ is a centered vector (thus with $(D-1)$ degrees of freedom), whilst here $\xbold^\alpha$ is unconstrained, with $D$ degrees of freedom. Let us define $Q_{\zbold}(\ybold) = \norm{\Hbold_D^{\top}\zbold - \alpha^{-1} \Gbold_D \ybold^\alpha}$. Then, 
$$\xbold = A_{\alpha-IT}^{-1}(\zbold) = \arg \min_{\ybold\in \mathbb{S}_0^D} Q_{\zbold}(\ybold).$$ 
When $\zbold$ belongs to the codomain of $A_{\alpha-IT}$, the minimum of $Q_{\zbold}$ is 0 and $\xbold \in \bbS^D$. Otherwise, the minimum $Q_{\zbold}$ is larger than 0 and it is achieved on the border of  $\bbS_0^D$, i.e., with at least one part of $\xbold$ being equal to 0. In practice, the minimum is found by using the function {\tt nlminb} in {\tt R} on $\wbold$, with $\ybold(\wbold) = \exp \wbold/ C(\exp \wbold)$, thus guaranteeing the positivity of all components of $\ybold$ and hence of $\xbold$.

The $\alpha$-IT only admits a numerical form for the inverse transform, but, as seen above, this does not hamper its use. If we wanted to have a closed-form of the inverse transformation, but still admit 0-values in the compositions, we should use either the $\alpha$-transformation \eqref{eq:alpha_trans_tsagris} introduced in \citet{tsagris-alpha} (which admits an explicit inverse as mentioned in Section \ref{sec:alpha_trans}) or the ALR Box-Cox transformation $\zbold_{\alr-BC}=[z_{i,\alr-BC}]^\top$, introduced by \citet{barcelo-pawlo}, inspired by the ALR transform, $\alr(\xbold)=\left[\ln {\frac {x_1}{x_D}}, \cdots, \ln {\frac {x_{D-1}}{x_D}}\right]^\top$ \citep{aitchison}, and defined, for each $i= 1,\dots,D-1$, as
\begin{equation}
z_{i,\alr-BC}= \frac{(x_i/x_D)^\alpha - 1}{\alpha}.
\label{eq:bcalr}
\end{equation}
The classical ALR transform is recovered from Equation \eqref{eq:bcalr} as $\alpha \to 0$. 
Although the two proposed Box-Cox-like transformations admit an inverse transform, the $\alpha$-transformation seems harder than necessary, since it introduces the closure operator \eqref{eq:closure} in the direct transform, and the ALR Box-Cox transform has the same issue of the classical ALR transform: it depends on the chosen denominator $x_D$ and cannot deal with 0-values in $x_D$. 

One may readily see similarities between the $\alpha$-CT and the Box-Cox transformation, defined for any vector $\xbold$ in $\mathbb{R}^D$ as $\ubold_{BC} = A_{BC}(\xbold)=\frac{\xbold^{\alpha} - \1bold_D}{\alpha}$. The difference between $\ubold_{BC}$ and $\ubold_{\alpha-CT}$ relies in the constant which is subtracted to $\xbold^{\alpha}$, namely $\1bold_D$ in the case of the Box-Cox and the $\alpha$-mean $m_\alpha(\xbold)\1bold_D$ for the $\alpha$-CT. However, when multiplying $\ubold_{BC}$ and $\ubold_{\alpha-CT}$ by the Helmert matrix $\Hbold_D$, the transformed composition is projected on the hyper-plane orthogonal to vectors colinear to $\1bold_D$. Hence, the Isometric coordinates $\Hbold_D \ubold_{BC}$ and $\Hbold_D \ubold_{\alpha-CT}$ coincide.

In Figure \ref{fig:difference_alfa_it}, we plot the codomains of the $\alpha$-IT (with a `shield' shape) along with the codomains of the $\alpha$-transformation (with a triangular shape), when these transformations are applied on the 3-dimensional simplex $\bbS^3$ for different values of $\alpha$. When $\alpha=0$, the codomains of the two transformations coincide with $\bbR^2$. For the other values of $\alpha$, one can notice that the codomains of the $\alpha$-IT appear `smaller in size' than those of the $\alpha$-transformation. This just reflects on a different scale for the transformed data if using the  $\alpha$-IT or the $\alpha$-transformation, which has no relevant impact on their usability in practice.

\begin{figure}[!htp]
	\centering
	\includegraphics[width=0.48\linewidth]{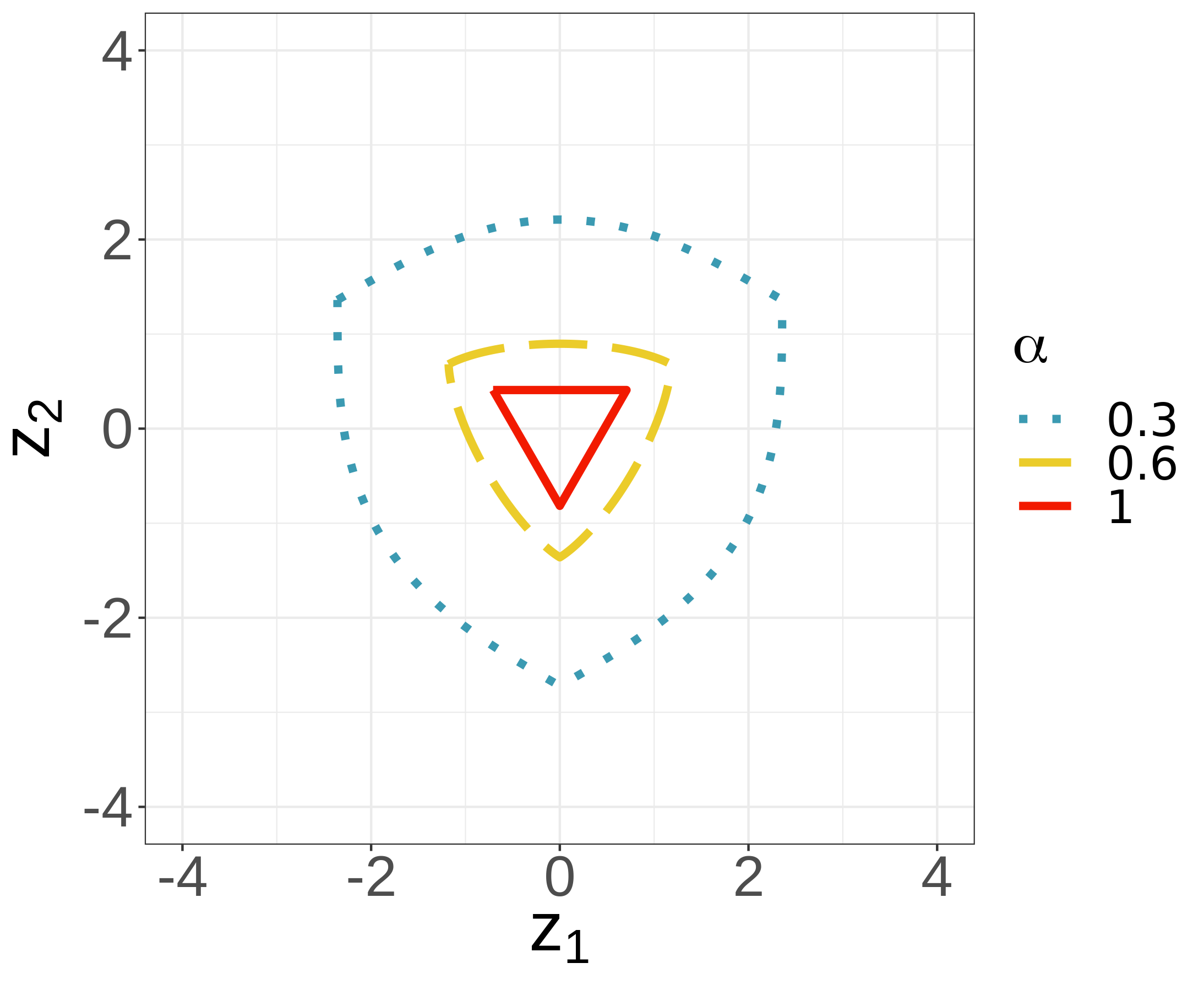}
	\includegraphics[width=0.48\linewidth]{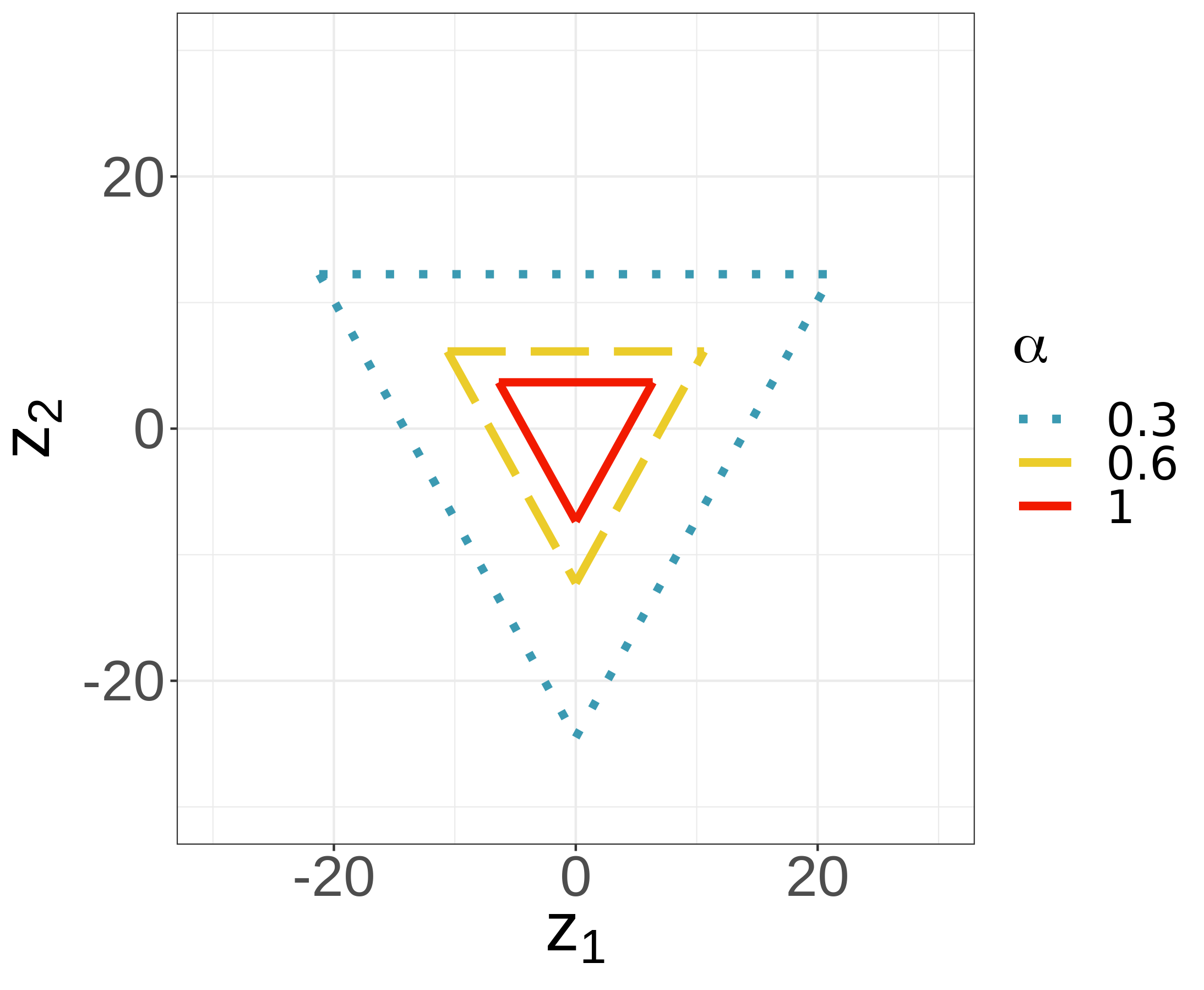}
	\caption{Codomains of the $\alpha$-IT (left) and the $\alpha$-transformation (right) for different values of $\alpha$.}
	\label{fig:difference_alfa_it}
\end{figure}

\subsection{The $\alpha$-IT-metric}

The Isometric $\alpha$-transformation of Equation \eqref{eq:alfa_ct} leads to a natural distance  between observations $\xbold, \ybold \in \bbS_0^D$, denoted by $d_{\alpha-IT}(\xbold, \ybold)$, and hereafter referred to as the \textit{$\alpha$-IT-metric}. For all values of $\alpha > 0$, $d_{\alpha-IT}(\xbold, \ybold)$ is defined as 
\begin{equation}
d_{\alpha-IT}(\xbold, \ybold) = \norm{A_{\alpha-IT}(\xbold) - A_{\alpha-IT}(\ybold)},
\label{eq:alfaitdist}
\end{equation}
where $\norm{\cdot}$ denotes the usual Euclidean distance (in $\bbR^{D-1}$).
When $\alpha\to 0$ we retrieve the Aitchison’s distance \eqref{eq:aitchison_dist}, while when $\alpha=1$ we retrieve the Euclidean distance $d_{1}(\xbold, \ybold) = [ \sum_{i=1}^D \left(x_i-1/D -y_i + 1/D\right)^2]^{1/2} = [\sum_{i=1}^D \left(x_i-y_i\right)^2]^{1/2}$. 

Note that one can define measures of central tendency based on the $\alpha$-IT-metric, using the concept of Fr\'echet mean -- similarly as in \citet{pawlowsky-01} and \citet{tsagris-alpha}. Indeed, a measure of central tendency for a compositional dataset $\xbold_1,...,\xbold_n \in \bbS_0^D$ can be defined as
\begin{equation}
\mathbf{m} = \arg\min_{\boldsymbol{\mu} \in \mathbb{S}^D_0} \left\{ \frac{1}{n} \sum_{i=1}^n d_{\alpha-IT}(\xbold_i, \boldsymbol{\mu}) \right\}.
\end{equation}
This is equivalent to back-transforming the sample mean of the transformed data, that is
\begin{equation}
\mathbf{m} = A_{\alpha-IT}^{-1} \left(\frac{1}{n}\sum_{i=1}^n A_{\alpha-IT}(\xbold_i)\right).
\end{equation}

\subsection{Spatial structure of the Centered and Isometric $\alpha$-transformations}\label{sec:spatial_alpha_IT}

As already mentioned, the classical way of performing a geostatistical analysis is to work on the coordinates of the transformation \citep{principle-coord, pawlo-egozcue}. As we have seen in Section \ref{sec:geostat}, we can directly relate the cross-covariance matrix of the ILR transformed data with the variation matrix of the composition. This section aims at deriving analogous results for the Isometric $\alpha$-transformation.

\begin{defi}\label{def:BC_covariance}
Consider a second-order LR stationary compositional random field $\Xbold(s)$. 
	The $\alpha$-IT \textit{covariance matrix} is the $(D-1,D-1)$ matrix $\Phibold_\alpha(h) = [\phi_{\alpha,ij}(h)]_{i,j=1}^{D-1}$, given by 
	\begin{equation}
	\Phibold_\alpha (h) = \hbox{\rm Cov}\left[\Zbold_{\alpha-IT}(s), \Zbold_{\alpha-IT}(s+h)\right] = \Hbold_D \Xibold_\alpha(h)  \Hbold_D^\top,
	\label{eq:alfa-it-cov-matrix}
	\end{equation}
	where $\Zbold_{\alpha-IT}(s)$ is the Isometric $\alpha$-transformation of $\Xbold(s)$  and $\Xibold_\alpha(h)=[\xi_{\alpha,ij}(h)]_{i,j=1}^{D}$ is the $\alpha$\textit{-CT covariance matrix} whose elements are
	\begin{equation}
	\xi_{\alpha,ij}(h) = \Cov \left[u_{i,\alpha-CT}(s), u_{j,\alpha-CT}(s+h)\right].
	\label{eq:alfa-ct-cov-matrix}
	\end{equation}
	In Equation \eqref{eq:alfa-ct-cov-matrix}, $u_{i,\alpha-CT}$ and $u_{j,\alpha-CT}$ are two components of the vector of the Centered $\alpha$-transformation of $\Xbold(s)$ (Equation \eqref{eq:alfa_ct}).
\end{defi}

\begin{prop}
The CLR covariance matrix $\Xibold(h)=[\xi_{ij}(h)]_{i,j=1}^{D}$ can be retrieved by taking the limit for $\alpha \to 0$ of the $\alpha$-CT covariance matrix $\Xibold_\alpha(h)$ of Equation \eqref{eq:alfa-ct-cov-matrix}.
The ILR covariance matrix $\Phibold(h)=[\phi_{ij}(h)]_{i,j=1}^{D-1}$ can be retrieved by taking the limit for $\alpha \to 0$ of the $\alpha$-IT covariance matrix $\Phibold_\alpha(h)$ of Equation \eqref{eq:alfa-it-cov-matrix}.
\label{prop:it_ct_covariance}
\end{prop}

\begin{proof}
$\Phibold(h)$ (respectively $\Phibold_\alpha(h)$) is the multiplication of the Helmert matrix $\Hbold_D$ and its transpose by the $(D,D)$ matrix $\Xibold(h)$ (respectively $\Xibold_\alpha(h)$). Hence, it is sufficient to verify that the formula of the CLR covariance matrix $\Xibold(h)$ is retrieved by taking the limit of $\Xibold_\alpha(h)$ for $\alpha \to 0$:
\begin{equation*}
\begin{aligned}
\lim_{\alpha \to 0} \xi_{\alpha,ij}(h) &=  \Cov \left[\lim_{\alpha \to 0} u_{i,\alpha-CT}(s), \lim_{\alpha \to 0} u_{j,\alpha-CT}(s+h)\right]\\
&=\Cov \left[\ln X_i(s) - \ln g(\Xbold(s)), \ln X_j(s+h) - \ln g(\Xbold(s+h))\right]\\
&= \Cov \left[\ln \frac{X_i(s)}{g(\Xbold(s))}, \ln \frac{X_j(s+h)}{g(\Xbold(s+h))}\right] = \xi_{ij}(h),
\end{aligned}
\end{equation*}
where we use the fact that $\lim_{\alpha\to 0} h_\alpha(\xbold) = \ln h_0(\xbold) = \ln g(\xbold)$.
Finally, 
\begin{equation*}
\lim_{\alpha \to 0} \phi_{\alpha,ij}(h) = \phi_{ij}(h).
\end{equation*}
\end{proof}

Proposition \ref{prop:it_ct_covariance} shows that $\Xibold_\alpha(h)$ is a generalization of the CLR covariance matrix. An analogous result holds for the $\alpha$-IT covariance matrix $\Phibold_\alpha(h)$ with respect to the ILR covariance matrix $\Phibold(h)$.

When analyzing Equation \eqref{eq:alfa-ct-cov-matrix}, the difference between $u_{i,CT}(s) - u_{j,CT}(s+h)$ that is involved in the computation of the covariance matrix can be interpreted as a contrast; as $\alpha$ tends to 0, this contrast tends to a contrast of logs, i.e., a log-ratio. This observation echoes with the need for working on ratios put forward by \citet{aitchison} and by the compositional community in \citet{pawlo-olea}, \citet{tolosana-boogart} and \citet{tolosana-missing}. One could argue that, more generally, one needs to work on $\alpha$-contrasts as those reported in Equation \eqref{eq:alfa-ct-cov-matrix}. In this way, not only the Aitchison geometry and in particular the ILR transformation leads to a well-defined spatial structure, but also the geometry derived from the Isometric $\alpha$-transformation.

\subsection{Geostatistics with the Isometric $\alpha$-transformation}
\label{sec:geo_alphait}
In the spatial analysis of compositional data, we follow the approach of the compositional community which makes use of the `Principle of working in coordinates' \citep{principle-coord}, but we ground the analysis upon the $\alpha$-IT. Having transformed the data, we apply prediction methods as cokriging within the Euclidean codomain of the $\alpha$-IT. The results  obtained in $\bbR^{D-1}$ 
are finally back-transformed to the simplex through numerical inversion of the $\alpha$-IT (see Section \ref{sec:back-transform}).
We discuss on the choice of the value of $\alpha$ to use in the $\alpha$-IT in Section \ref{sec:ML}.

An important remark is that the restriction of the codomain of the $\alpha$-IT with $\alpha>0$ to a subset of $\bbR^{D-1}$ (see Figure \ref{fig:difference_alfa_it}) does not reflect on serious limitations for its use in practice.
One could argue that sometimes the geostatistical predictions would fall out of the codomain of the $\alpha$-IT, implying the actual impossibility to back-transform the results. Although possible in principle, this would rarely happen in kriging -- and, in fact, it never happened during the numerical experiments carried out in this work. Recall that ordinary kriging is an interpolation technique based on a linear combination of the data with coefficients adding up to 1. When all weights are positive,  the result of kriging is a convex combination of the data -- and it is thus guaranteed to lie within the `shield' domain of the $\alpha$-IT, for all $\alpha$. 
When some weights are negative, kriging could give interpolation results out of the convex hull of the data. This can be seen as a limitation of the approaches based on the $\alpha$-IT (as well as of the $\alpha$-transformation of \citet{tsagris-alpha}), which however can be faced, e.g., by constraining the kriging weights to be positive \citep{cressie}. This point is not investigated further in this paper, as predictions outside the codomain have never been observed.

\subsection{Estimating the parameter $\alpha$}
\label{sec:ML}

Kriging is the Best Linear Unbiased predictor in the $L^2$ sense and it is the overall optimal $L^2$ predictor when the  data being analyzed are Gaussian \citep{cressie}. In order to achieve good prediction performances, the parameter $\alpha$ of the transformation will thus be estimated using maximum likelihood, assuming that the transformed values are independent multivariate Gaussian vectors. We first consider that all compositions have positive parts, i.e., $\xbold_k \in \mathbb{S}^D, k=1,\dots,n$. Compositional vectors with at least one null part will be considered later. Since it is equivalent and more convenient to maximize the log-likelihood, the maximum likelihood estimator is thus  $\hat{\alpha}  =  \arg\max_{\alpha} {\cal L}(\alpha; \xbold_1,\dots,\xbold_n)$, with 
\begin{equation}
 {\cal L}(\alpha; \xbold_1,\dots,\xbold_n)
 =  - \frac{n}{2} \ln |\hat{\boldsymbol{\Sigma}}| - \frac{1}{2} \sum_{k=1}^n (\zbold_k -\hat{\boldsymbol{\mu}})^\top   \hat{\boldsymbol{\Sigma}}^{-1} (\zbold_k -\hat{\boldsymbol{\mu}}) + \sum_{k=1}^n \ln |\boldsymbol{J}(\xbold_k)|,
\label{eq:MLE}
\end{equation}
where $|\boldsymbol{J}(\xbold)|$ is the determinant of the Jacobian of the transformation $A_{\alpha-IT}$ and $\zbold_k = A_{\alpha-IT}(\xbold_k)$ is a $(D-1)$ Gaussian vector vector with expectation $\boldsymbol{\mu}$ and covariance matrix $\boldsymbol{\Sigma}$ -- both being common to all the observations because of stationarity. The Jacobian $\boldsymbol{J}(\xbold)$ is the $(D-1,D-1)$ matrix whose elements are 
\begin{equation}
J_{ij}(\xbold) = \frac{\partial z_i}{\partial x_j} = H_{ij} x_j^{\alpha-1} - H_{iD} (1-x_1 \cdots - x_{D-1})^{\alpha-1} 
= H_{ij} x_j^{\alpha-1} - H_{iD} x_D^{\alpha-1},
\label{eq:jacobian}
\end{equation}
with $1 \leq i,j\leq D-1$. Here, the index $k$ was dropped for convenience of notation.
In \eqref{eq:MLE}, $\hat{\boldsymbol{\mu}}$ and $\hat{\boldsymbol{\Sigma}}$ are the ML estimators of $\boldsymbol{\mu}$ and $\boldsymbol{\Sigma}$ respectively, given the multi-Gaussian vectors $(\zbold_1,\dots,\zbold_n)$.

We now consider the case where some of the components of $\xbold$ are equal to 0, implying that the Jacobian in Equation \eqref{eq:jacobian} is undefined when $\alpha \in [0,1)$. In this case, the $D$ dimensional Gaussian distribution is actually concentrated on a simplex with dimension $D' < D$. The proper way of dealing with this situation is thus to consider that the composition lies in $\mathbb{S}^{D'}$ and that the $\alpha$-IT vector $\zbold'$ is given by $\zbold'=\alpha^{-1}\Hbold_{D'}\xbold'$, where $\xbold'$ is the $D'$ dimensional subvector of $\xbold$ with positive parts only. 

The full log-likelihood to be maximized is then 
\begin{equation}
    {\cal L}_0 = {\cal L} + \sum_{i=1}^D {\cal L}_i + \sum_{i=1}^D \sum_{j=i+1}^D {\cal L}_{i,j} + \cdots.
    \label{eq:likelihood_zeros}
\end{equation}
where ${\cal L}$ is the log-likelihood of the data lying in $\mathbb{S}^D$, given in Equation \eqref{eq:MLE}. ${\cal L}_i$ is the log-likelihood of the $(D-1)$-dimensional data obtained by removing the coordinate $x_i=0$ and such that $x_j>0$, with $j \neq i$.  ${\cal L}_{i,j}$ is the log-likelihood of the $(D-2)$-dimensional data  obtained by removing the coordinates $x_i=x_j=0$ and such that $x_k>0$, with $k \not\in  \{i,j\}$. The sum continues with decreasing dimensions until the lowest possible dimension  $D'=2$, corresponding to edges, with $\mathbb{S}^2 = (0,1)$. If $D=3$, there are 4 terms to be considered in ${\cal L}_0$. When $D=4$, the number of terms increases to 11.  Notice that each log-likelihood in Equation \eqref{eq:likelihood_zeros} necessitates at least $D'$ data to be properly computed since $\hat{\boldsymbol{\Sigma}}$ must be of full rank. If there are less than $D'$ data, the contribution of this sub-simplex is simply ignored.

Note that, if the Gaussian hypothesis on the multivariate random field $\Zbold(s)$ generating the data is in force, the transformed data move away from Gaussianity as $|\alpha-\hat{\alpha}|$ increases, making kriging sub-optimal, in the sense that there exists a non-linear predictor improving on kriging in the $L^2$ sense. In this sense, $\hat \alpha$ is expected to be the value of $\alpha$ optimizing the performance of kriging.

\section{A simulation study}
\label{sec:simulation_study}

The aim of this section is to analyze the Isometric $\alpha$-transformation introduced in Section \ref{sec:alpha_it}, by applying it to simulated spatial compositional data. Here, we will evaluate the performances of this approach in comparison to the classical log-ratio transformations introduced by \citet{aitchison} and \citet{ilr-egozcue}.

\subsection{Simulation process}

We focus here on spatial compositional data in $\bbS^D$, $D=3$, but the same approach could be applied to any dimension $D \in \mathbb{N}$, up to some minor modifications. The spatial domain $\mathcal{D}$ is the square $[0,10]\times[0,10]$ filled with $2000$ random uniform locations $\{s_j: j=1, \dots, 2000\}$. At these locations, bivariate Gaussian values $\{\Zbold(s_j)\}$ are simulated from a parsimonious bivariate Whittle-Matérn Model \citep{gneiting} with covariance $C_{ij}(h)= c_{ij}W_{\nu_{ij}}(h/s_{ij})$. Each function $W_{\nu_{ij}}$ is given by
\begin{equation}
    W_{\nu_{ij}}(r)=\frac{2^{1-\nu_{ij}}}{\Gamma(\nu_{ij})}r^{\nu_{ij}} K_{\nu_{ij}} (r),
\end{equation}
where $\nu_{ij}>0$ and $K_\nu$ is the modified Bessel function of second kind. $B=100$ different realizations from this bivariate Whittle-Matérn model are drawn using the {\tt R} package {\tt RandomFields} with the following parameters: $c_{ii} = 1$,  $c_{ij, i \neq j} = 0.8$, $\nu_{ii} = \nu_{ij, i\neq j} = 0.5$, $s_{ii} = s_{ij, i\neq j} = 1$.

Using shifting and scaling of the bivariate data
according to $\Zbold_{new}(s_j) = \sigma(\Zbold(s_j) - \tilde{\zbold}), \forall j$, where $\sigma$ and $\tilde{\zbold}$ are a scaling and shifting parameter respectively, various scenarios are built from these realizations. The values $\sigma$ and $\tilde{\zbold}$ are chosen such that for a given $\alpha_0$ (hereafter set to $\alpha_0=0.2$ or $\alpha_0=0.6$) the simulated sets lie within the shield-shaped codomain of the $\alpha$-IT. Compositional data are then computed using the inverse $\alpha$-IT presented in Section \ref{sec:back-transform} for each vector $\Zbold_{new}(s_j)$, with $j=1,\dots,2000$. Thanks to the different scaling and shifting scenarios, different patterns of compositional data are created in the simplex.

Figure \ref{fig:scaled_data} shows the shifted and scaled data located in the center, at the border and in the corner of the shield-shaped codomain, along with the corresponding compositional datasets once the data have been back-transformed to the simplex with $\alpha_0$. The specific values of $\alpha_0$, $\sigma$ and $\tilde{\zbold}$ that will be used throughout this work are given in Table \ref{tab:data_pattern}.

\begin{center}
\begin{tabular}{cccccc}

\hline \hline
Pattern & & $\alpha_0=0$ & $\alpha_0=0.2$ & $\alpha_0=0.6$ & $\alpha_0=1$\\
	\hline
	Center & $\tilde{\zbold}$ & $(0,0)$ & $(0,0)$ & $(0,0)$ & $(0,0)$ \\
	 & $\sigma$ & 1 & $0.50$ & $0.15$ & 0.065 \\
	Border & $\tilde{\zbold}$ & $(-2.3,1)$ & $(-2.3,1)$ & $(-2.3,1)$ &  $(-2.3,1)$ \\
	& $\sigma$ & 1 & $0.50$ & $0.15$  & 0.065\\
	Corner &$\tilde{\zbold}$ & $(4,-3)$ & $(4,-3)$ & $(4,-3)$ & $(4,-3)$ \\
	& $\sigma$ & 1 & $0.38$ & $0.11$ & 0.045\\
	\hline 
\end{tabular}
\captionof{table}{Shifting and scaling parameters used to create patterns of data in the center, at the border and in a corner of the simplex.}
\label{tab:data_pattern}
\end{center}

\begin{figure}[!htp]
	\centering
	\includegraphics[width=0.28\linewidth]{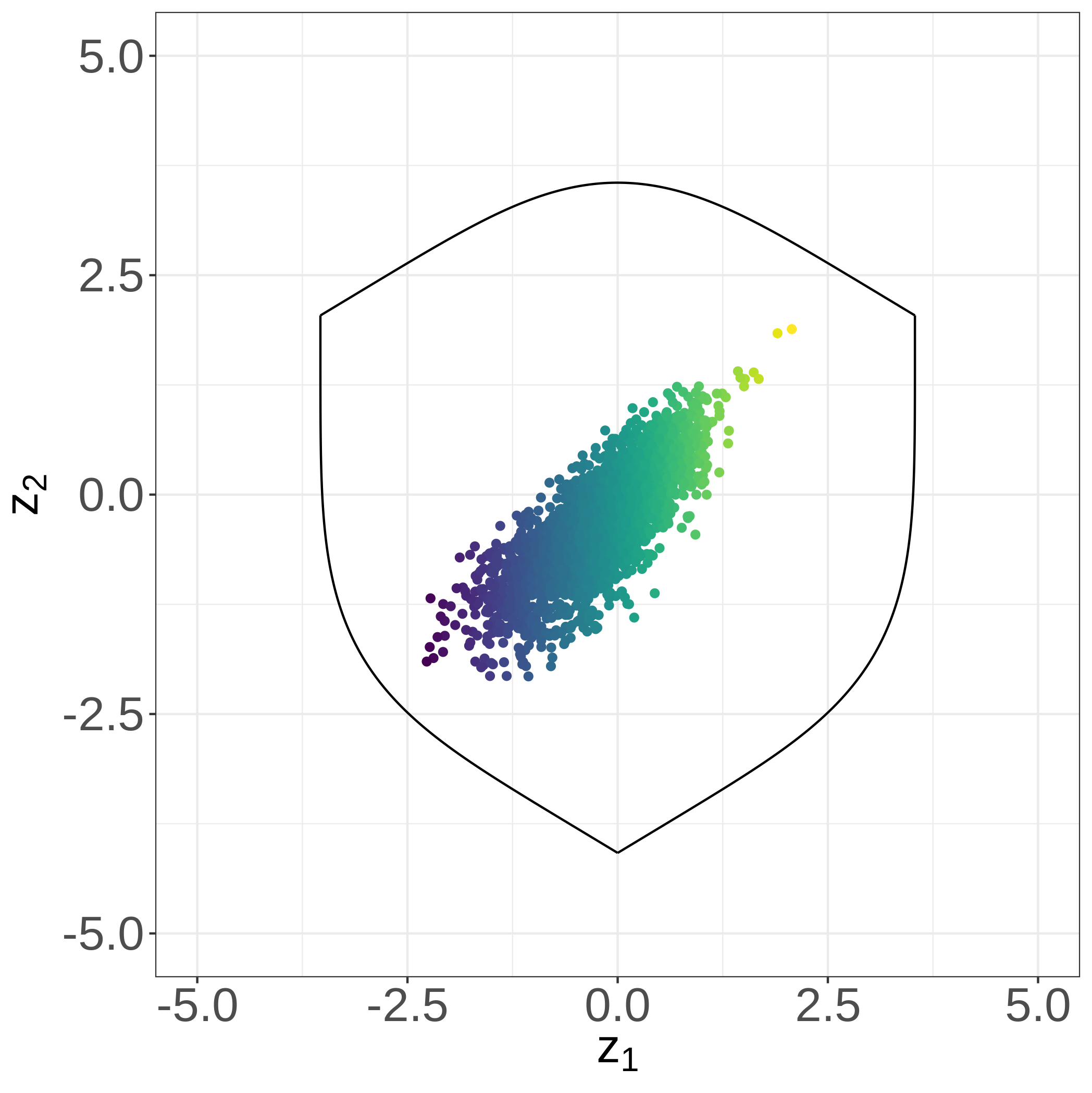}
	\includegraphics[width=0.28\linewidth]{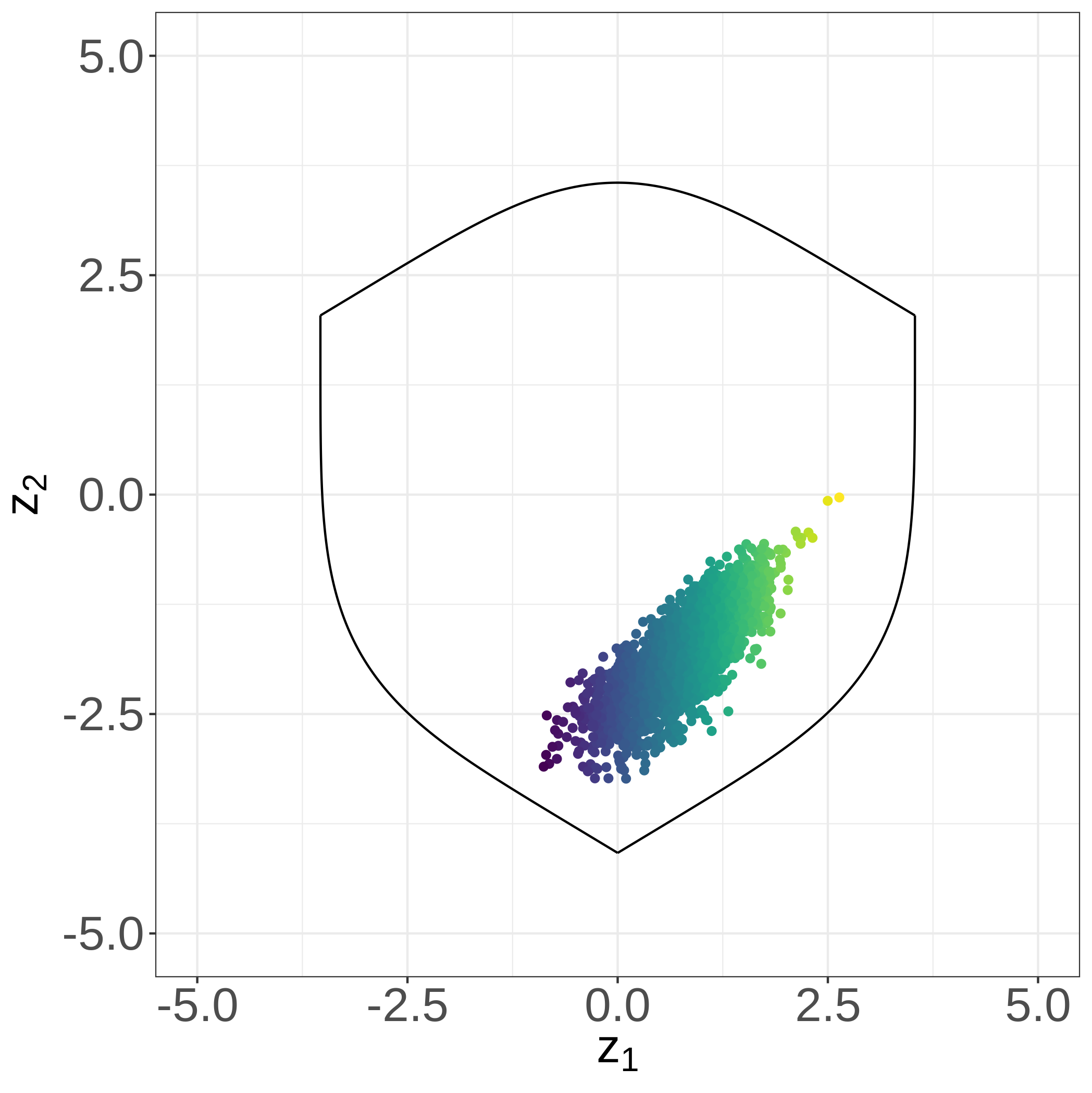}
	\includegraphics[width=0.28\linewidth]{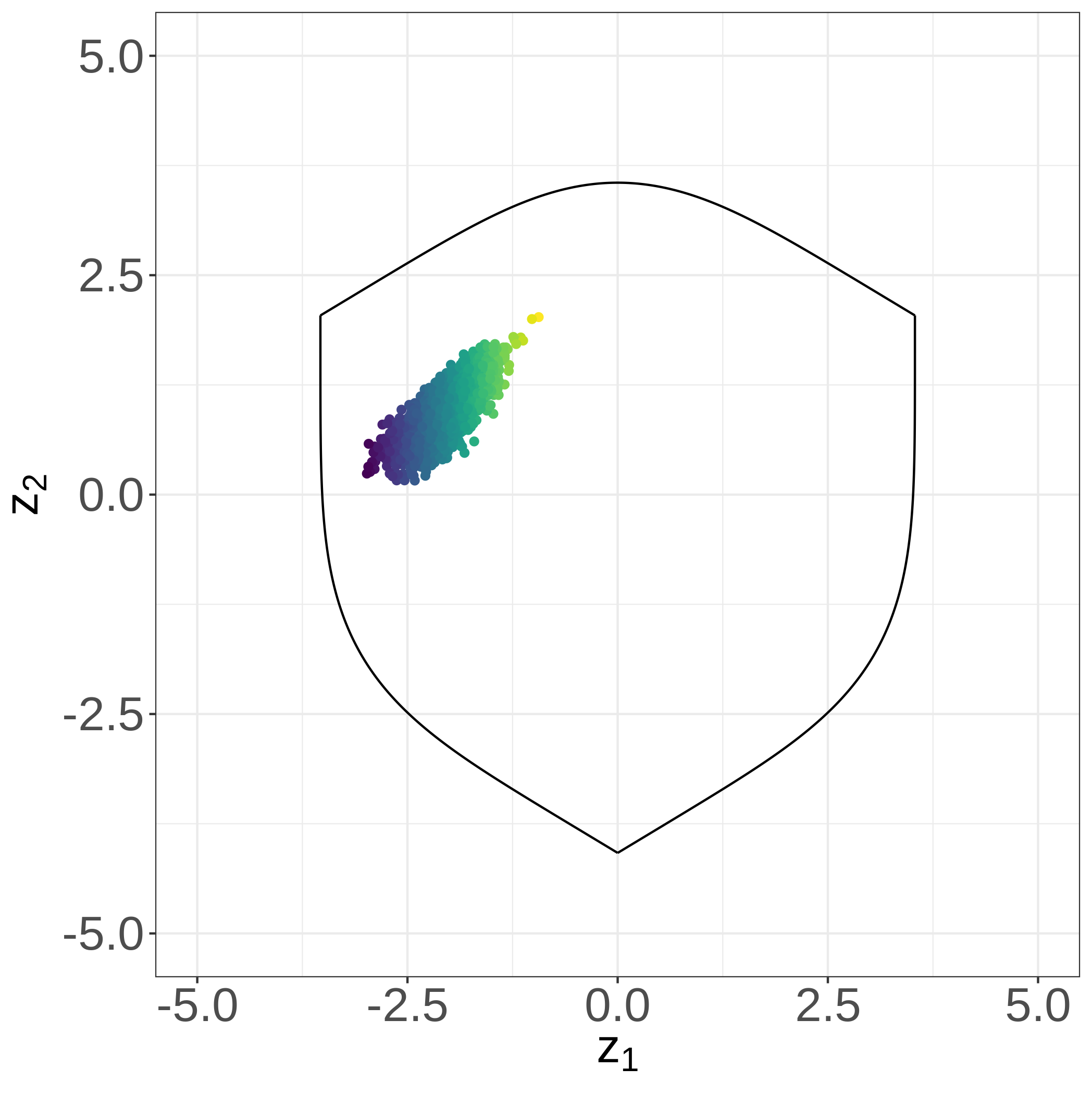}
	\includegraphics[width=0.28\linewidth]{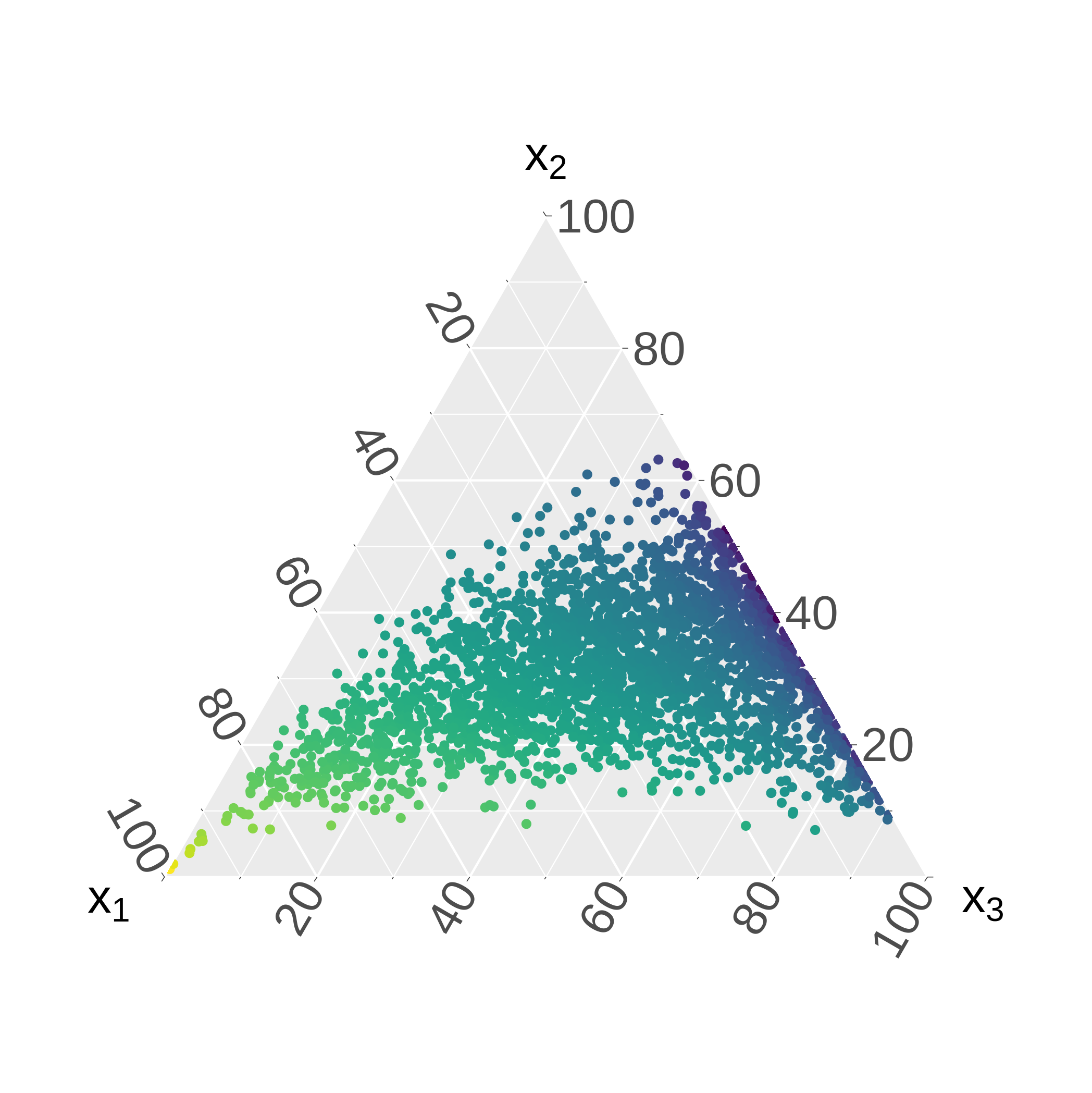}
	\includegraphics[width=0.28\linewidth]{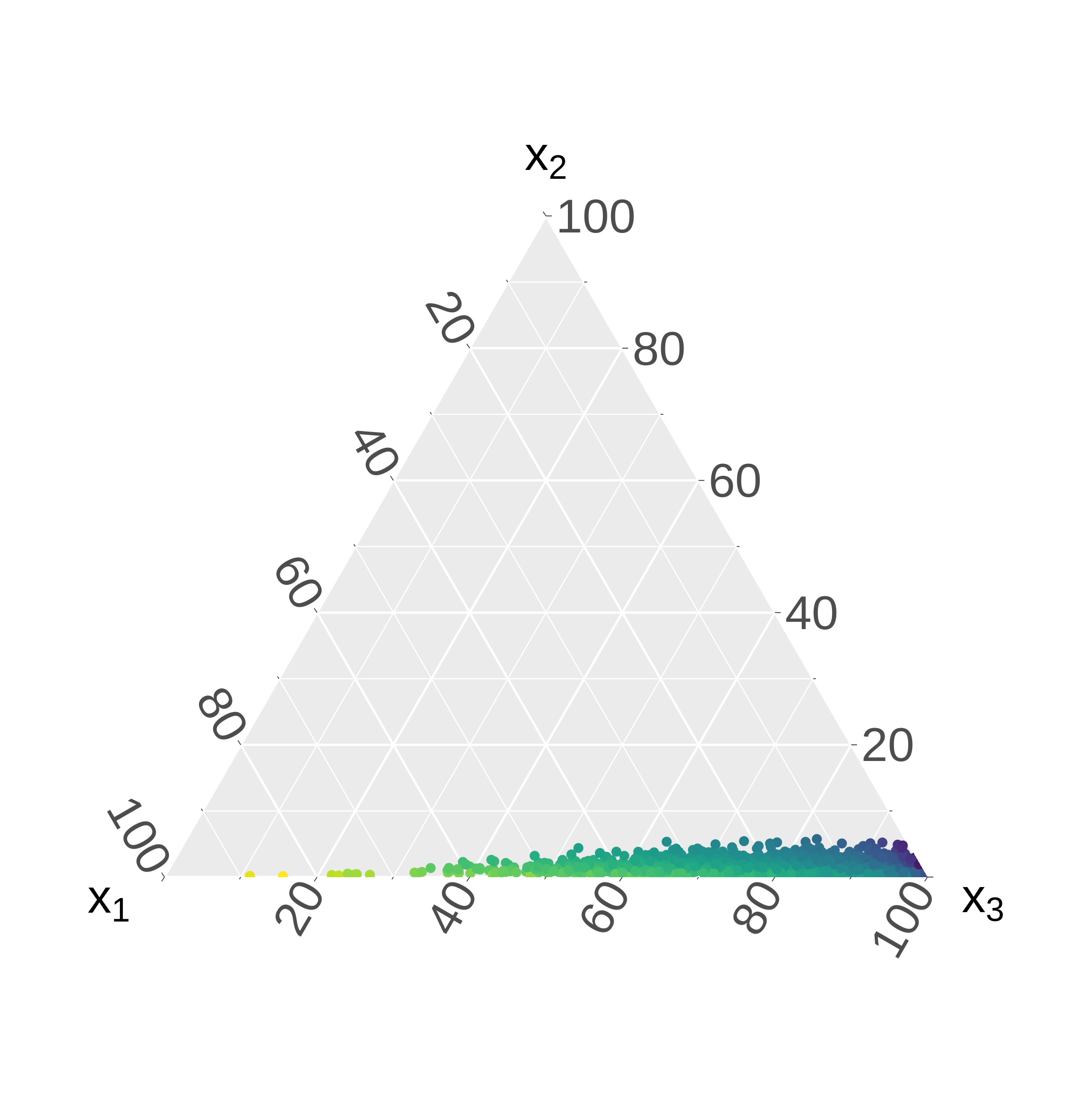}
	\includegraphics[width=0.28\linewidth]{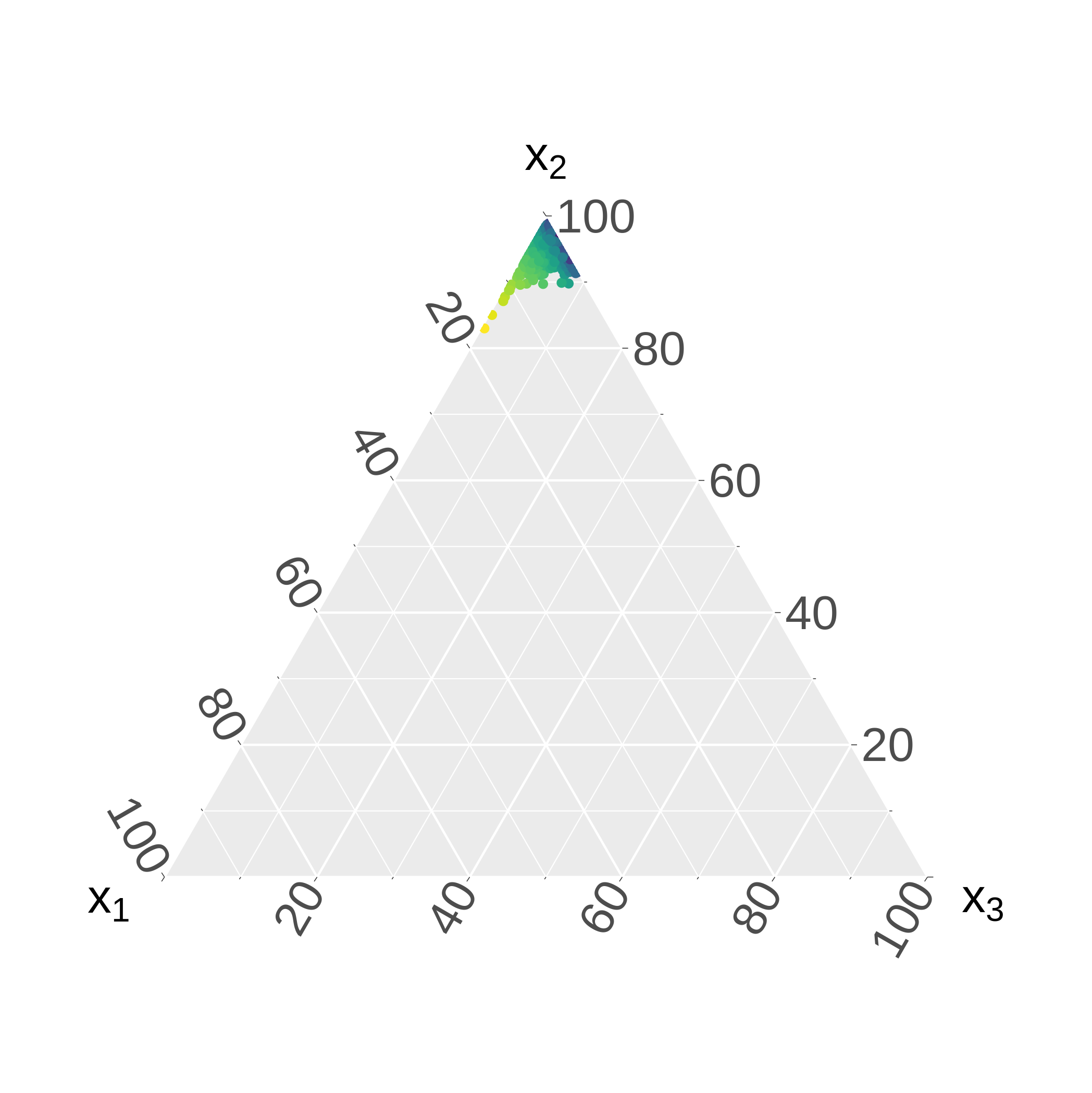}
	\caption{Top: bivariate data after shifting and scaling. The continuous line shows the limit of the codomain for $\alpha_0=0.2$. Bottom: corresponding compositional datasets once transformed with the inverse  $\alpha_0$-IT, as in Equation \ref{eq:numerical_alfainv}. Shades of color follow the $z_1$ axis.}
	\label{fig:scaled_data}
\end{figure}

To evaluate the kriging performances, each of the compositional datasets plotted at the bottom of Figure \ref{fig:scaled_data} is then divided into a training set of $n=200$ data randomly sampled from the realization and a test set made of the remaining $N=1800$ data. A bivariate proportional model is estimated on the training set and kriging is performed to predict at the test locations. Predicted values are then back-transformed into the simplex. 
On each of the $B=100$ realizations, this procedure is applied for a sequence of values of $\alpha$ ranging in $[0,1]$. Note that the locations of the training and test sets are fixed among all the $B$ different realizations.

\subsection{Distances and prediction scores}
\label{sec:distances}

Since a composition in $\bbS^D$ can be seen as a discrete distribution over $\{1,\dots,D\}$, we will use distances between distributions to assess the prediction performances. Note that, even though we here consider scores based on distances between distributions, other measures of discrepancy could be used for the same purpose, such as those based on the concept of divergence, e.g., Kullback-Leibler  divergence \citep{kullback-Leibler} or the $\alpha$-divergence \citep{medak-cressie}. Specifically, we will use the Hellinger distance and the Total Variation distance. These distances allow us to look at the prediction performances of compositional data in an objective way that does not depend upon a specific transformation of the data. These distances are defined between probability measures but for our purposes we here introduce their definition when discrete distributions are concerned. Consider two discrete probability distributions $P = (p_1, \dots, p_D)$ and $Q = (q_1,\dots, q_D)$. The Hellinger distance, $d_H$,  is defined as
\begin{equation}
    d_{H}(P,Q) = \frac{1}{\sqrt{2}} \sqrt{\sum_{i=1}^D(\sqrt{p_i}-\sqrt{q_i})^2},
    \label{eq:hellinger_discrete}
\end{equation}
and the Total Variation distance, $d_{TV}$, is
\begin{equation}
    d_{TV}(P,Q) = \frac{1}{2}\sum_{i=1}^D|p_i-q_i|.
    \label{eq:total_variation_discrete}
\end{equation}
Classical inequalities between $L_1$  and $L_2$ norms provide 
$$d_H^2(P,Q) \leq d_{TV}(P,Q) \leq \sqrt{2} d_H(P,Q).$$

For each simulation $b$, with $b=1,\dots,B$, the prediction scores will be the average of the Hellinger (respectively Total Variation) distances between the kriging prediction $\xbold_j^*$ at location $s_j \in {\cal D}$ and the corresponding true value, $\xbold_j$, with $j=1,\dots,N$. These average distances, denoted $\delta^b_H$ and $\delta^b_{TV}$, are thus 
\begin{equation}
\delta^b_H = \frac{1}{N} \sum_{j=1}^N d_H(\xbold^*_{j,b},   \xbold_{j,b}),\qquad
\delta^b_{TV} = \frac{1}{N}  \sum_{j=1}^N d_{TV}(\xbold^*_{j,b},   \xbold_{j,b}).
\label{eq:mae_metric}
\end{equation}
Finally, we also use the $\alpha$-IT-distance defined in Equation \eqref{eq:alfaitdist} and we compute
\begin{equation}
\delta^{b}_{\alpha} = \left[
\frac{1}{N} \sum_{j=1}^N d_{\alpha-IT}(\xbold^*_{j,b}, \xbold_{j,b})^2 \right]^{1/2}.
\label{eq:rmse_error}
\end{equation}
We recall that the parameter $\alpha$ used to compute this distance is not necessarily equal to the parameter $\alpha_0$.

\subsection{Results}

First, the Maximum Likelihood Estimator (MLE) is computed following the  method detailed in Section \ref{sec:ML}. Table \ref{tab:alfa_profile} reports the averages and standard deviations of the MLE estimates computed from $B=100$ realizations with $n_S = 500$ sample points, in several situations and for 4 different values of $\alpha_0$, including $\alpha_0=0$ which corresponds to an ILR model and $\alpha_0=1$ which corresponds to conditional Gaussian compositions. Overall, the MLE shows good performances. The parameter $\alpha$ tends to be slightly underestimated, at the exception of $\alpha_0=0$. In this case, the true value is at the border of the possible values for $\alpha$, so that estimation errors are one-sided. The standard deviation increases with $\alpha_0$. For $\alpha_0=0.6$ and border data, the average standard deviation have also been computed with increasing numbers of sample points, i.e. $n_S=500, 1000, 2000$. The averages on $B=100$ realizations of independent random vectors were $(0.591,0.592,0.615,0.601)$. The standard deviation of $\hat \alpha$ was equal to $(0.208,0.128,0.082,0.062)$ indicating that it decreases roughly at the $n^{-1/2}$ rate. Similar results have been obtained on further simulations (not reported here). In presence of spatial dependence, the rate of convergence was found to be slower under increasing domain asymptotics. Convergence under infill asymptotics has not been explored, this point being left to further research. An important issue was to check whether MLE is able to estimate correctly $\alpha$ when the data correspond to an ILR transformation ($\alpha=0$) or to a linear transformation ($\alpha=1$). Our findings provide evidence that this is indeed the case.

\begin{table}[hbt]
\begin{center}
\begin{tabular}{ccccc}
\hline \hline
Pattern & $\alpha_0=0$ & $\alpha_0=0.2$ & $\alpha_0=0.6$ & $\alpha_0=1$\\
	\hline
	Center & 0.025 (0.03) & 0.181 (0.09) & 0.563 (0.24) & 0.957 (0.37) \\
	Border & 0.051 (0.07) & 0.190 (0.08) & 0.574 (0.26) & 0.940 (0.41)\\
	Corner & 0.053 (0.07) & 0.178 (0.05) & 0.587 (0.18) & 0.997 (0.36)\\
	\hline 
\end{tabular}
\caption{Average (standard deviation) of $\hat \alpha$ for different configurations of data (as described in Table \ref{tab:data_pattern}). $B=100$ realizations with $n_S=500$ samples.}
\label{tab:alfa_profile}
\end{center}
\end{table}

Starting from the compositional datasets, we then explore the kriging performance for different values of $\alpha$, pretending that the true parameter $\alpha_0$ is not known. Specifically, for a given value $\alpha \in [0,1]$, the $\alpha$-IT is applied to the compositional training set, thereby defining $\zbold_\alpha= (\zbold_{1,\alpha}, \zbold_{2,\alpha})$ in $\bbR^2$. Notice that when $\alpha \neq \alpha_0$, $\zbold_\alpha$ will be different from the original simulated values, which are only retrieved  when $\alpha = \alpha_0$. A Linear Model of Coregionalization (LMC) is then fitted on $\zbold_\alpha$ using the \texttt{R} package \texttt{RGeostats}. Cokriging is then applied at the $N$ locations of the test set. Predicted values are then back-transformed to $\bbS^3$ with the parameter $\alpha$. Finally, the prediction scores, presented in Section \ref{sec:distances}, are computed. This process is repeated for several values of $\alpha$ ranging from 0 to 1. 

We first assess the performance in the (transformed) Euclidean space by computing the average kriging RMSE, which is nothing but the average over the $B$ simulations of the metric $\delta^b_\alpha$  defined in \eqref{eq:rmse_error}. Results are plotted  in Figure \ref{fig:means_rmse} for different values of $\alpha$. Each point represents the average metric $\overline{\delta}_\alpha = \sum_{n=1}^B \delta^b_\alpha / B$, normalized by $\sigma$, which is the empirical standard deviation computed on the transformed vectors $\zbold_\alpha$.  In the left panel, the data are originated from the inverse $\alpha$-IT with $\alpha_0=0.2$, while the right panel represents the data generated with $\alpha_0=0.6$. We note that the minimum of the kriging RMSE is attained around $\alpha=\alpha_0$ in the two cases. As $\alpha$ gets away from $\alpha_0$, the RMSE becomes larger.

\begin{figure}[!htp]
	\centering
	\includegraphics[width=0.48\linewidth]{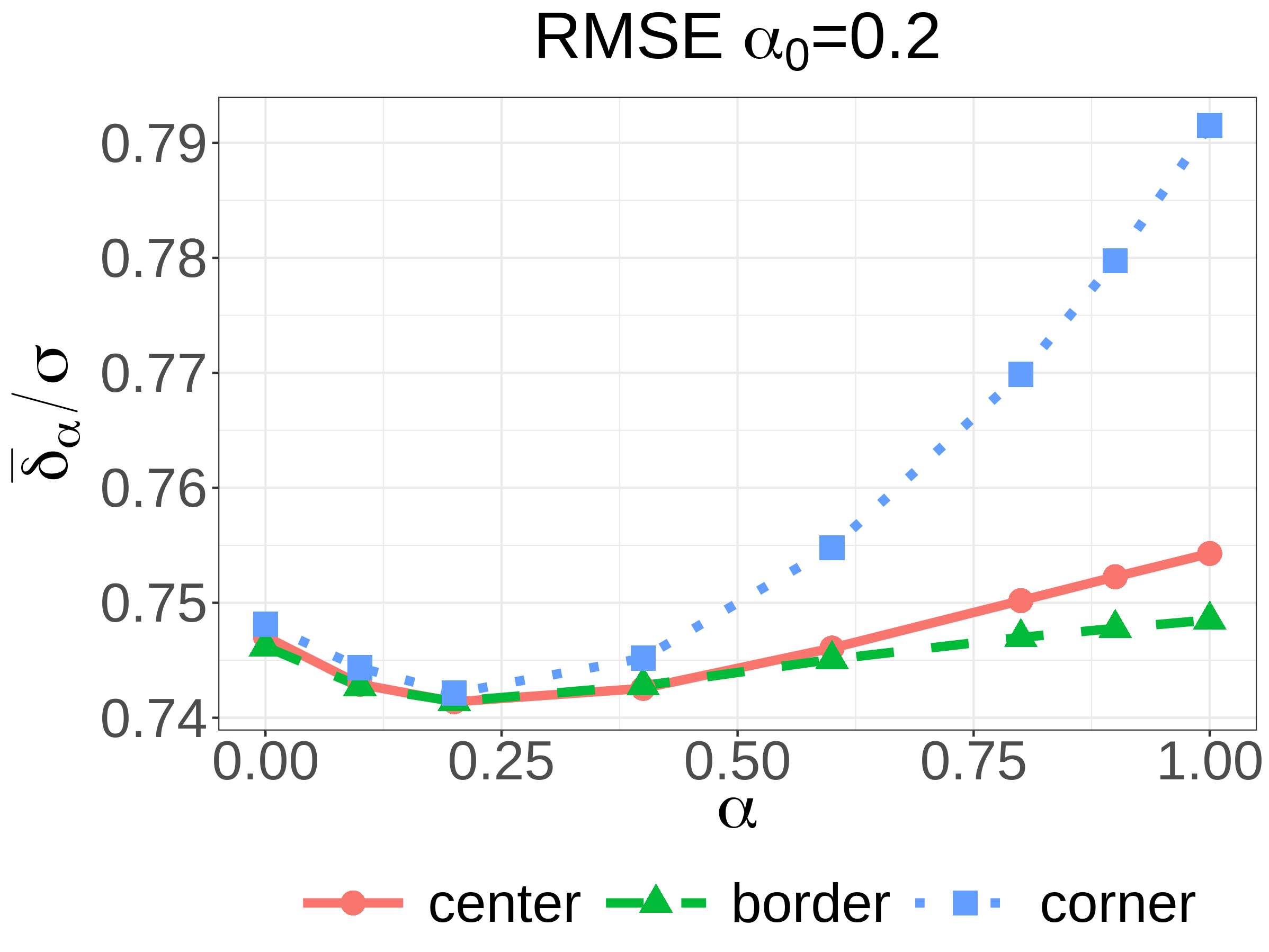}
	\includegraphics[width=0.48\linewidth]{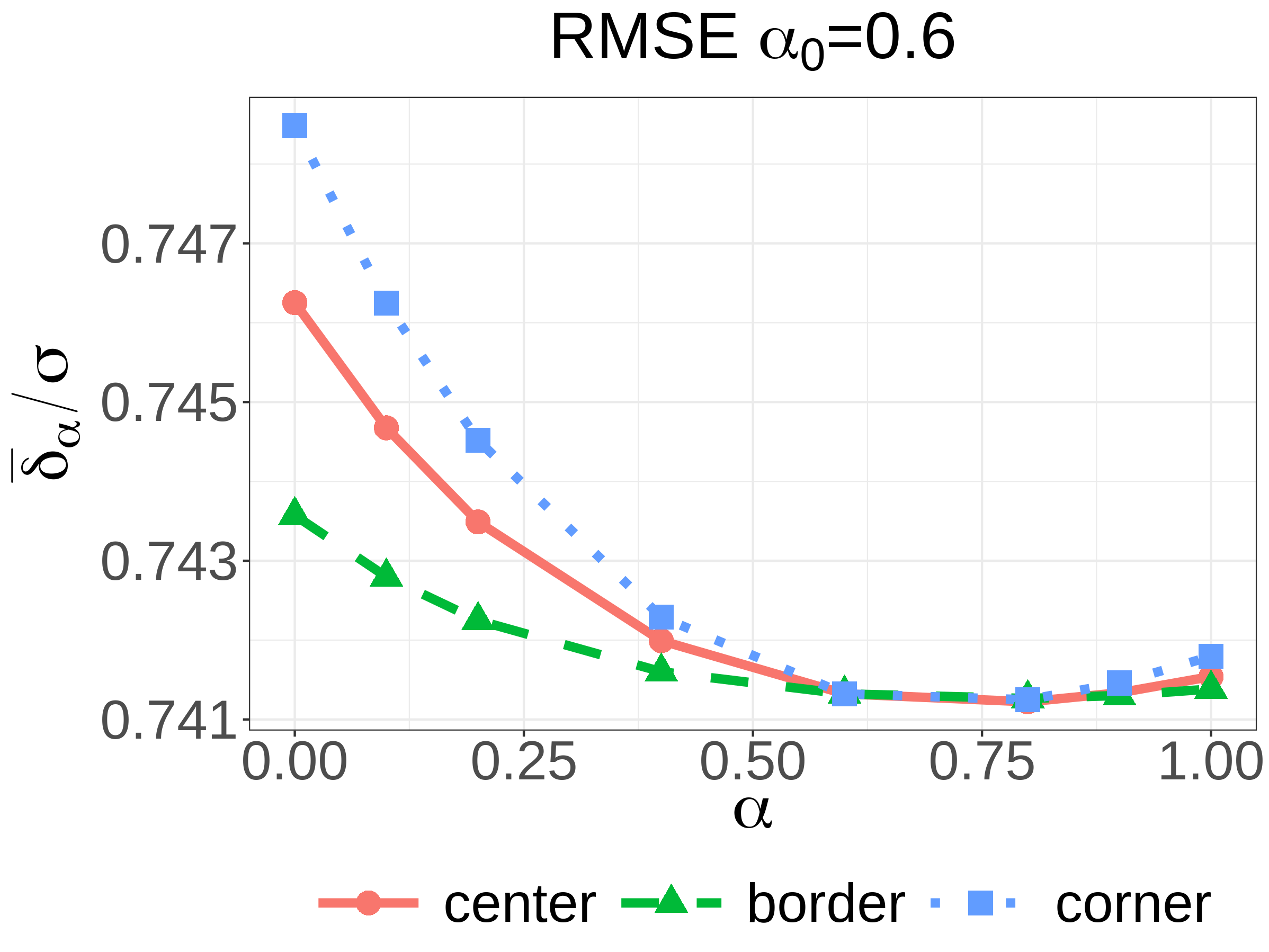}
	\caption{Average kriging RMSE
	$\overline{\delta}_\alpha /\sigma$ over $B=100$ tests for center, border and corner data in $\bbR^2$ before inverse $\alpha$-IT for $\alpha_0=0.2$ (left) and $\alpha_0=0.6$ (right).}
	\label{fig:means_rmse}
\end{figure}

In Figure \ref{fig:means_center} we plot the average of the kriging prediction scores (Total Variation and Hellinger metric defined in Equation \eqref{eq:mae_metric}) computed on the data in the simplex after inverse $\alpha$-IT  over the $B=100$ simulations for the center data with $\alpha_0=0.2$ and $\alpha_0=0.6$. Even though the difference is less marked,  this plot does indicate that a value of $\alpha=\alpha_0$ achieves the best scores. The parameter $\alpha$ resulting in minimal kriging error is thus close to the value $\alpha_0$ used to simulate the data and also close to the ML estimates reported in Table \ref{tab:alfa_profile}. These findings provide good support for the use of ML estimation for estimating the parameter $\alpha$ when analyzing data.

Prediction scores based on Total Variation and Hellinger metrics are less variable with respect to $\alpha$ than the kriging errors computed in the Euclidean space $\bbR^2$.
This is due to the fact that the inverse transform shrinks the data (both observed and kriged) into a smaller area of the simplex causing the absolute error to decrease. 

\begin{figure}[!htp]
	\centering
	\includegraphics[width=0.48\linewidth]{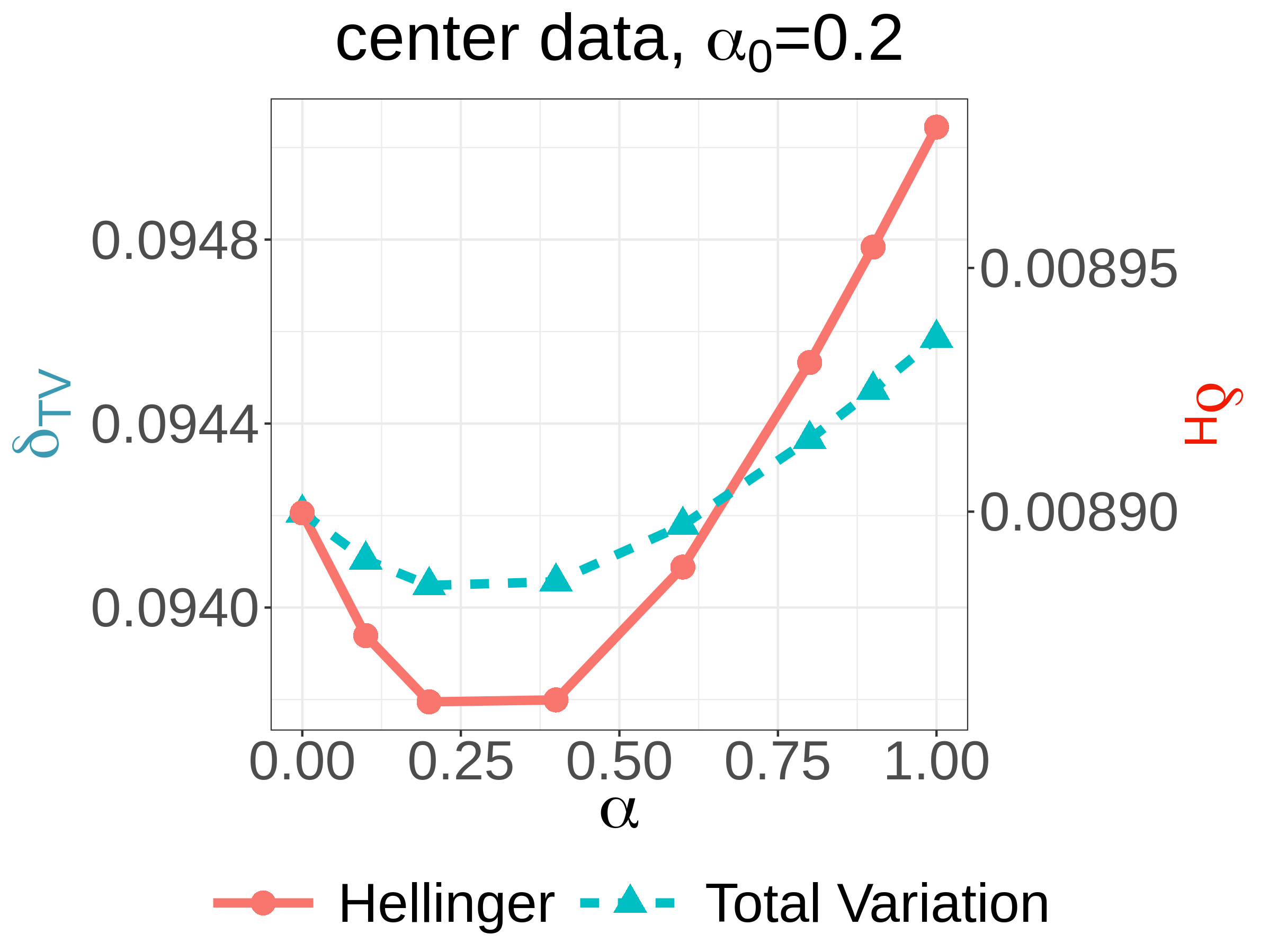}
	\includegraphics[width=0.48\linewidth]{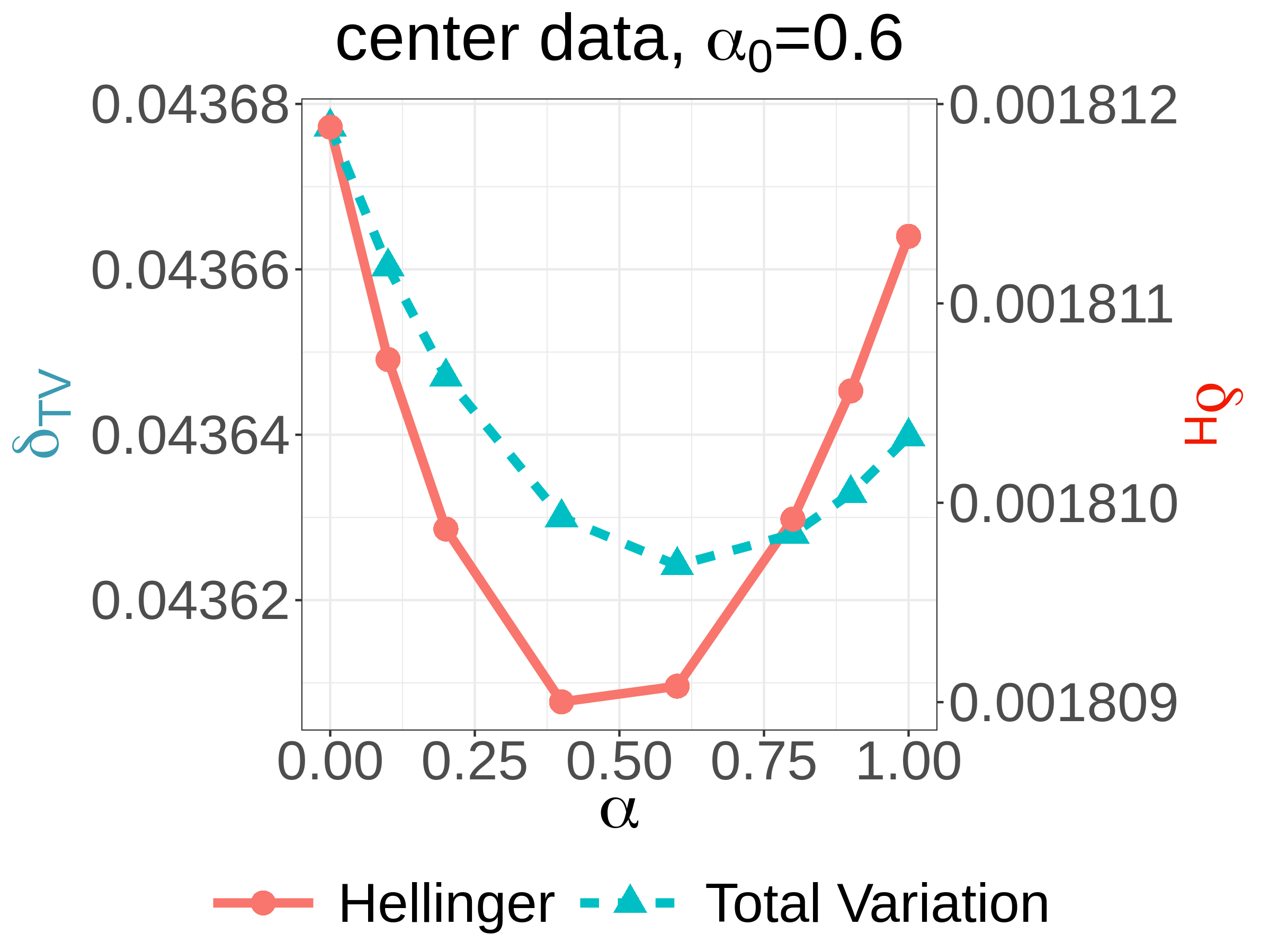}
	\caption{Average error over $B=100$ tests for center data after $\alpha$-IT with different $\alpha$ and $\alpha_0=0.2$ (left) and $\alpha_0=0.6$ (right): Total Variation metric in blue dashed lines and Hellinger metric in red continuous lines.}
	\label{fig:means_center}
\end{figure}

Figure \ref{fig:means_0_1} in the Supplementary Material shows the same result as Figure \ref{fig:means_center} for the extreme cases $\alpha_0=0$ (data simulated from inverse ILR) and $\alpha_0=1$ (data simulated from linear transformation). As expected, the plot shows that the Total Variation metric and the Hellinger metrics recover $\alpha=0$ (respectively $\alpha=1$) not far from the best value of $\alpha$ for kriging.

Figure \ref{fig:means_alphait} shows the normalized averages of the kriging errors of the center data with $\alpha_0=0.2$ and $\alpha_0=0.6$ for different Isometric $\alpha$-metrics. The result shown in the plot is less straightforward to interpret, since the minimal kriging error computed with the normalized $\alpha$-IT metric is attained close to the value $\alpha$ used in the metric instead of being close to $\alpha_0$ as one would expect. This implies that the choice of the metric influences the result of the analysis in a sort of self-confirming way and that it is thus safer to use metrics that are not related to the transformation of the compositional data, as it is the case for the Total Variation or the Hellinger metrics. In particular, this result indicates that one should not use the Aitchison metric, corresponding to $\alpha=0$, in all circumstances, irrespective of data at hand.

\begin{figure}[!htp]
	\centering
	\includegraphics[width=0.48\linewidth]{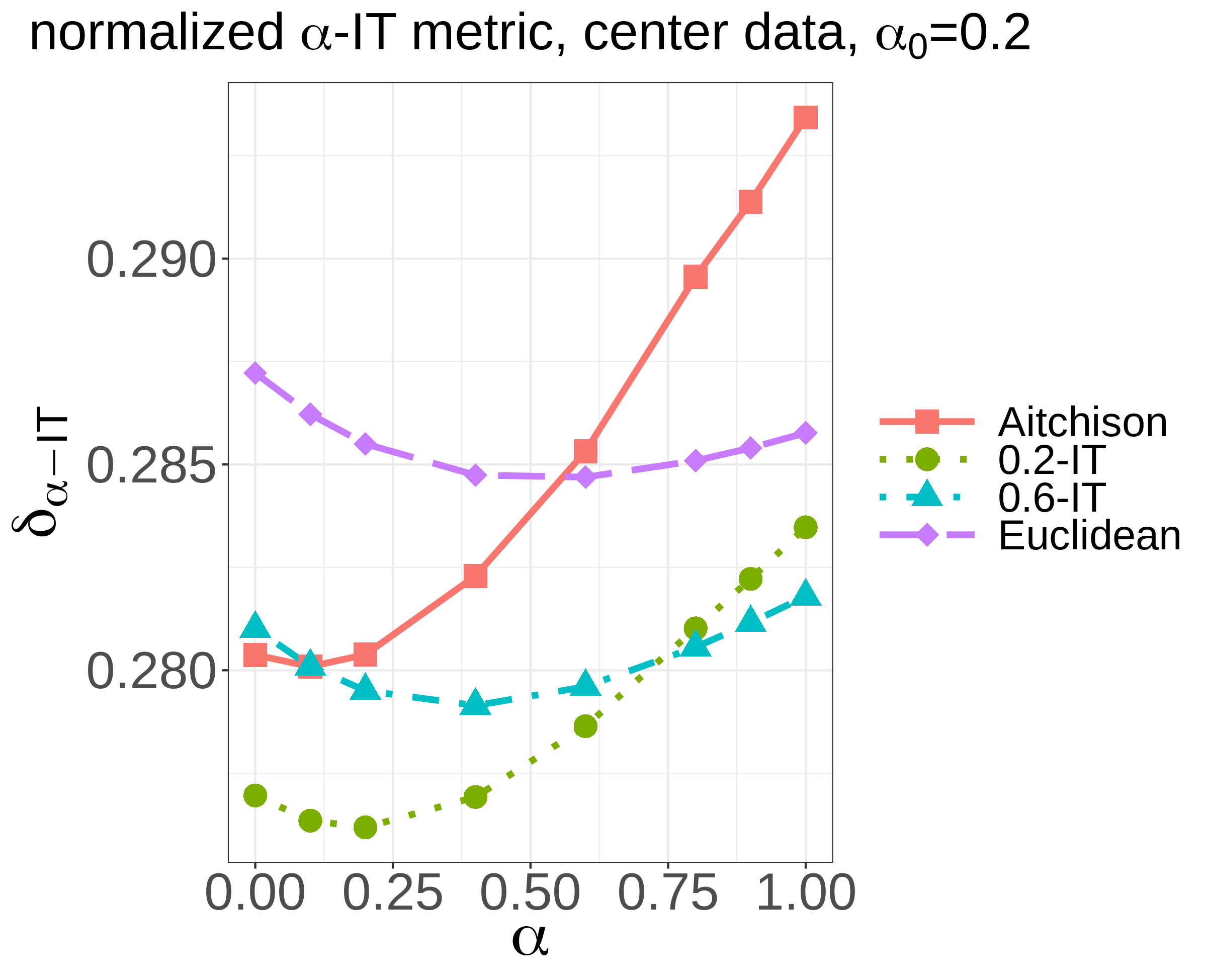}
	\includegraphics[width=0.48\linewidth]{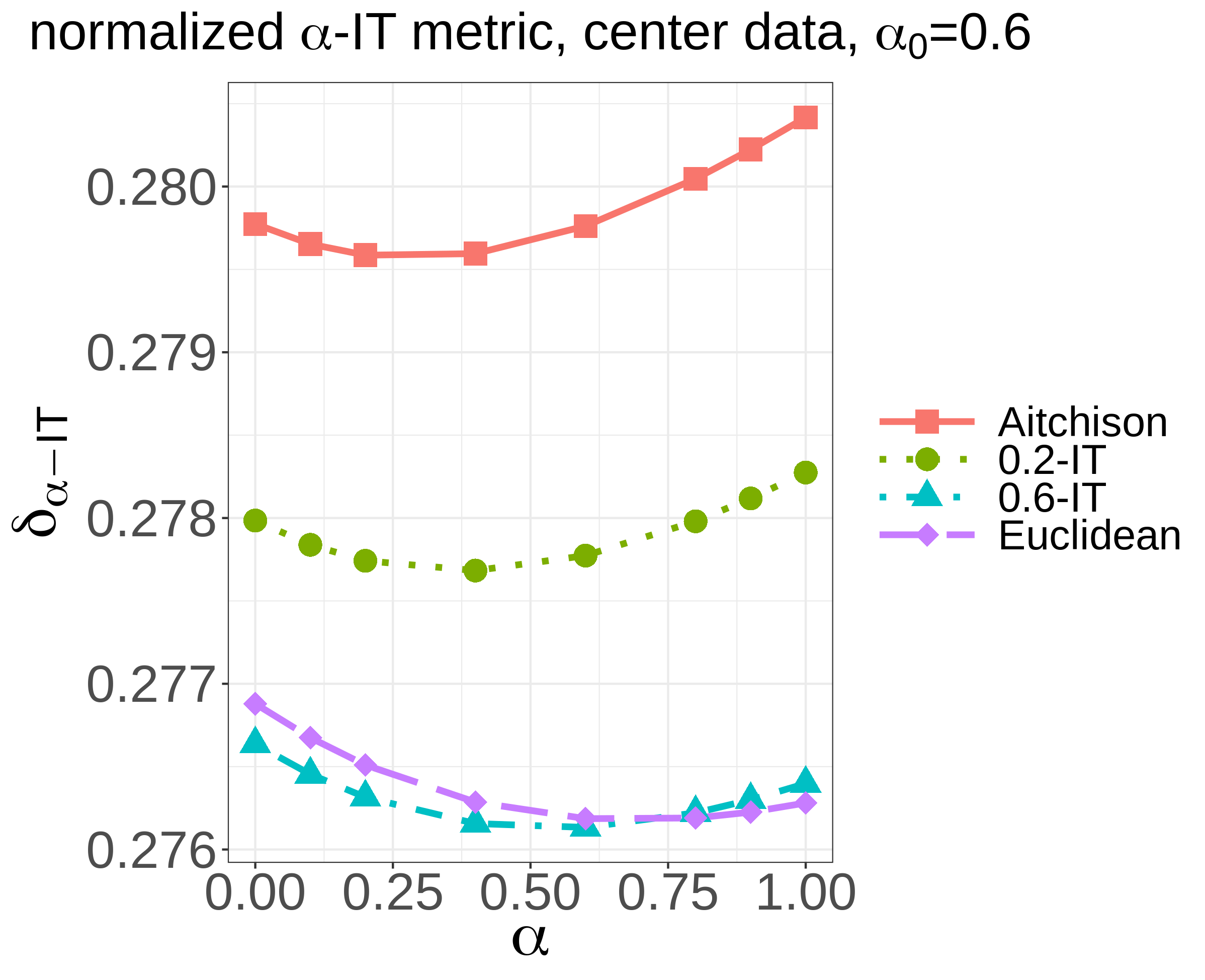}
	\caption{Average of normalized $\alpha$-IT error over $B=100$ tests for center data after $\alpha$-IT with different $\alpha$ and $\alpha_0=0.2$ (left) and $\alpha_0=0.6$ (right).}
	\label{fig:means_alphait}
\end{figure}

Finally, Figure \ref{fig:means_border_corner} in the Supplementary Material shows the same scores of Figure \ref{fig:means_center}, but for border and corner data. We remark that the evolutions of the errors are very similar to the center case. We note that in none of the plots (center, border or corner) the value $\alpha=0$, which corresponds to the classical ILR transform, leads to the lowest error in kriging.


\section{Application to the Copernicus Land Cover Map}
\label{sec:copernicus}

A geostatistical analysis of a spatial compositional dataset is now conducted following the lines presented above: cokriging of the transformed data by $\alpha$-IT followed by the back-transformation into the simplex. Using Monte Carlo cross-validation, we shall assess the performance of our approach for several values of $\alpha$, including the value $\alpha^*$ maximizing the likelihood \eqref{eq:likelihood_zeros}, even though the transformed variables are not necessarily multivariate Gaussian.

\subsection{The dataset}
The Copernicus Dynamic Land Cover Map (CGLS-LC100) \citep{copernicus} delivers a global land cover map at 100 m spatial resolution, offering a primary land cover scheme. The different proportions of covers sum to 1 for every pixel, since each of them represents the proportion of the pixel covered by that given land cover. These proportions can thus be considered as compositional data since they are positive and sum to 1. Scientists in spatial analysis of compositional data have been using land cover datasets to apply prediction techniques; recent works can be found in \citet{lungarska18}, \citet{nguyen21} and \citet{thomas-agnan21}.

For our study, we consider only a part of the Copernicus dataset that corresponds to the Po' Valley in Italy, a major geographical feature of Northern Italy. It extends approximately 650 km in an E-W direction, with an area of $46,000 \text{km}^2$. As standard practice in geostatistics, latitude/longitude coordinates are transformed into the corresponding UTM coordinate system. In a set of 2000 pixels, randomly picked in the spatial domain, we select the 4 main parts of the compositional samples, i.e., those which explain the highest proportion of the composition, and create a new compositional dataset with 4 parts, \textit{crops, shrub, grass} and \textit{tree}, belonging to the simplex $\bbS_0^4$. Since some land covers are not found in each pixel, some of the parts are equal to 0, which makes this dataset an interesting case study for assessing how performances vary with the amount of data containing parts equal to 0.

The predictive performances will be measured using a Monte Carlo cross-validation approach, in which the following procedure is repeated $B=20$ times on the set of 2000 pixels: for each $b=1,\dots,20$, a random sample of size 500 is selected to be the training set and the remaining 1500 compositions form the validation set; an LMC is fitted to the $\alpha$-IT transformed data and cokriging is performed for prediction at the validation locations. Scores are computed using the Total Variation and the Hellinger metrics introduced in Equation \eqref{eq:mae_metric}.

\subsection{Spatial analysis with positive parts}
\label{sec:spatial_no0s}

We first  consider data belonging to $\bbS^4$, i.e., data with positive values for all parts.  With this setting, a fair comparison is possible with the classical ILR transform. In Figure \ref{fig:comp_data_barplot} (right), the set of 2000 data is represented in stacked barplots, where each bar represents a compositional sample. The same data are mapped in Figure \ref{fig:comp_data}, where each of the four components varies between 0 and 1. Visual inspection confirms that the area is quite homogeneous, and that stationarity can reasonably be assumed for the purpose of this study.

The log-likelihood \eqref{eq:likelihood_zeros} is maximized with respect to $\alpha$ for each of the $B$ training samples, thereby providing $B$ estimates of $\alpha^\star$. The average of these estimates is equal to $0.120$ and their standard deviation is $0.048$, showing a small variability around the average value. It is noticeable that $\alpha=0$ lies outside the 2-standard deviations confidence interval, a fact that provides strong evidence that setting $\alpha=0$ for an ILR transform is not supported by the data. For the rest of this Section, we shall thus consider that $\alpha^*=0.12$ corresponds to the ML estimates of the $\alpha$-IT transform.

\begin{figure}[!htp]
	\centering
    \includegraphics[valign=c,width=0.4\linewidth]{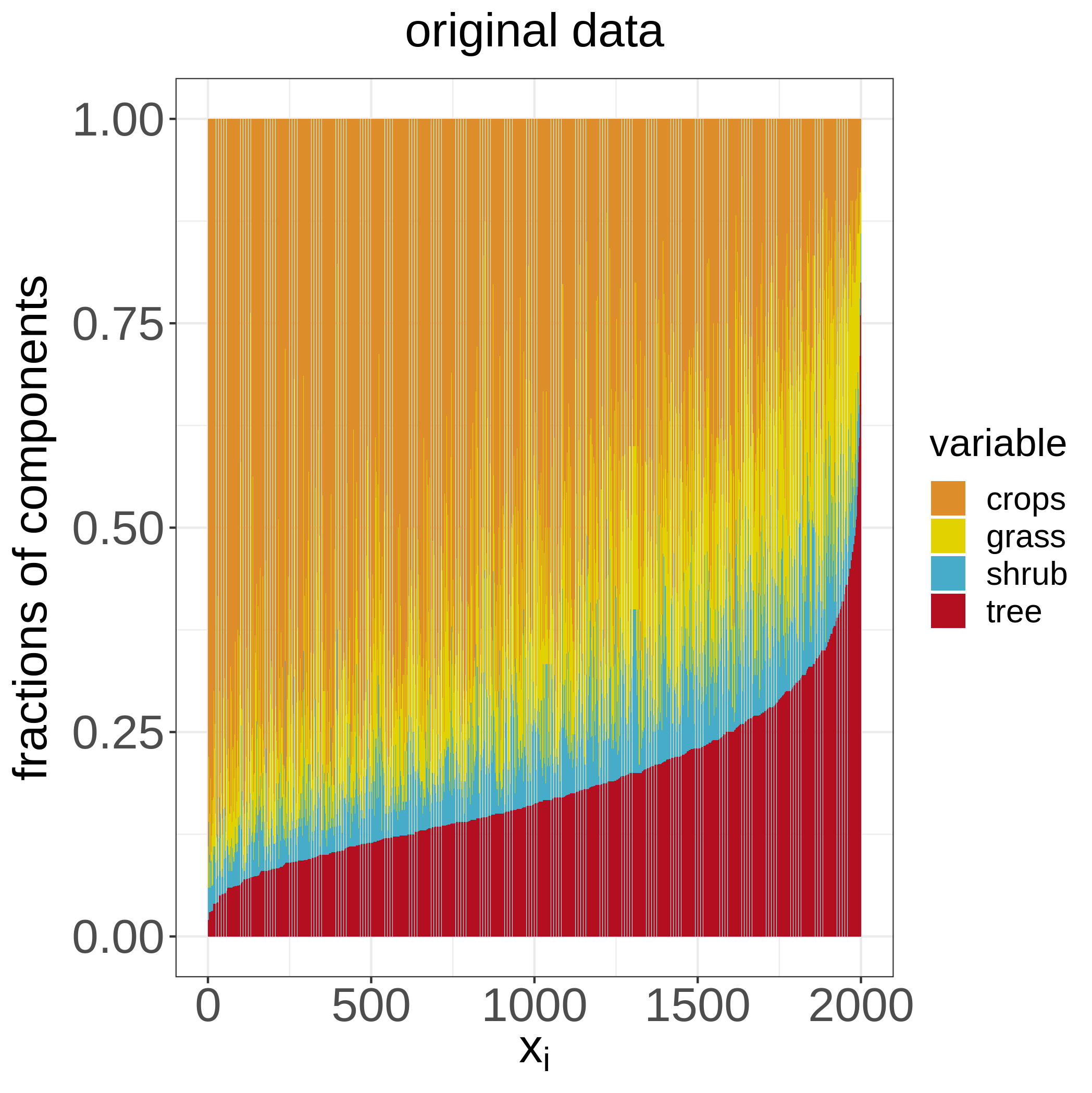}
	\caption{Stacked barplots for the compositional set of size 2000 excluding compositions with 0s.}
	\label{fig:comp_data_barplot}
\end{figure}

\begin{figure}[!htp]
	\centering
	\includegraphics[width=0.48\linewidth]{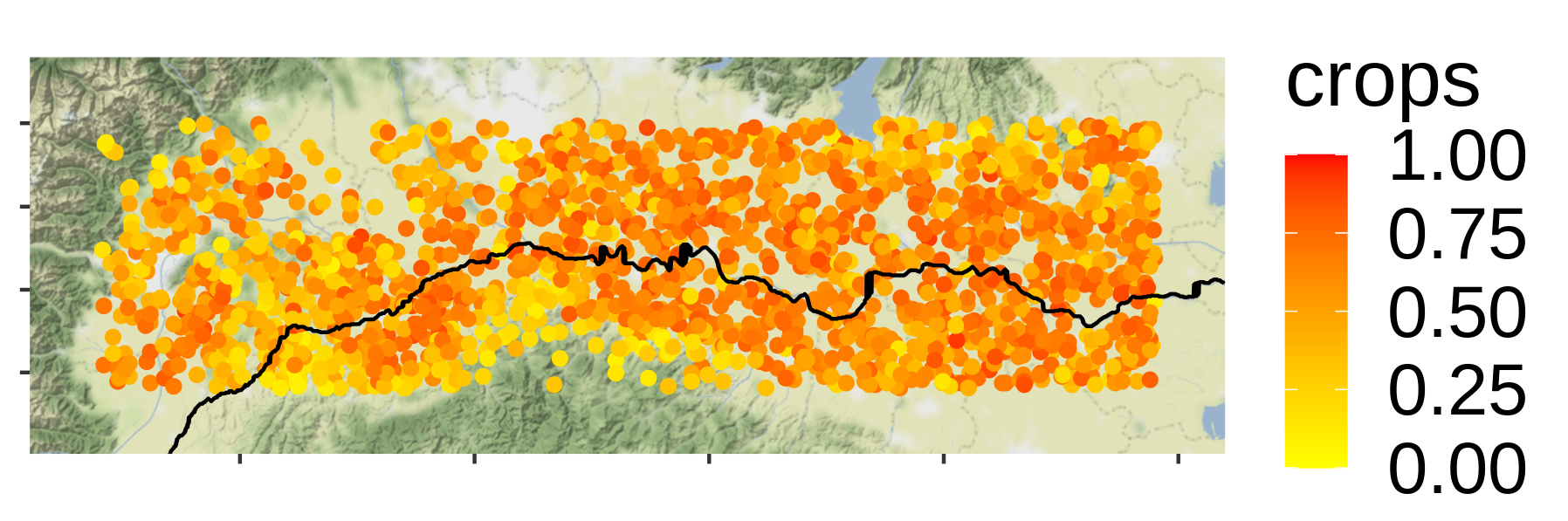}
    \includegraphics[width=0.48\linewidth]{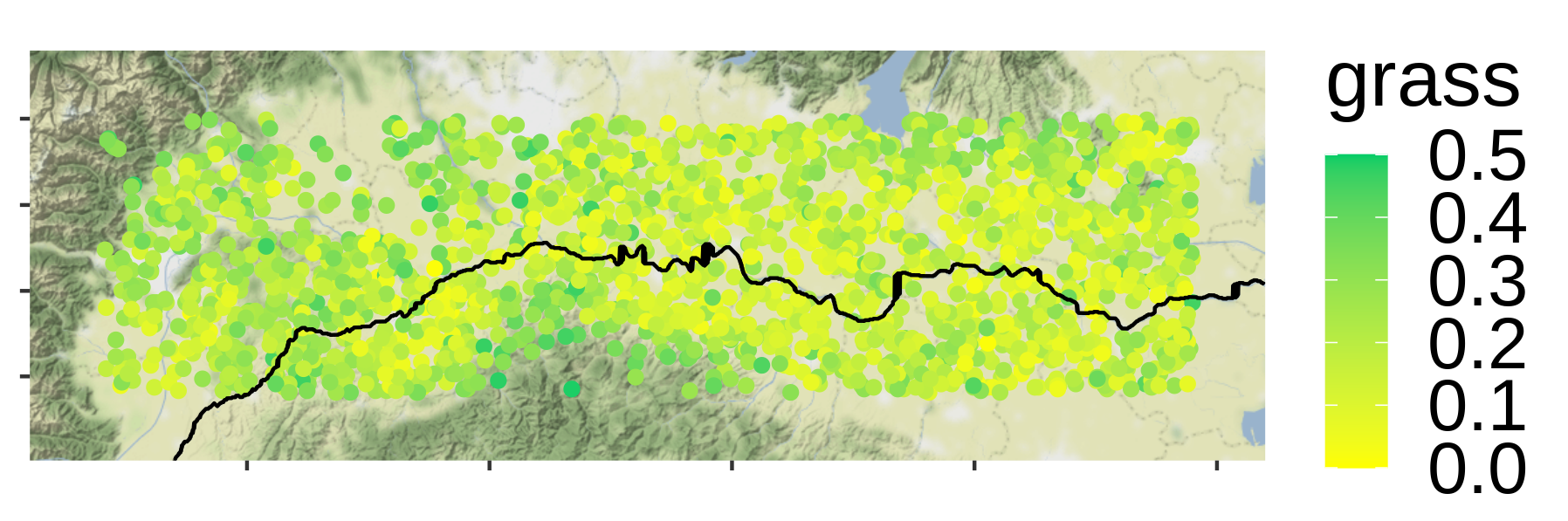}\\
    \includegraphics[width=0.48\linewidth]{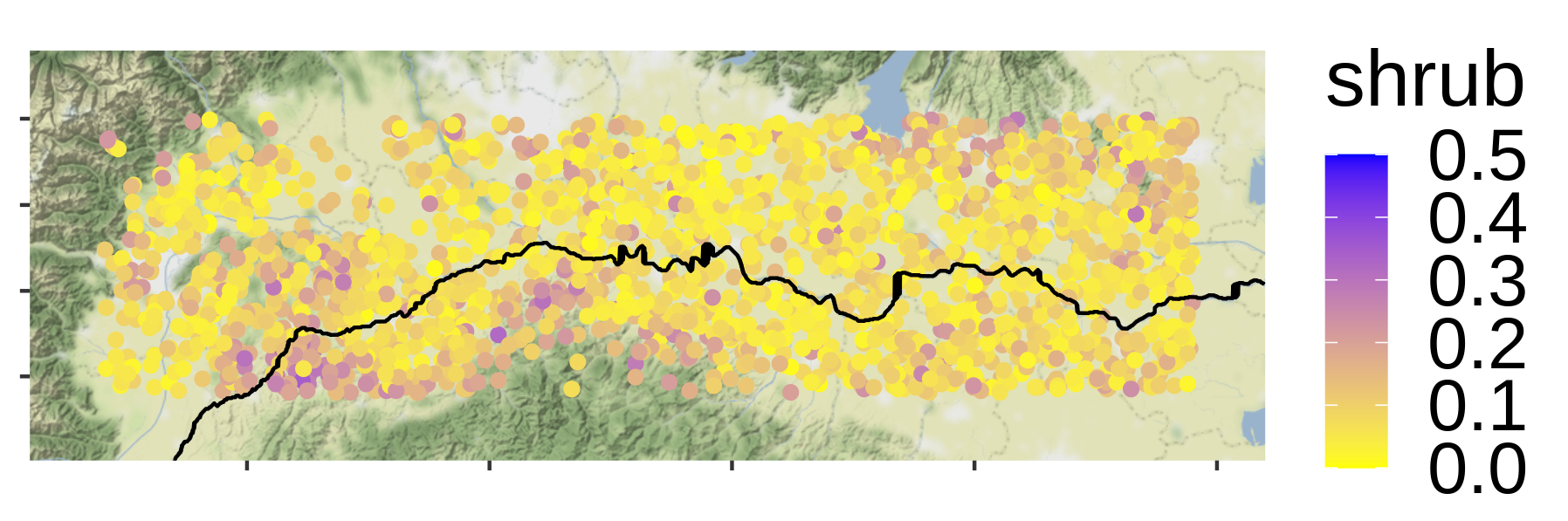}
    \includegraphics[width=0.48\linewidth]{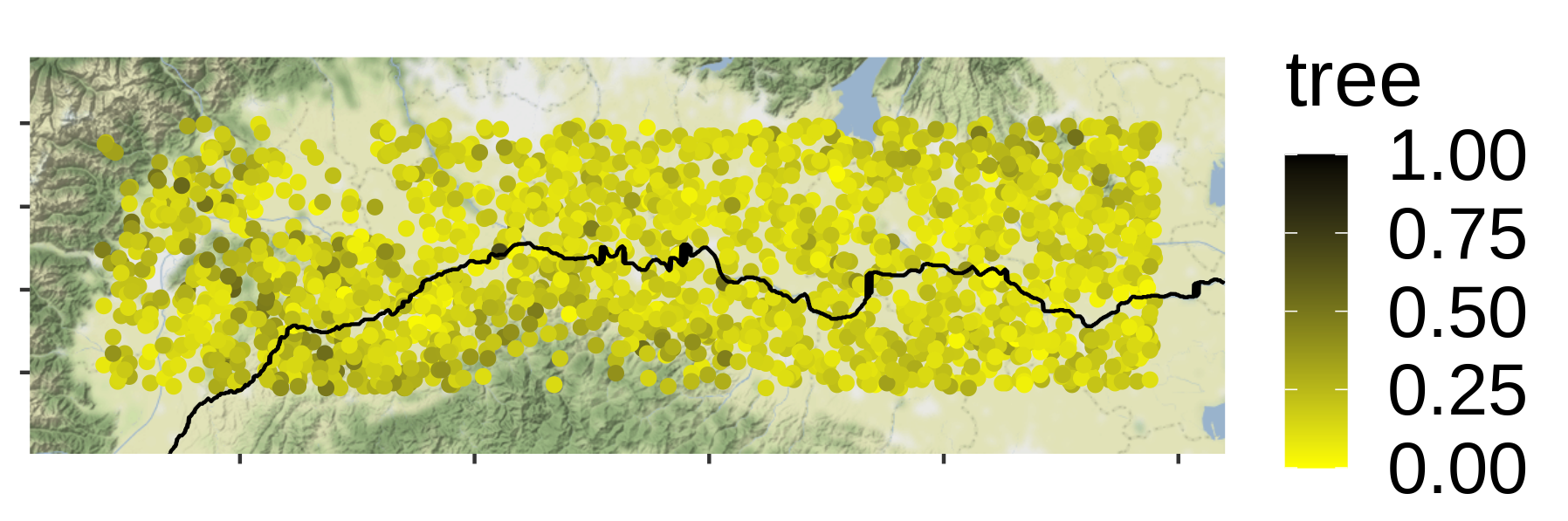}
	\caption{Compositional set of size 2000 excluding compositions with 0s.}
	\label{fig:comp_data}
\end{figure}

Figure \ref{fig:euc_data_no_zero} displays the scatterplots of the 2000 transformed data after the $\alpha$-IT with $\alpha=0$, $\alpha^\star=0.12$ and $\alpha=1$. 
When $\alpha=1$ the triangular shapes reflect that a linear transformation has been applied on the data originated from the simplex $\bbS^4$. When $\alpha=0$, one can notice that groups of data are slightly separated from the main cluster of data. We shall return to this in Section \ref{sec:spatial_with0}.

\begin{figure}[!htp]
	\centering
	\includegraphics[width=0.32\linewidth]{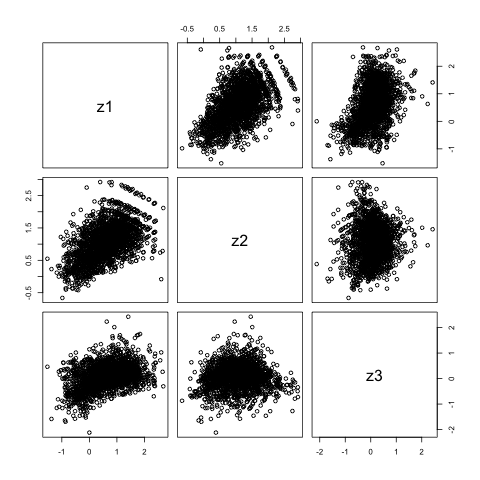}
	\includegraphics[width=0.32\linewidth]{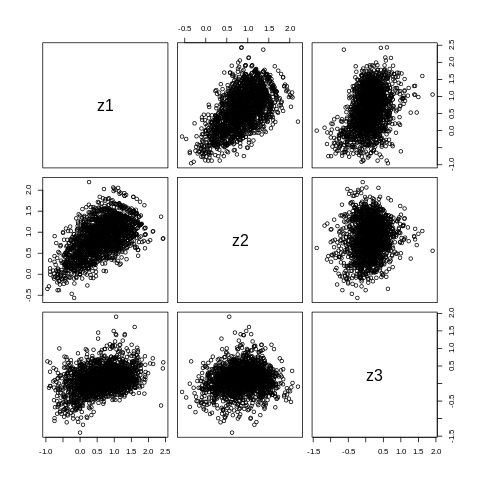}
	\includegraphics[width=0.32\linewidth]{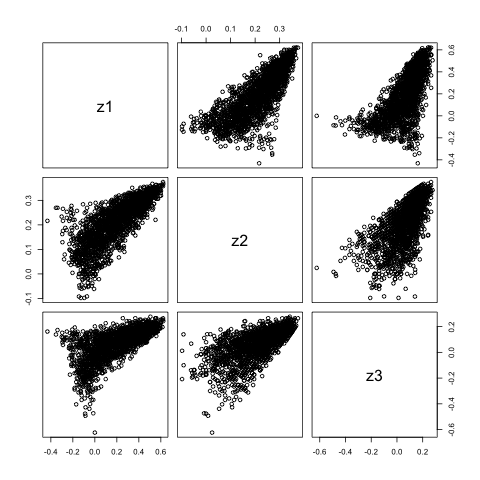}
	\caption{Scatterplot of $z_1$, $z_2$ and $z_3$ after Isometric $\alpha$-transformation with $\alpha=0$ (left),  $\alpha^\star= 0.12$ (center) and $\alpha= 1$ (right).}
	\label{fig:euc_data_no_zero}
\end{figure}

Once the data have been transformed to the Euclidean space $\bbR^3$ for each value of $\alpha \in A=\{0, \alpha^\star, 0.3, 0.5, 0.75, 1\}$, the empirical (co)variograms are computed and an LMC is fitted to the transformed data. Figure \ref{fig:cross_variogram_po_no_zero} shows the empirical and fitted (co-)variograms for $\alpha^\star=0.12$. One can observe the presence of a significant nugget effect -- larger than the partial sill in every direct or cross-variogram -- entailing a possible effect on the quality of kriging predictions. Similar fit were obtained for all values of $\alpha$. The shapes were identical, but the sill decreased as $\alpha$ increases, owing to the associated reduction of the codomain in $\bbR^3$ (see Figure \ref{fig:difference_alfa_it}).

\begin{figure}[!htp]
	\centering
	\includegraphics[width=0.32\linewidth]{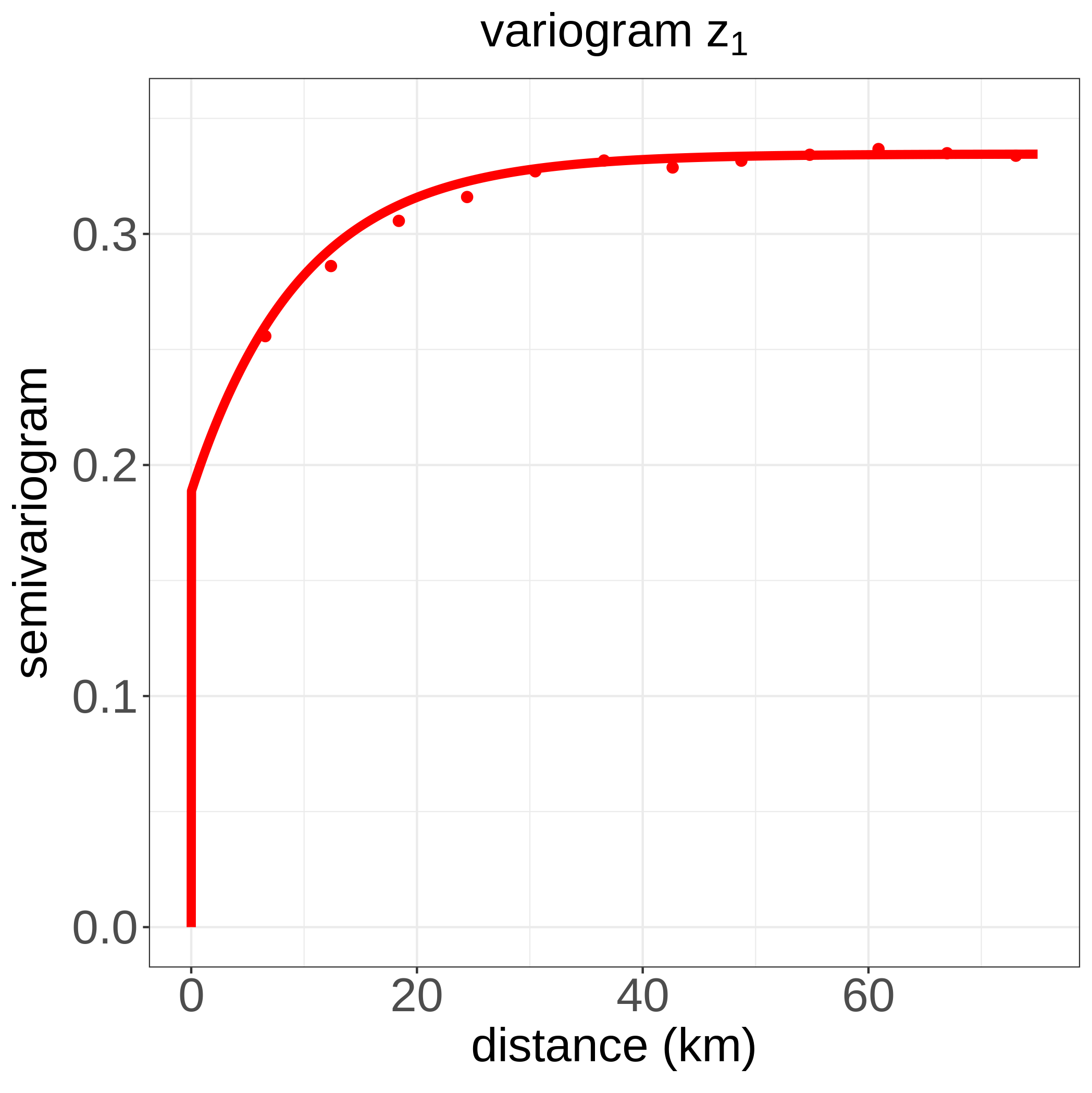}
	\includegraphics[width=0.32\linewidth]{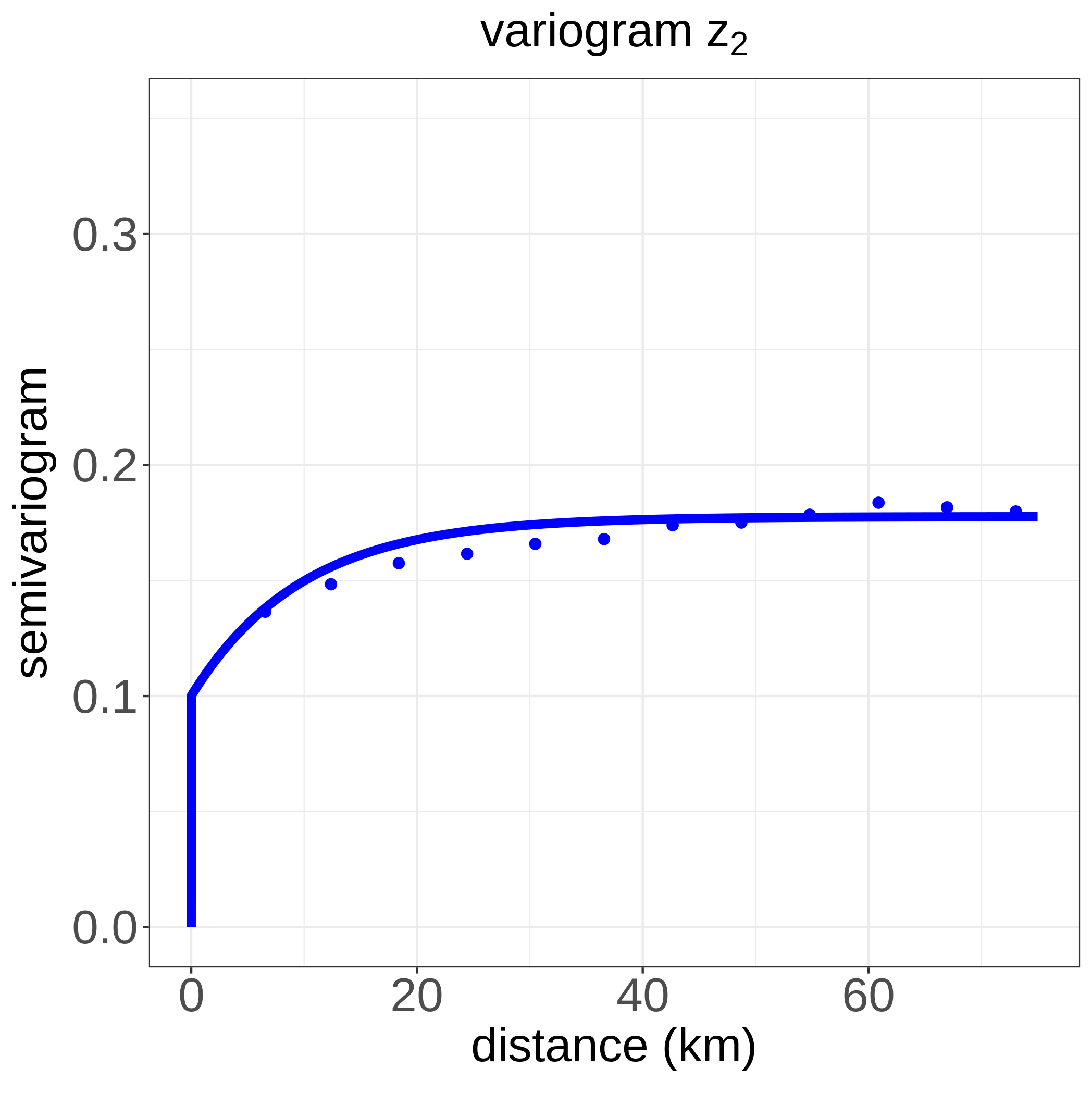}
	\includegraphics[width=0.32\linewidth]{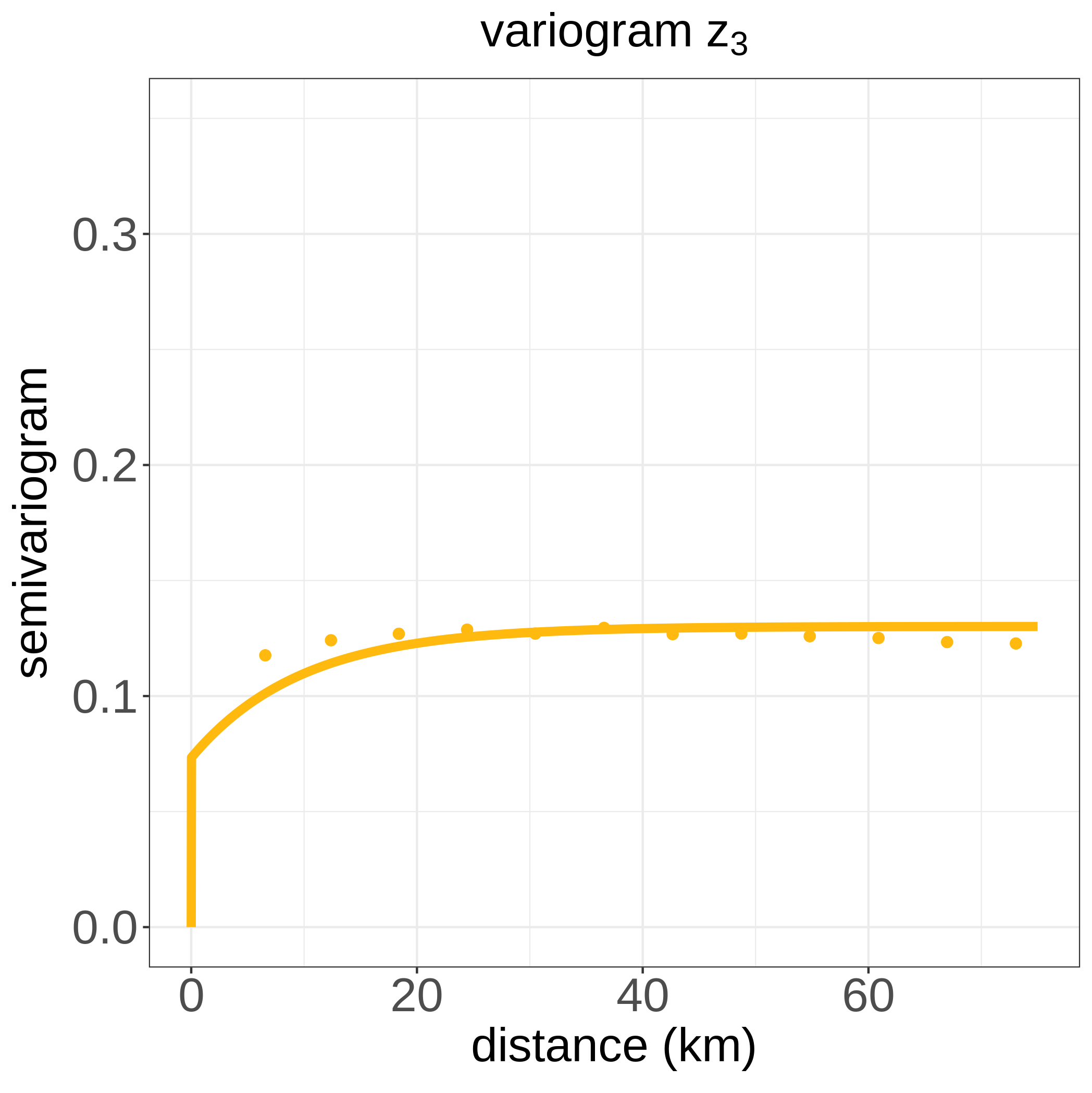}\\
	\includegraphics[width=0.32\linewidth]{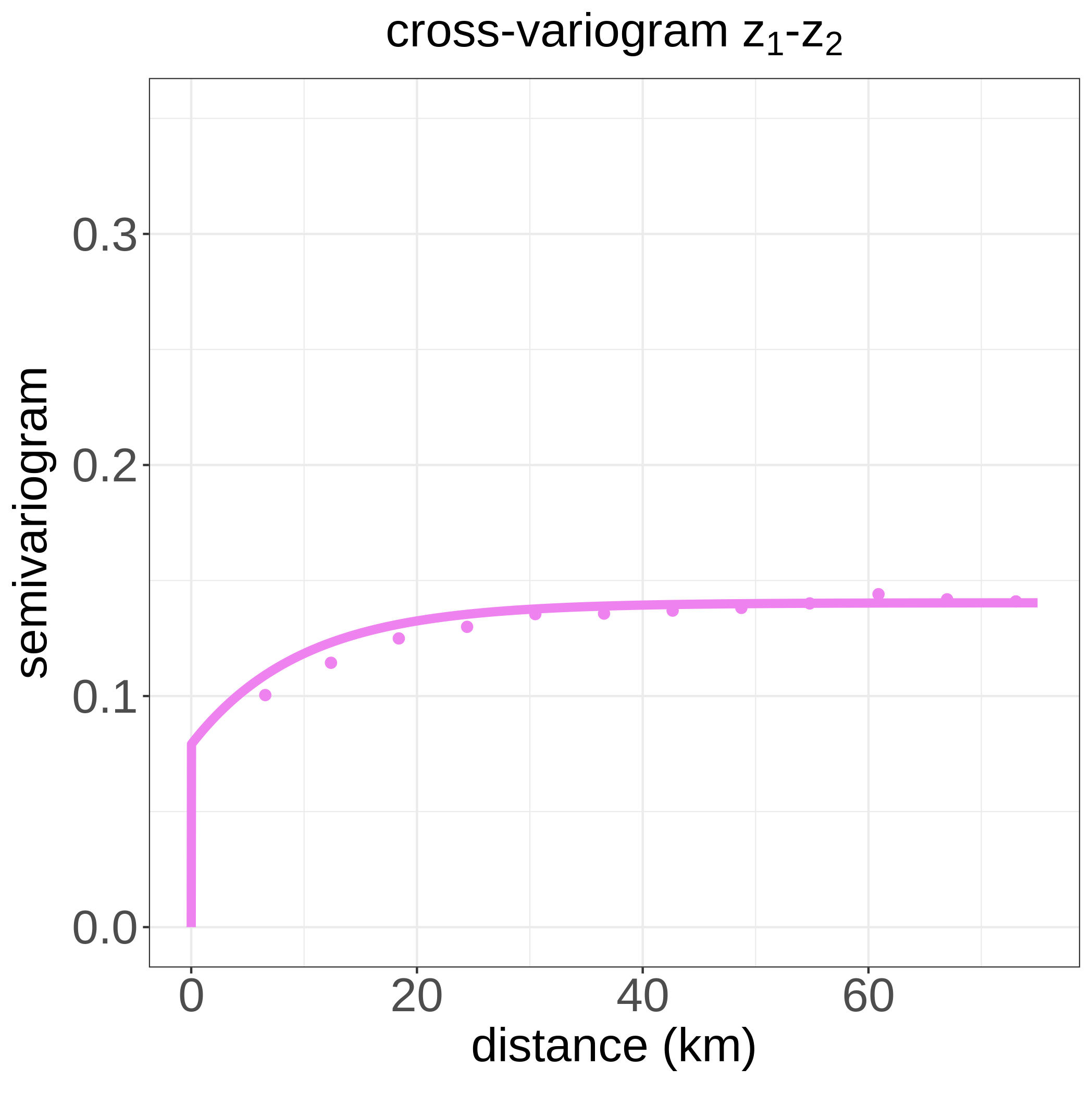}
	\includegraphics[width=0.32\linewidth]{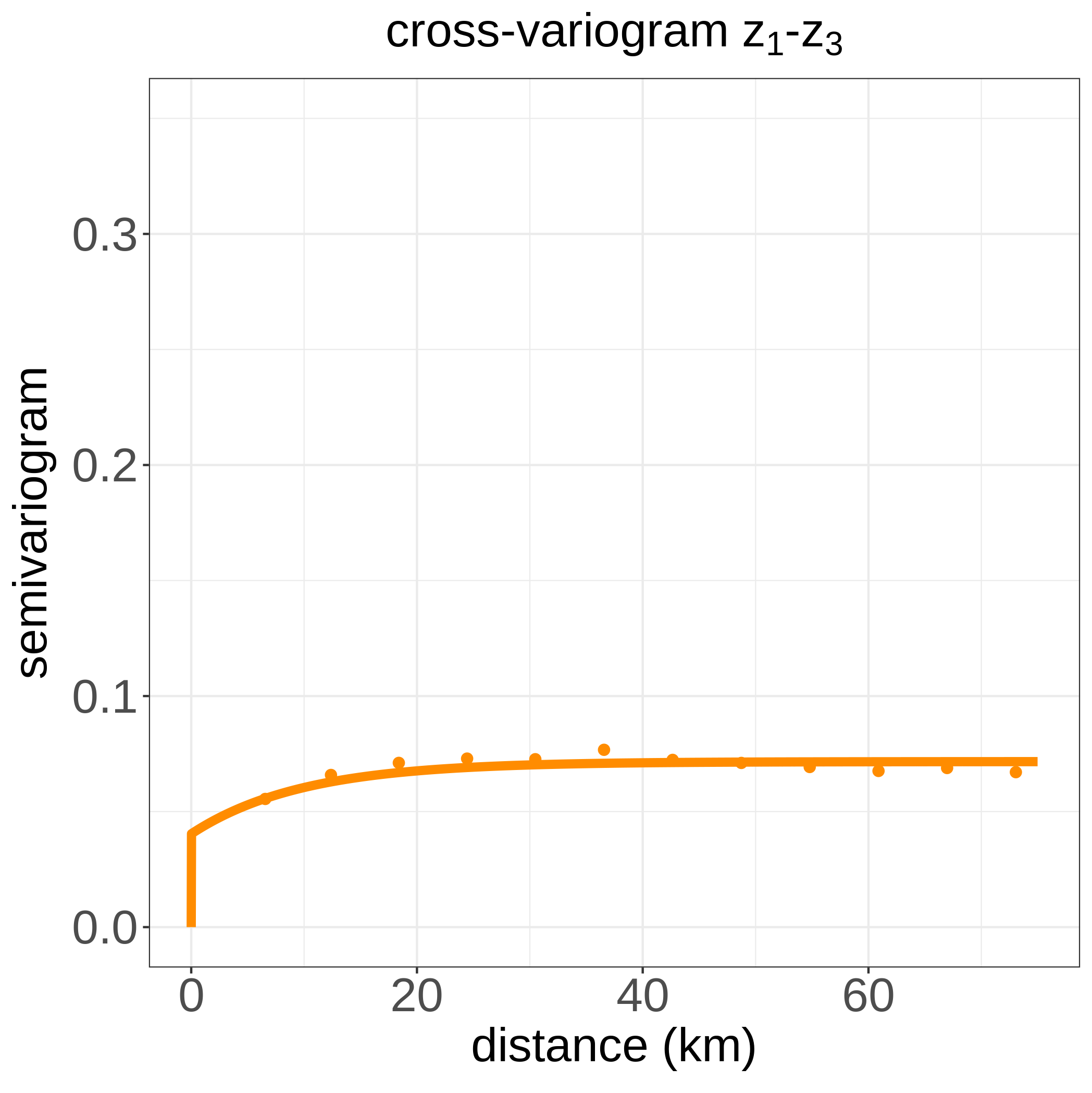}
	\includegraphics[width=0.32\linewidth]{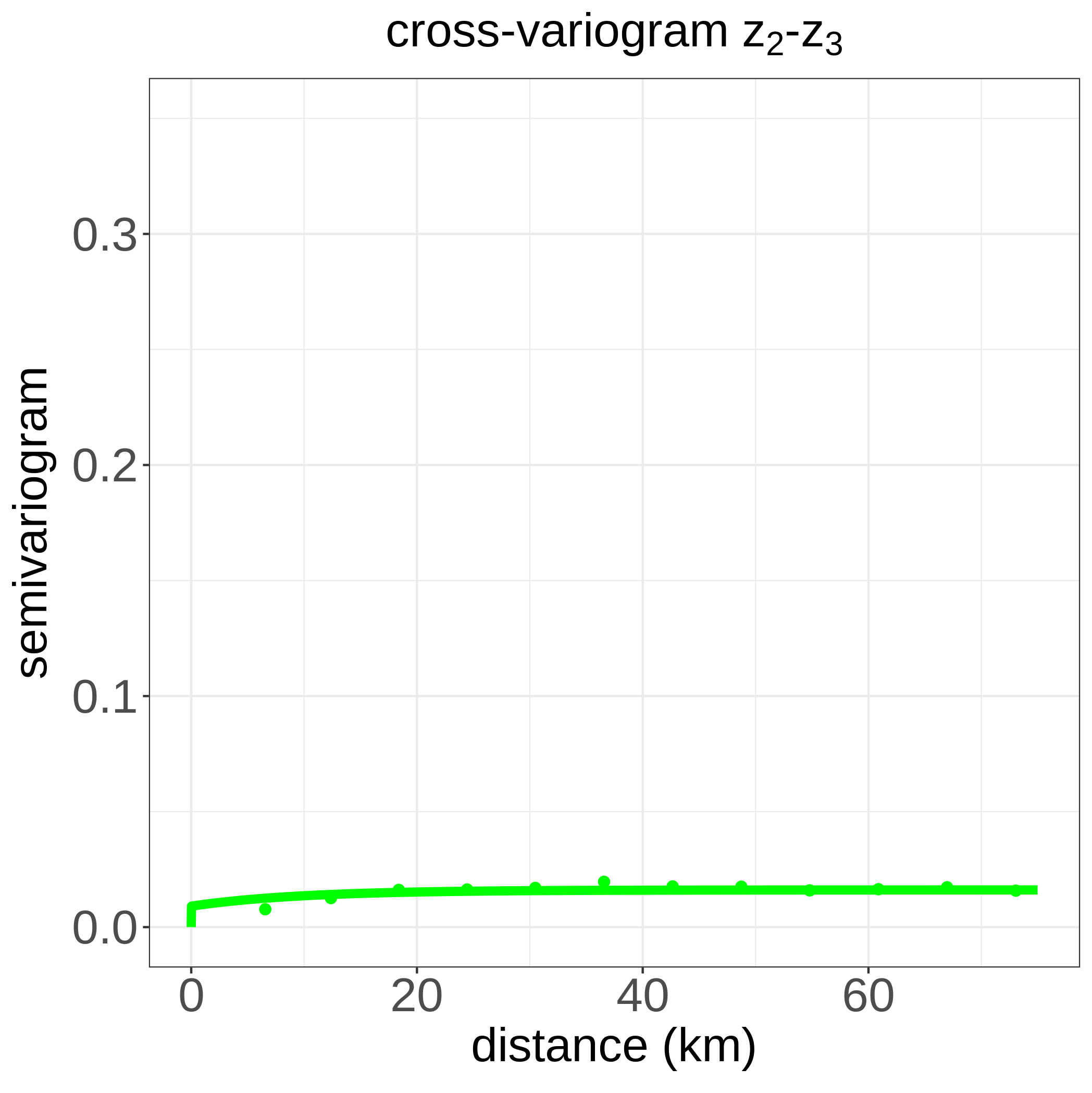}
	\caption{Variograms and cross-variograms for $z_1$, $z_2$, $z_3$ of the dataset without 0s after $\alpha$-IT with $\alpha^\star=0.12$.}
	\label{fig:cross_variogram_po_no_zero}
\end{figure}

The average of the kriging errors at the validation locations across the $B$ simulations for different values of $\alpha$ computed for the two metrics (Hellinger metric and Total Variation metric) presented in Equation \eqref{eq:mae_metric} are reported in Figure \ref{fig:mean_no_zero_tv_hel}. According to the Total Variation metric, $\alpha^\star=0.12$ leads to the best result compared to all other values of $\alpha$, including $\alpha=0$ and $\alpha=1$.  Regarding the Hellinger metric, the minimum of the kriging error among the 6 values of $\alpha$ is obtained for $\alpha=0.3$, which is not exactly the value that corresponds to the maximal proximity to Gaussianity (i.e., $\alpha^\star = 0.12$), but it is close to this value. Notice that the true minimum could be attained for some other value within the range $(0.12,0.50)$, but this has not been evaluated here. Figure \ref{fig:final_kriging_no_zero} shows the back-transformed proportions of the kriging values over a grid made of 10000 points after the $\alpha$-IT with $\alpha^\star=0.12$ given $500$ observed points. Notice that at each location, these predicted proportions are non-negative and that their sum is equal to 1. 

To summarize, this analysis has shown that the values $\alpha=0$ and $\alpha=1$, corresponding to the Aitchison and the linear transformation respectively, are not optimal for both the Hellinger and the Total Variation metrics. An intermediate geometry, corresponding to the $\alpha$-IT with $\alpha^\star=0.12$, seems to better describe the dataset in a geostatistical setting. 
 
\begin{figure}[!htp]
	\centering
	\includegraphics[width=0.5\linewidth]{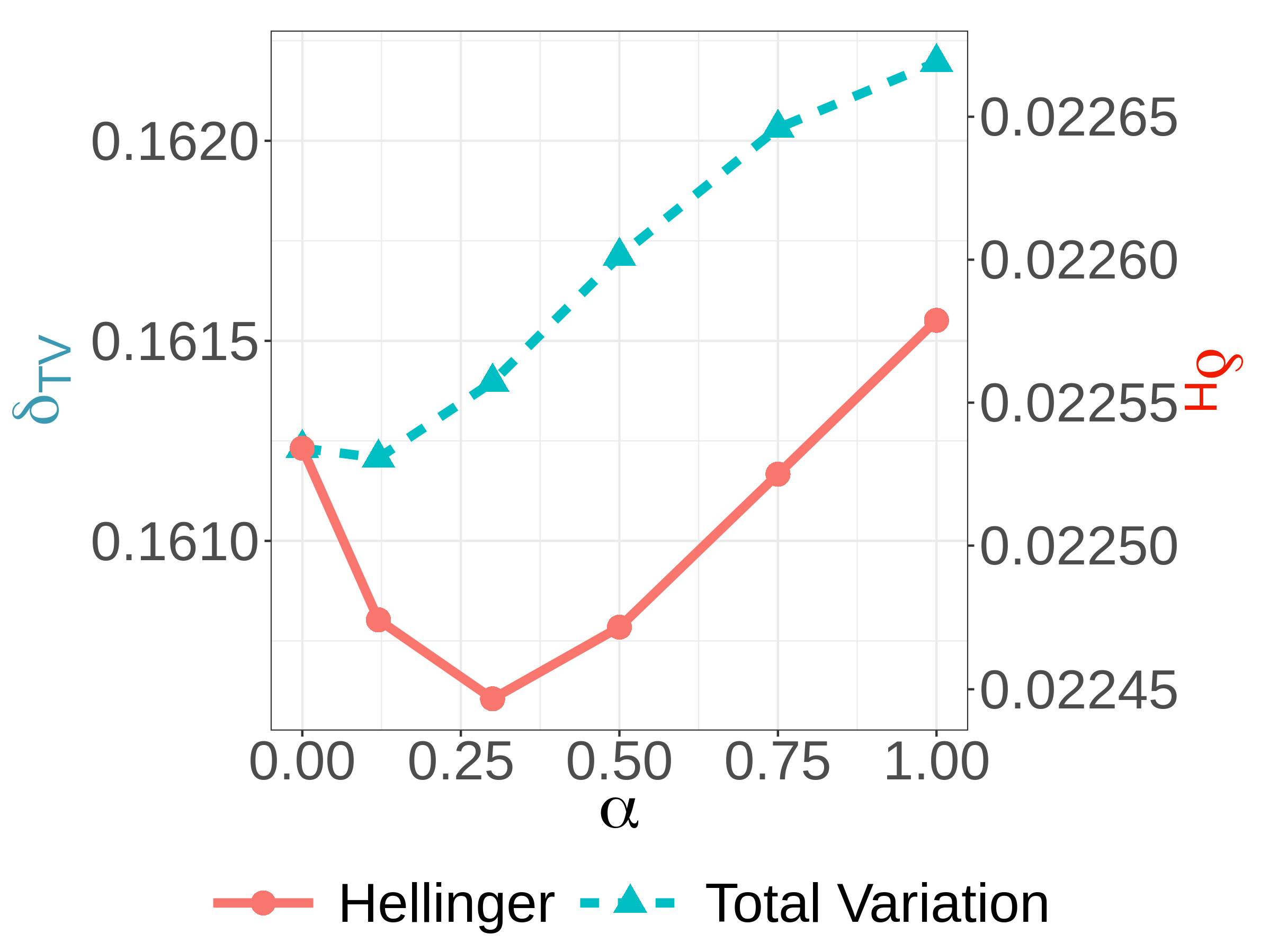}
	\caption{Average of the kriging errors: Total Variation metric in blue dashed lines and Hellinger metric in red continuous lines.}
	\label{fig:mean_no_zero_tv_hel}
\end{figure}

\begin{figure}[!htp]
	\centering
	\includegraphics[width=0.48\linewidth]{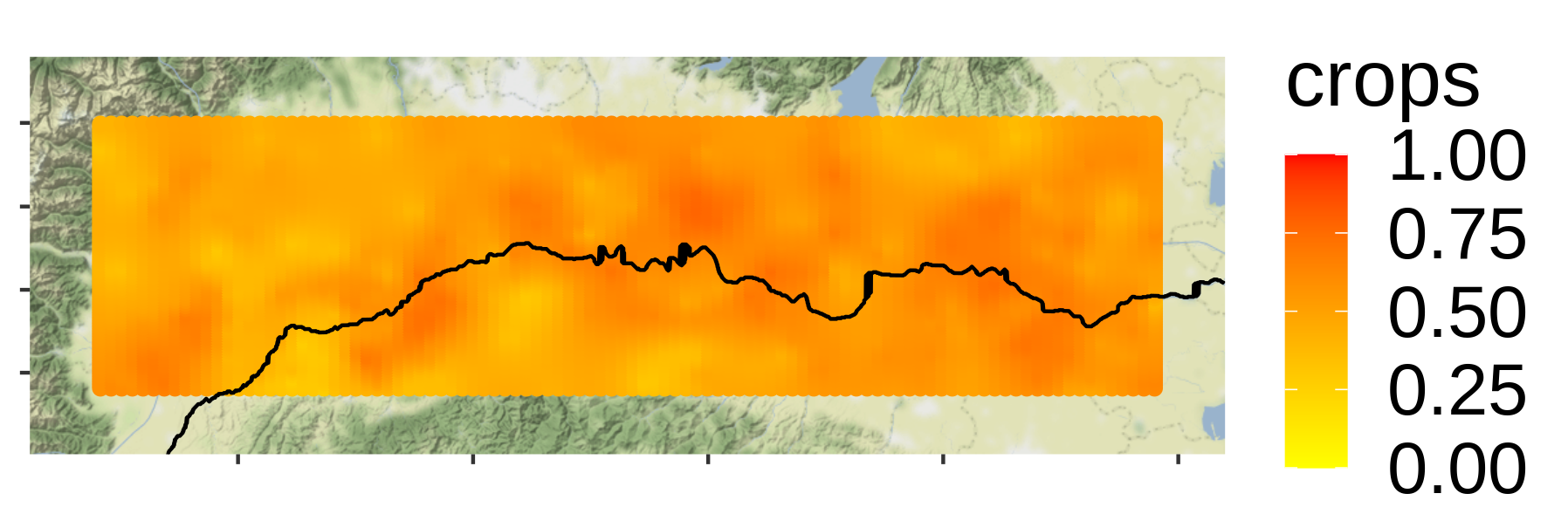}
	\includegraphics[width=0.48\linewidth]{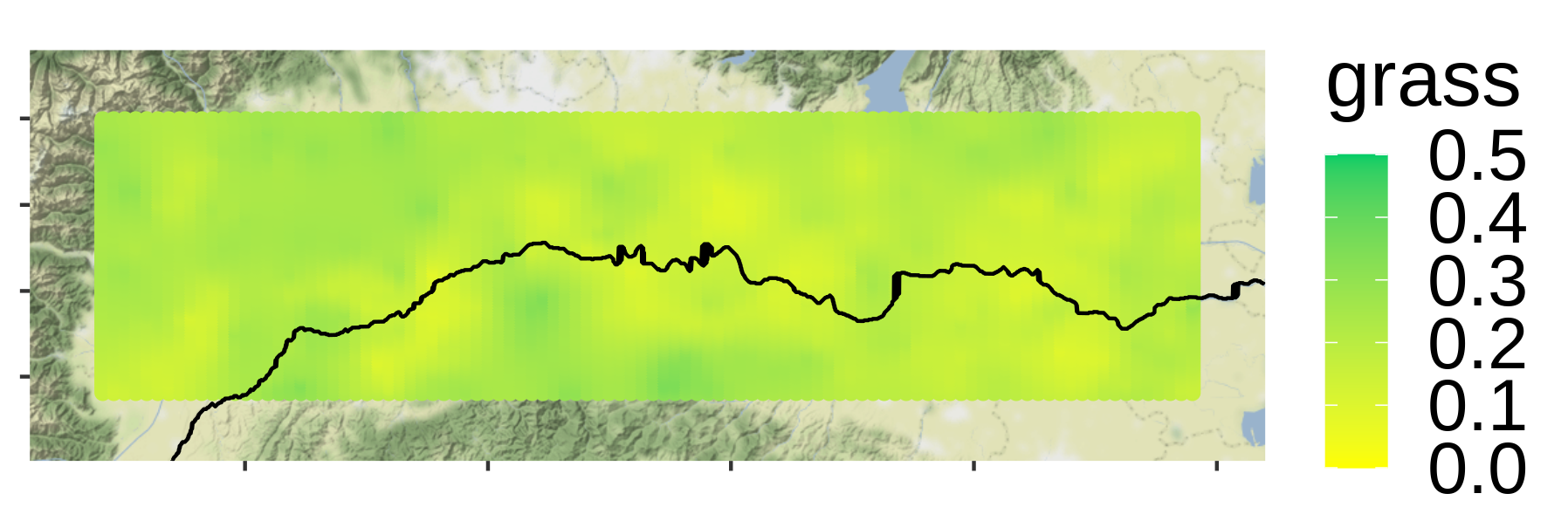}\\
	\includegraphics[width=0.48\linewidth]{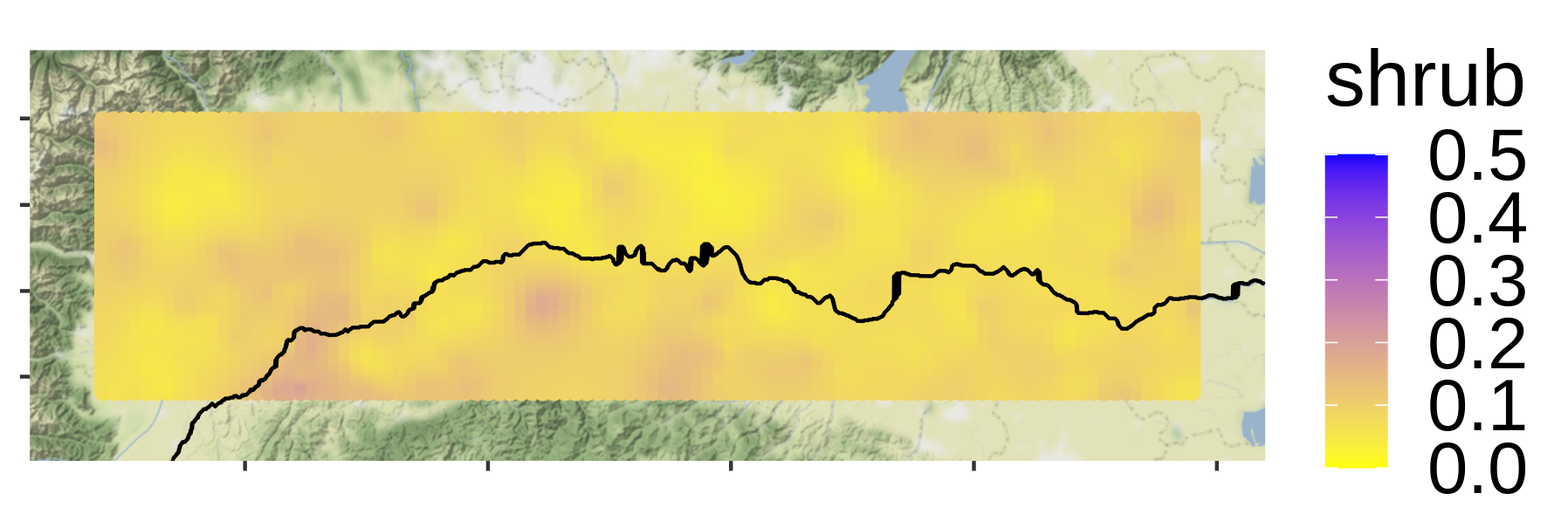}
	\includegraphics[width=0.48\linewidth]{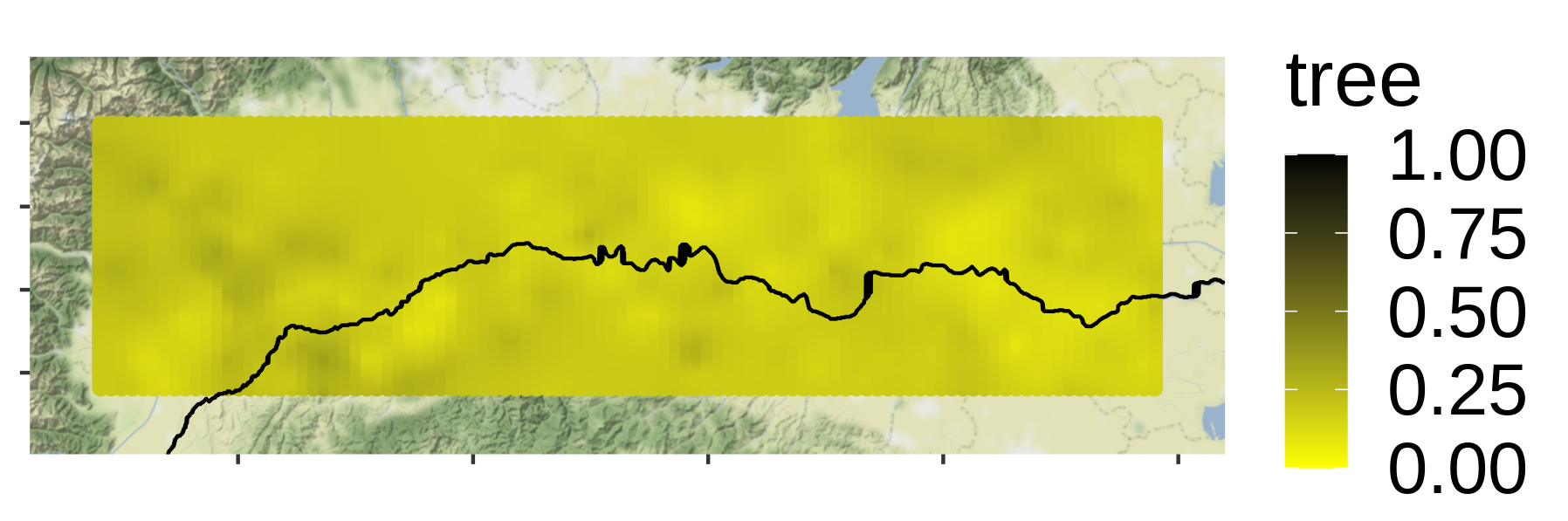}
	\caption{Kriging proportions over a grid of 10000 points with $\alpha^\star=0.12$. Proportions are positive and sum to 1.}
	\label{fig:final_kriging_no_zero}
\end{figure}

\subsection{Spatial analysis with 0s}
\label{sec:spatial_with0}

In the complete Copernicus dataset the land cover proportion is actually equal to 0 for at least one of the four parts in about half of the pixels (53\% of compositions without 0s, 31\% with 1 zero, 12\% with 2 zeros, 4\% with 3 zeros). 
A compositional geostatistical analysis thus requires to address the issue raised by null values in the vector of proportions. Dealing with 0s is a challenge for the log-ratio approach. In \citet{tolosana-missing}, the authors propose to use the undersampled cokriging, where, even if some coordinates of the observation vectors may be missing, it is possible to look for both an interpolation of full vectors at unsampled locations and the completion of the missing variables at the sampled locations. \citet{imputation_1, imputation_2, imputation_3, imputation_4} propose various methods to impute 0s with small values before the statistical analysis of the data; as an alternative, these authors suggest to remove the samples with at least one 0.

As discussed in Section \ref{sec:alpha_it}, the $\alpha$-IT (with $\alpha > 0$) can be applied to compositions that include one or more components equal to 0, thus overcoming  the theoretical limitations of the log-ratio approach. In order to test our approach on such datasets, the Monte Carlo cross-validation assessment presented in Section \ref{sec:spatial_no0s} is here repeated on the data containing 0s.  Two settings are considered. In the first setting, a moderate proportion of data (set to 10\%) presents at least one part being equal to 0. Among those, the proportion of data with one and two parts equal to 0 reflects the proportion seen in the data. In the second setting, the sample of size 2000 is drawn from the general dataset regardless of the number of null compositions. Hence, on average, about half of the data contain at least one part equal to 0. In both settings, data with a single positive part (thus equal to 1) are excluded. 

In the first setting, the average of the 20 ML estimates for $\alpha$ is equal to $0.120$ and the standard deviation is $0.036$. In the second setting, the average of ML estimates is equal to $0.089$ and the standard deviation is $0.040$. Notice that $\alpha=0$ lies outside the confidence interval in all cases.  The $\alpha$-IT is performed with $\alpha \in A'= \{0.01, \alpha^\star, 0.3, 0.5, 0.75, 1\}$, where 0.01 was chosen as an approximation of 0 (0 being not acceptable). Figure \ref{fig:cop_zeros} represents the Total Variation and the Hellinger scores of the kriged compositions as a function of $\alpha$. The errors are maximal when $\alpha$ is close to $0$. There is a local minimum at $0.12$ and $0.30$ in the first and second setting respectively. Then, the scores reach a floor when $\alpha>0.5$ and they remain almost equal up to $\alpha=1$. The same behavior can be observed on the normalized $\alpha$-IT prediction error (in the Euclidean space) instead of Total Variation or Hellinger distances (in the simplex). For example, in the second setting, $\bar{\delta_\alpha}/\sigma$ is equal to (1.64, 1.58, 1.49, 1.56, 1.42 and 1.41) as $\alpha$ goes from 0.01 to 1 in $A'$.

The results above thus show that even though the ML estimates are close to $\alpha=0.1$, the prediction scores are minimized close to $\alpha=1$, suggesting that absence of transformation of the data is a good option when it comes to kriging in presence of compositions with null parts.

\begin{figure}[!htp]
	\centering
	\includegraphics[width=0.48\linewidth]{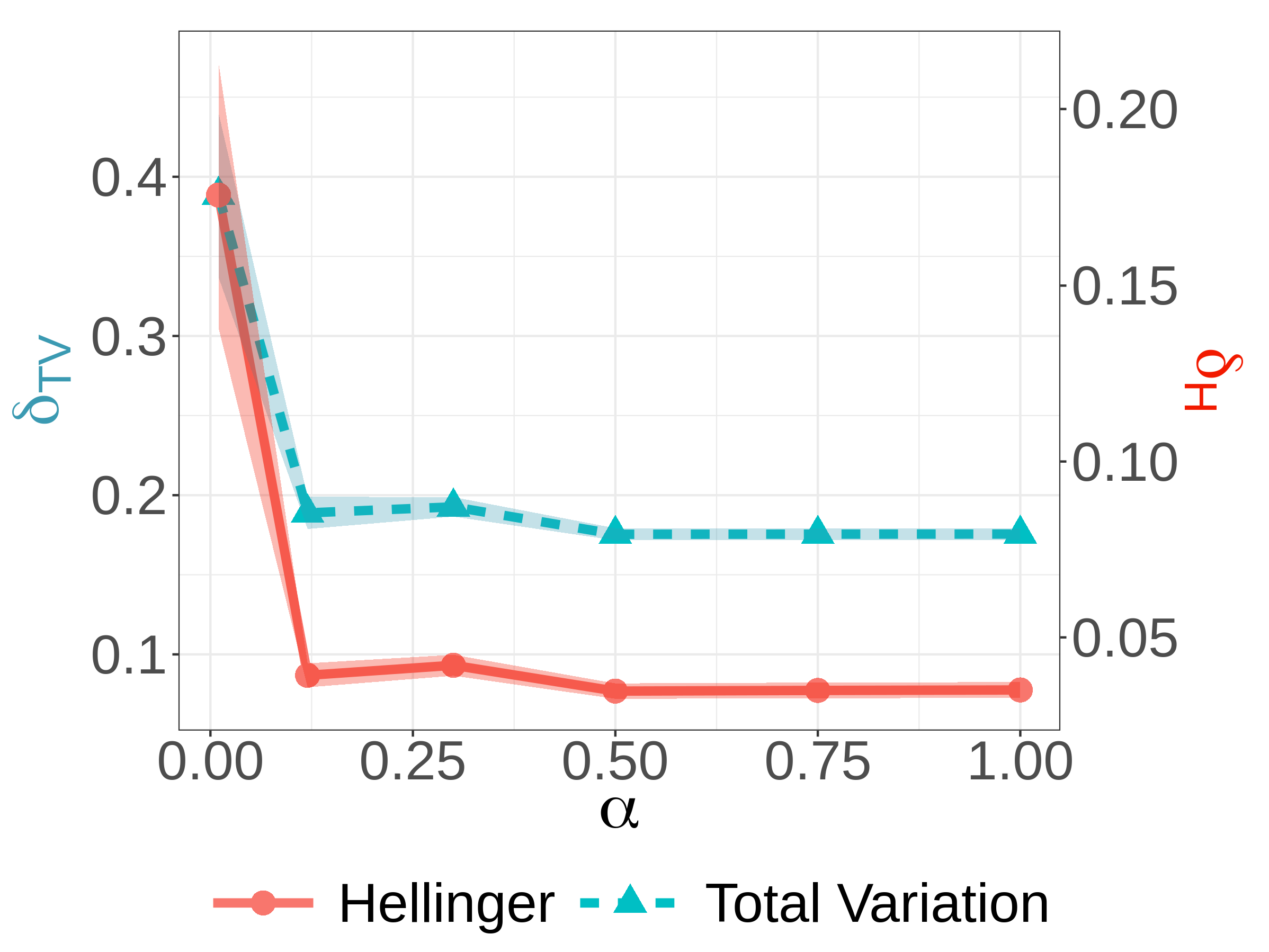} \quad
	\includegraphics[width=0.48\linewidth]{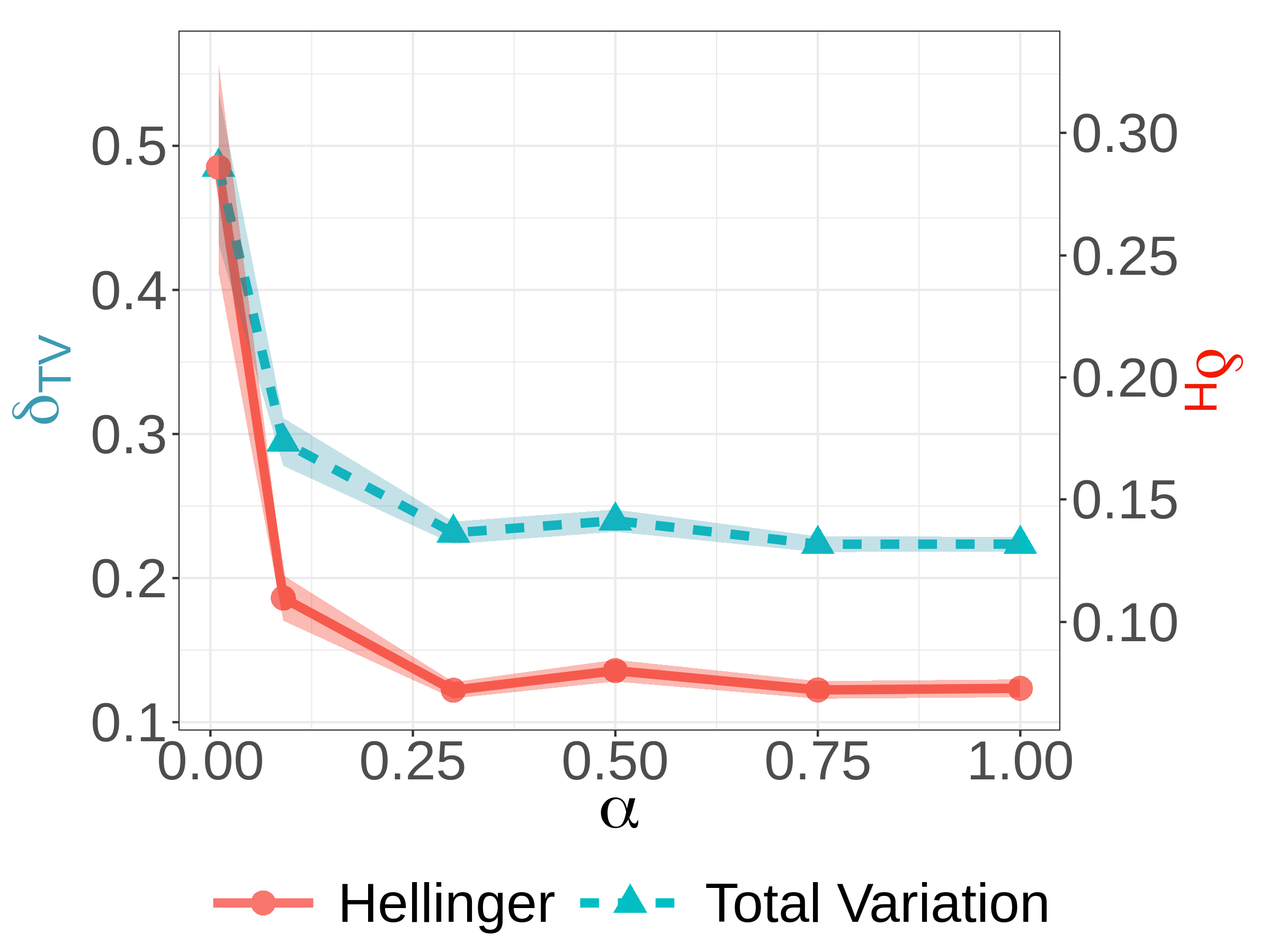}
	\caption{Hellinger and Total Variation metrics computed on prediction errors for 20 random sample datasets with 0s. Thick line: average; light shadow : $\pm 2\sigma$ envelopes. Left:  moderate (10\%) amount of 0s. Right: large amount of 0s.}
	\label{fig:cop_zeros}
\end{figure}

In order to understand this counter-intuitive behavior, Figure \ref{fig:euc_cop_zeros_3D} offer a 3D representation of the set of data with a large proportions of compositions with null parts, along with the kriged vectors in the Simplex (left) and in the 3D Euclidean space  (right). Two values of $\alpha$ are considered: in the top row, $\alpha$ is the ML estimate $\alpha^\star$, and in the bottom row, $\alpha=1$ is used.  It is clearly visible that the $\alpha$-IT produces clusters of data in the Euclidean space. These clusters are associated to different groups of data. The most numerous group lies in the center, and corresponds to data with positive values for all parts. Then, there are four clusters corresponding to the four possible groups of data with exactly one part equal to 0, i.e. data belonging to one of the facets of the tetrahedron. Finally, there are six smaller groups  corresponding to data with two parts equal to 0. These data lie on the edges of the tetrahedron. When $\alpha=1$, the transformation is linear, and these clusters are as close to each other as in the original dataset. As $\alpha$ gets smaller, the distance between these clusters increases. When $\alpha$ approaches 0, the distance between clusters tends to infinity, in relation with the fact that the logarithm of 0 is undefined.  The likelihood in Equation \eqref{eq:likelihood_zeros} is the sum of partial likelihoods computed within each group, but it does not account for the distance between groups. The maximum of the likelihood is thus reached when the groups are the closest to Gaussianity, regardless of the clustering effect. 

When kriging is computed at a target location, data in the neighborhood can belong to different clusters (some data with only positive parts, some others with one or two 0s, possibly not for the same parts). In this case, kriging, which is the weighted average of values belonging to different clusters, can be located in an area of the Euclidean space where there are no data. This is the case for the light blue points represented in Figure \ref{fig:euc_cop_zeros_3D}. 
The kriging RMSE in the $\alpha$-IT geometry will then increase as $\alpha$ decreases. Once back-transformed, these points will end up in regions in the simplex where there are no data, implying high values of the Total Variation and Hellinger metrics. Figure \ref{fig:euc_cop_zeros10pc_3D} represents the same plots for a moderate proportion of compositions with null parts. 

In conclusion, the clustering effect due to compositions with 0s is amplified as $\alpha$ decreases, entailing poor kriging performances both within the $\alpha$-IT geometry and for metrics in the simplex. Note that a similar problem is addressed in Compositional Data Analysis when analyzing the so-called \emph{essential zeros}, i.e., 0s which are not due to rounding or low sensitivity of the measurement instrument. In these cases, zero-replacement is clearly not appropriate. A strategy of analysis proposed by some authors (e.g., \citet{Aitchison-Kay2003}) consists of two separate steps, namely (\emph{i}) analyzing the patterns of zeros (i.e., the data clusters) and (\emph{ii}) analyzing, within clusters, the non-zero parts according to the Aitchison geometry. A similar strategy can be envisioned in the setting of $\alpha$-IT, including an explicit modeling of the clusters induced by the presence of 0s. A relevant advantage here would be that the geometry induced by the $\alpha$-IT is well-defined over all the compositions, thus allowing modeling across clusters, and not only within clusters.

\begin{figure}[!htp]
	\centering
    $\vcenter{\hbox{\includegraphics[width=0.48\linewidth]{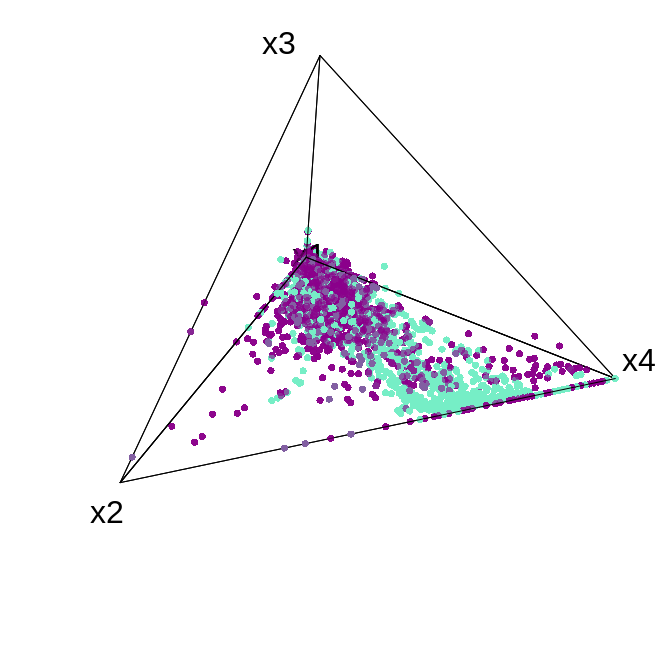}}}$
    $\vcenter{\hbox{\includegraphics[width=0.48\linewidth]{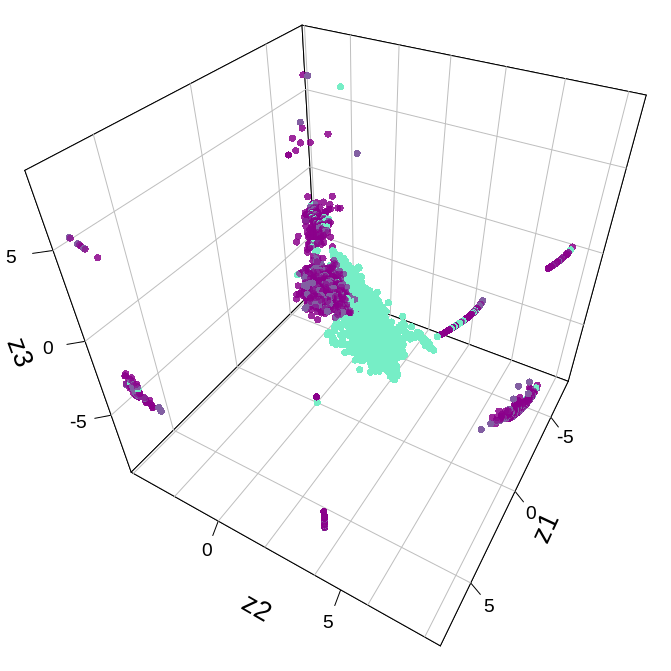}}}$\\
    $\vcenter{\hbox{\includegraphics[width=0.48\linewidth]{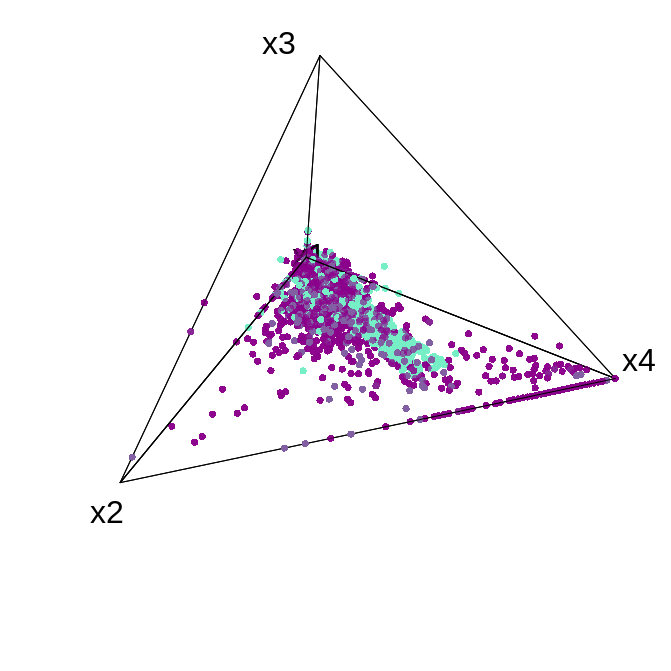}}}$
    $\vcenter{\hbox{\includegraphics[width=0.48\linewidth]{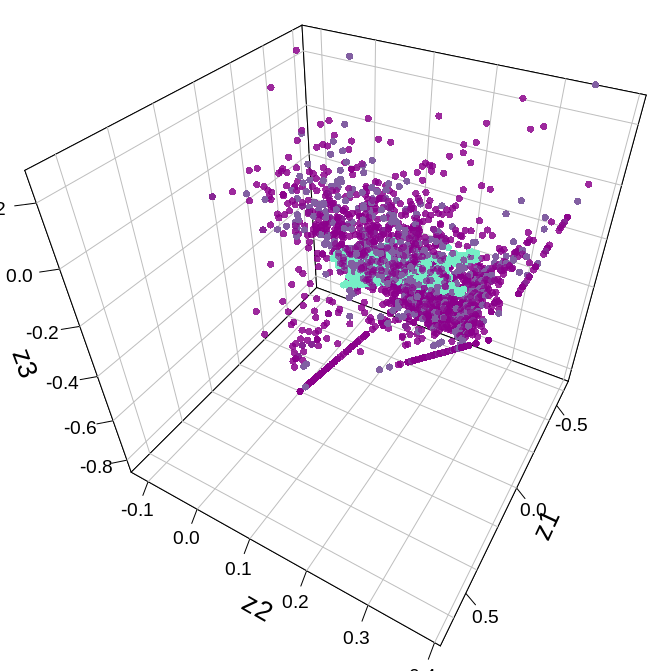}}}$
    \caption{Original data (violet) and kriging predictions (light blue)  with a high (about 50\%) fraction of data with 0s (second setting). Left: simplex. Right: Euclidean space. Top: $\alpha^\star=0.089$. Bottom: $\alpha=1$.}
	\label{fig:euc_cop_zeros_3D}
\end{figure}

\section{Discussion and conclusion}
\label{sec:discussion}

A new class of $\alpha$-transformations, named  $\alpha$-IT, has been proposed, to allow for the geostatistical analysis of compositional data. The transformation has been proved to converge to the Isometric Log-Ratio transformation as $\alpha$ approaches 0, while it reduces to a linear transformation when $\alpha=1$. In this sense, the $\alpha$-IT represents a compromise between the Aitchison geometry ($\alpha=0$) and the Euclidean one ($\alpha=1$). Nonetheless, the presence of the parameter $\alpha$ controlling the degree of transformation applied to the data offers the relevant advantage of letting the problem suggest the most appropriate transformation to use on the data. As far as kriging is concerned, we proposed a maximum likelihood estimator for $\alpha$, which maximizes the Gaussianity of the transformed data -- hence the kriging performances. On simulations of compositions originating from Gaussian random fields with inverse $\alpha$-IT, our results show that the ML estimator provides close to unbiased estimates for $\alpha$, i.e., it allows to correctly identify the $\alpha$-IT yielding to data Gaussianity and optimizing kriging performances. 

It is worth mentioning that, unlike alternative classes of $\alpha$-transformations nowadays available \citep{tsagris-alpha}, the $\alpha$-IT also allows for an explicit characterization of the covariance structure of the field according to the geometry induced by the transformation. In particular, this enables one to formulate a generalization of the so-called ILR covariance matrix (see Section \ref{sec:geostat}) -- widely-used in compositional data analysis -- while drawing a direct connection with the latter as $\alpha$ approaches 0.

Simulation results and data analyses on Copernicus land cover data confirm that the Aitchison geometry -- retrieved for $\alpha=0$ -- may not be optimal for kriging, as it proved to be outperformed by the $\alpha$-IT when $\alpha$ was set through ML. On the other hand, the $\alpha$-IT allows one to explicitly deal with compositions with 0s, for which the log-ratio approach is not well-defined. Our investigation on this point provides evidence that the use of a small value for $\alpha$ may not be appropriate in the presence of null parts, because this amplifies the grouping structure merely induced by the 0 themselves, with detrimental effect on kriging performances. Our results indeed suggest that the use of a geometry close to the Euclidean one ($\alpha \simeq 1$) may be more appropriate in these cases instead. These findings complement, from a different perspective, the theoretical results established in \citet{allard2018means}, which state that linear combinations of the data (and thus kriging with $\alpha=1$) are the only unbiased central tendency characteristics satisfying a small set of axioms, namely continuity, reflexivity, and marginal stability. 

In this work, we computed the cokriging of the transformed vector $\Zbold(s_0)$ at unsampled locations $s_0$, conditional on the vectors $\Zbold(s_j) = A_{\alpha-IT}(\Xbold(s_j))$, $j=1,\dots,n$. Ultimately, the goal is to predict the compositional vector $\Xbold(s_0)$ given the same information. However, by back-transforming the predictions from $\mathbb{R}^{D-1}$ to $\mathbb{S}^D$ with a non-linear transformation $A^{-1}_{\alpha-IT}$, one actually introduces a bias even though the prediction of $\Zbold(s_0)$ is unbiased. This issue is closely related to the above mentioned results in \citet{allard2018means}. However, it should be noted that the definition of the bias is tightly related to the geometry under consideration. Implicitly, the geometry in use when evaluating the bias on the prediction of $\Xbold(s_0)$ is the Euclidean geometry in the simplex, whilst the geometry implied by the $\alpha$-IT distance defined in Section 3.3 is equal to the Euclidean distance in the transformed space. Defining and correcting the bias is an interesting yet very broad issue  worth to be further investigated.  The problem of dealing with missing data, when some of the compositional vectors are incomplete, has not been tackled here. It is also left for future work.

More in general, we remark that, while this work focuses on the analysis of spatial compositional data with a focus on spatial prediction (kriging), the $\alpha$-IT transformation can be used for any type of compositional analysis, ranging from exploratory analysis, to classification, inference and regression. 
In these settings, one may use the $\alpha$-IT to convert the original compositional data into Euclidean data, analyze them according to the Euclidean geometry, and then transform back the results to the simplex, with the same process as that demonstrated in this work. While we generally envision a clear path of research to this scope, further research is needed to export the idea here presented to further settings, in the direction of, e.g., establishing criteria to set the optimal value for $\alpha$. Nonetheless, our $\alpha$-IT can be a promising alternative to other transformations for the very same reasons that led us to its development in the spatial setting.

\section*{Acknowledgments} The authors are grateful to  four anonymous reviewers for their very careful reading and the many valuable comments that helped to improve the manuscript.
{\tt R} scripts reproducing our analyses can be found at \url{http://github.com/luciclar/alphaIT_spatial_compositional}. 

\newpage

\bibliographystyle{apalike}
\bibliography{references}

\pagebreak

\section{Supplementary Material}
\label{sec:sup}
\beginsupplement

\begin{figure}[!htp]
	\centering
	\includegraphics[width=0.48\linewidth]{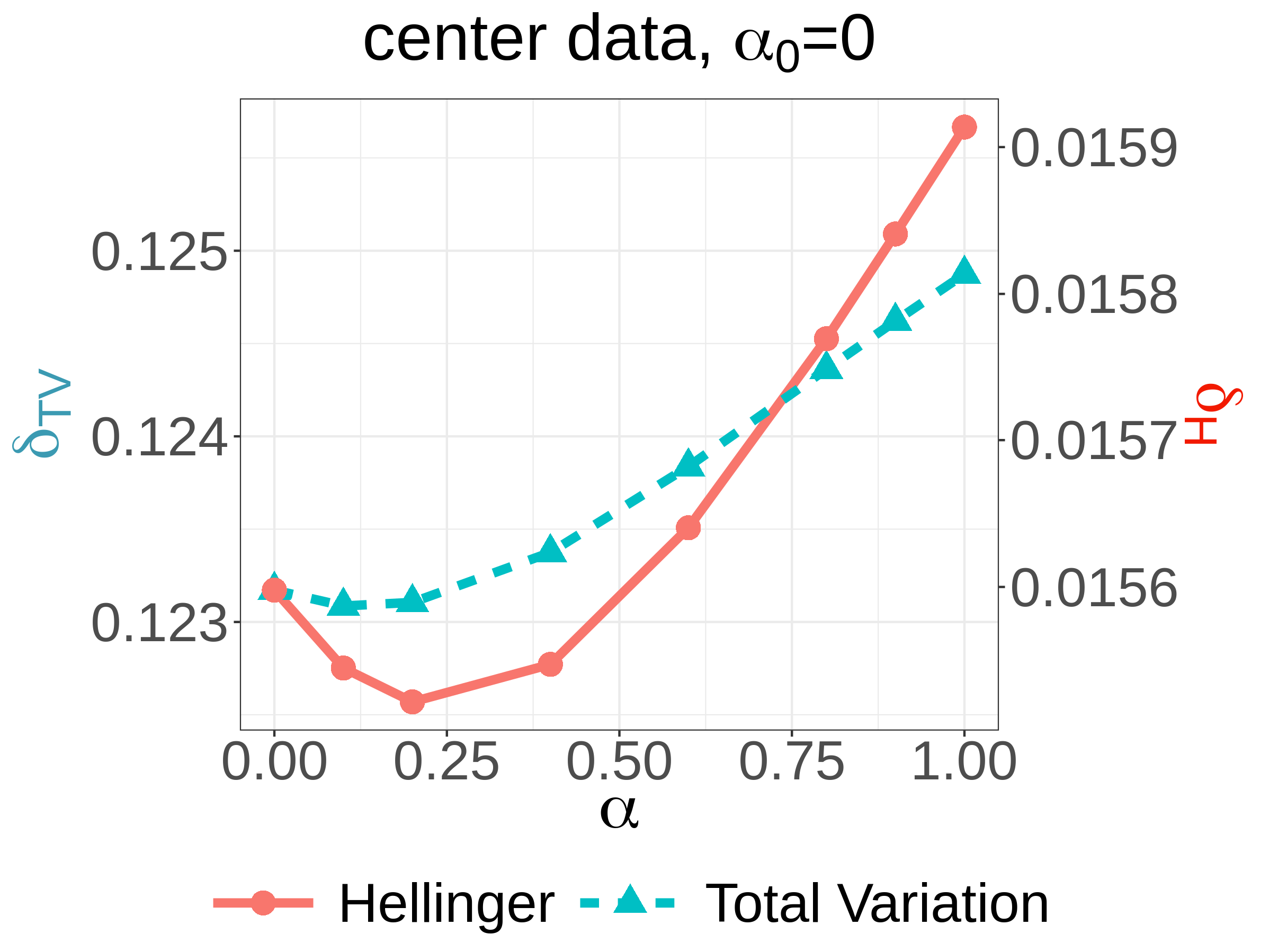}
	\includegraphics[width=0.48\linewidth]{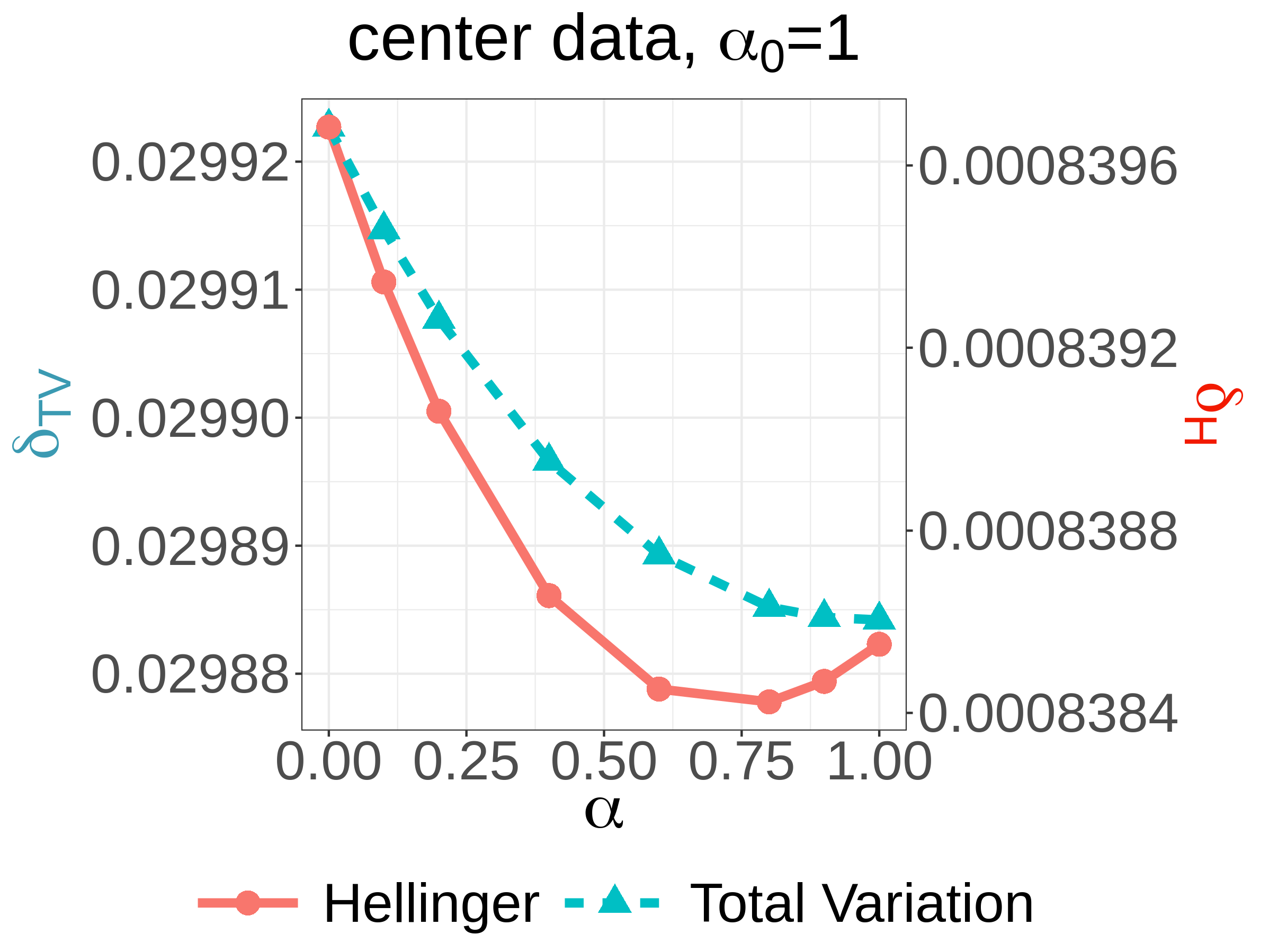}
	\caption{Average error over $B=100$ tests for center data after $\alpha$-IT with different $\alpha$ and $\alpha_0=0$ (left) and $\alpha_0=1$ (right): Total Variation metric in blue dashed lines and Hellinger metric in red continuous lines.}
	\label{fig:means_0_1}
\end{figure}

\clearpage

\begin{figure}[!htp]
	\centering
	\includegraphics[width=0.48\linewidth]{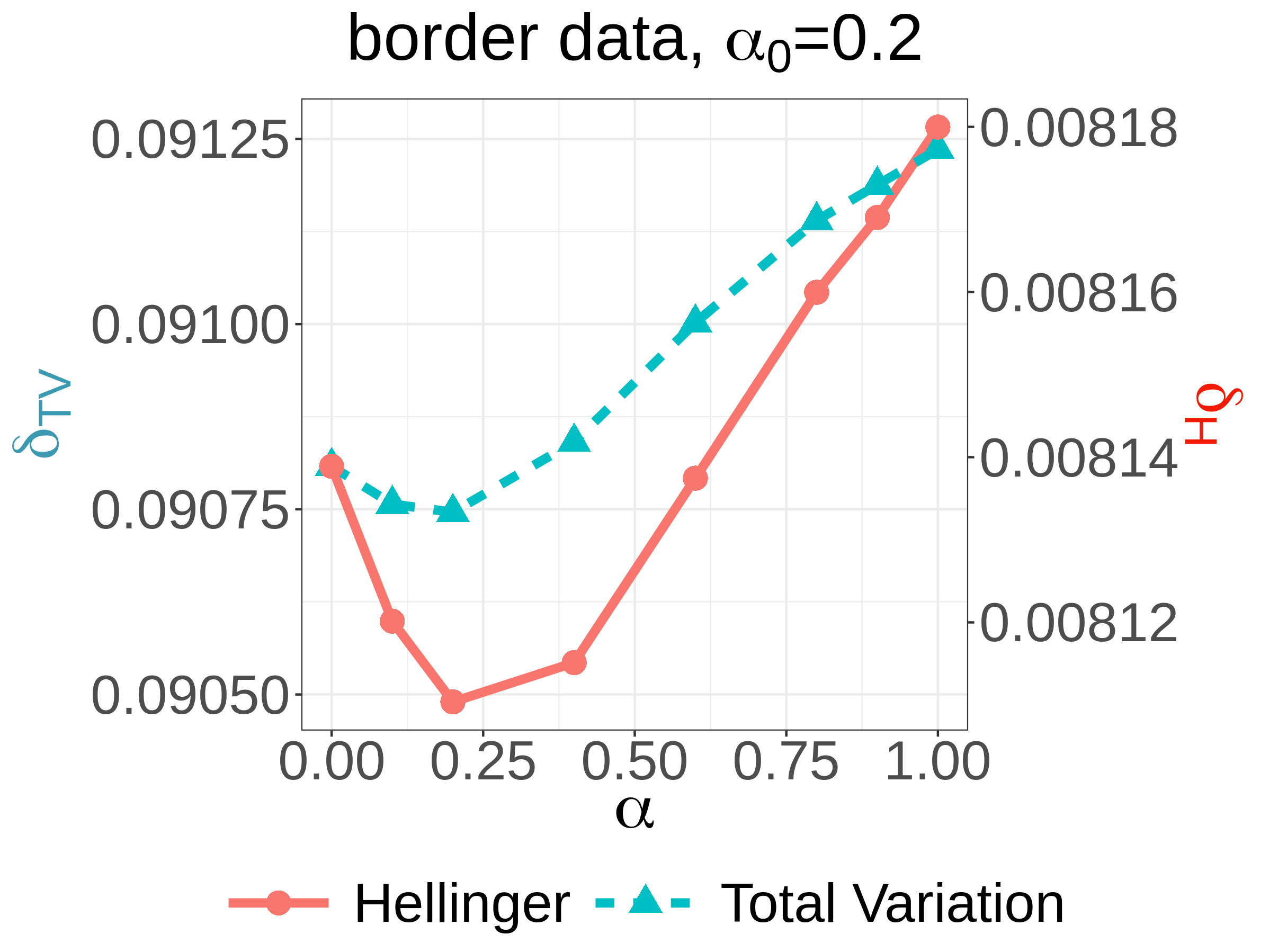}
	\includegraphics[width=0.48\linewidth]{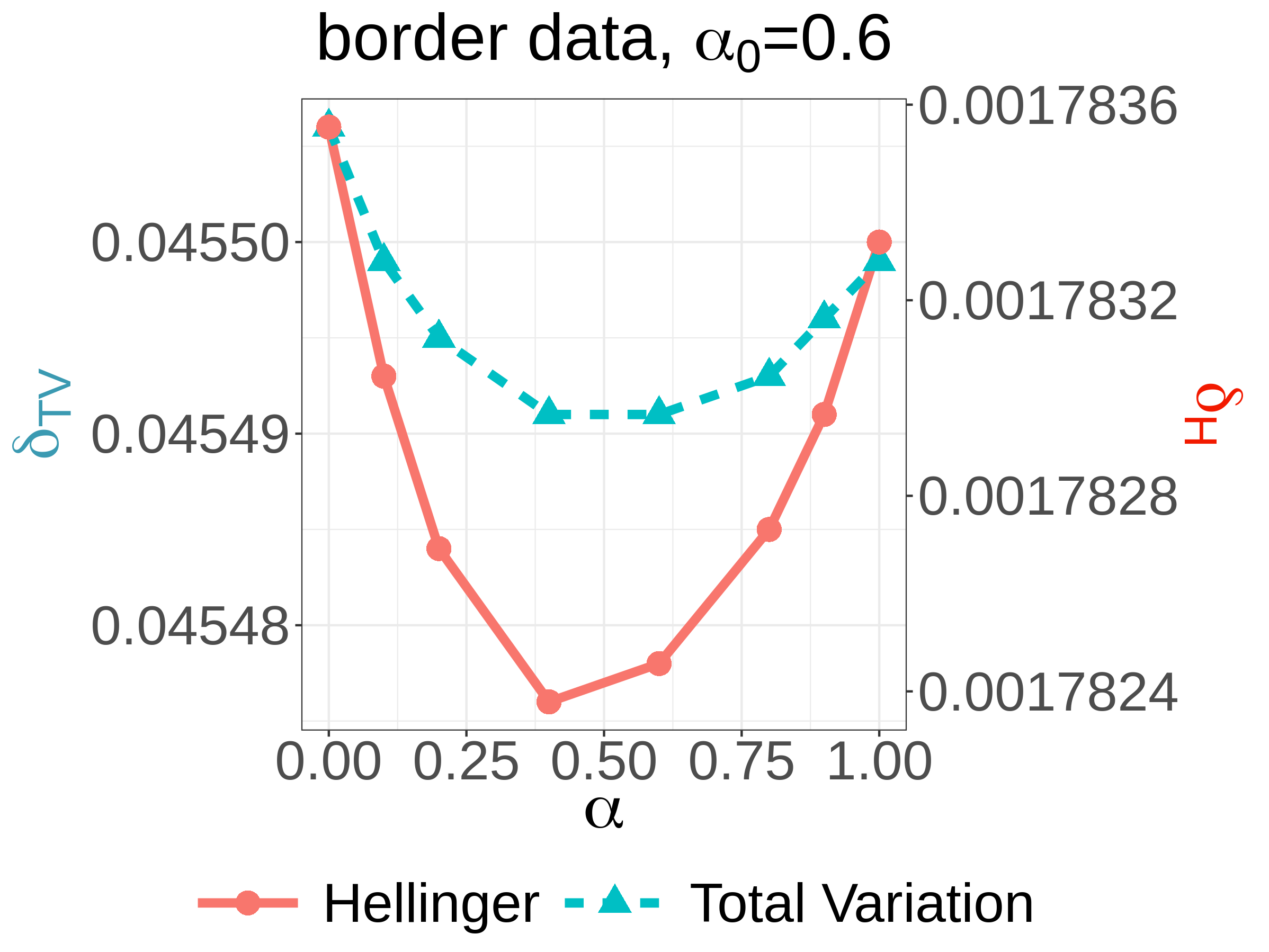}\\
	\vspace{1cm}
	\includegraphics[width=0.48\linewidth]{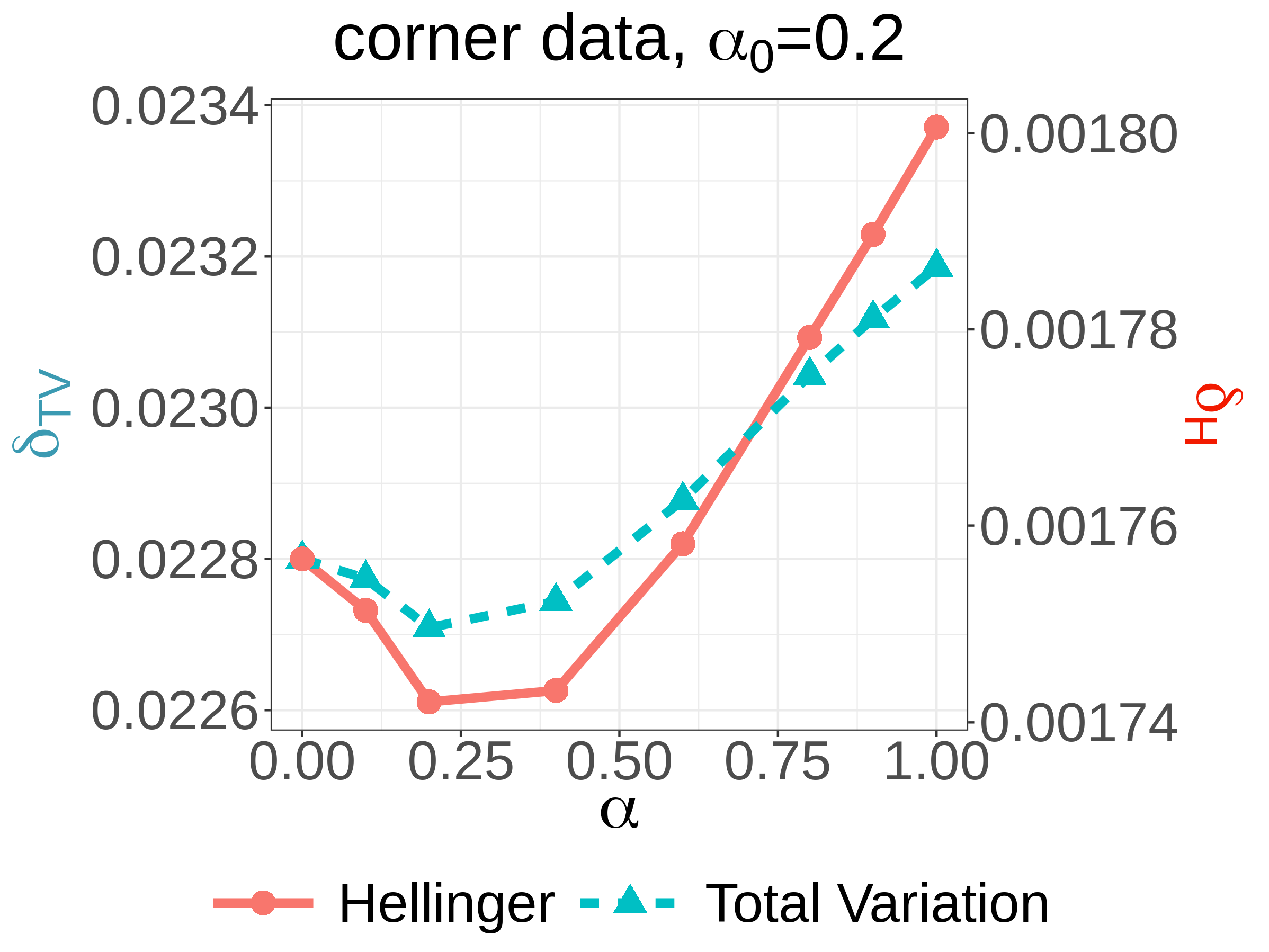}
	\includegraphics[width=0.48\linewidth]{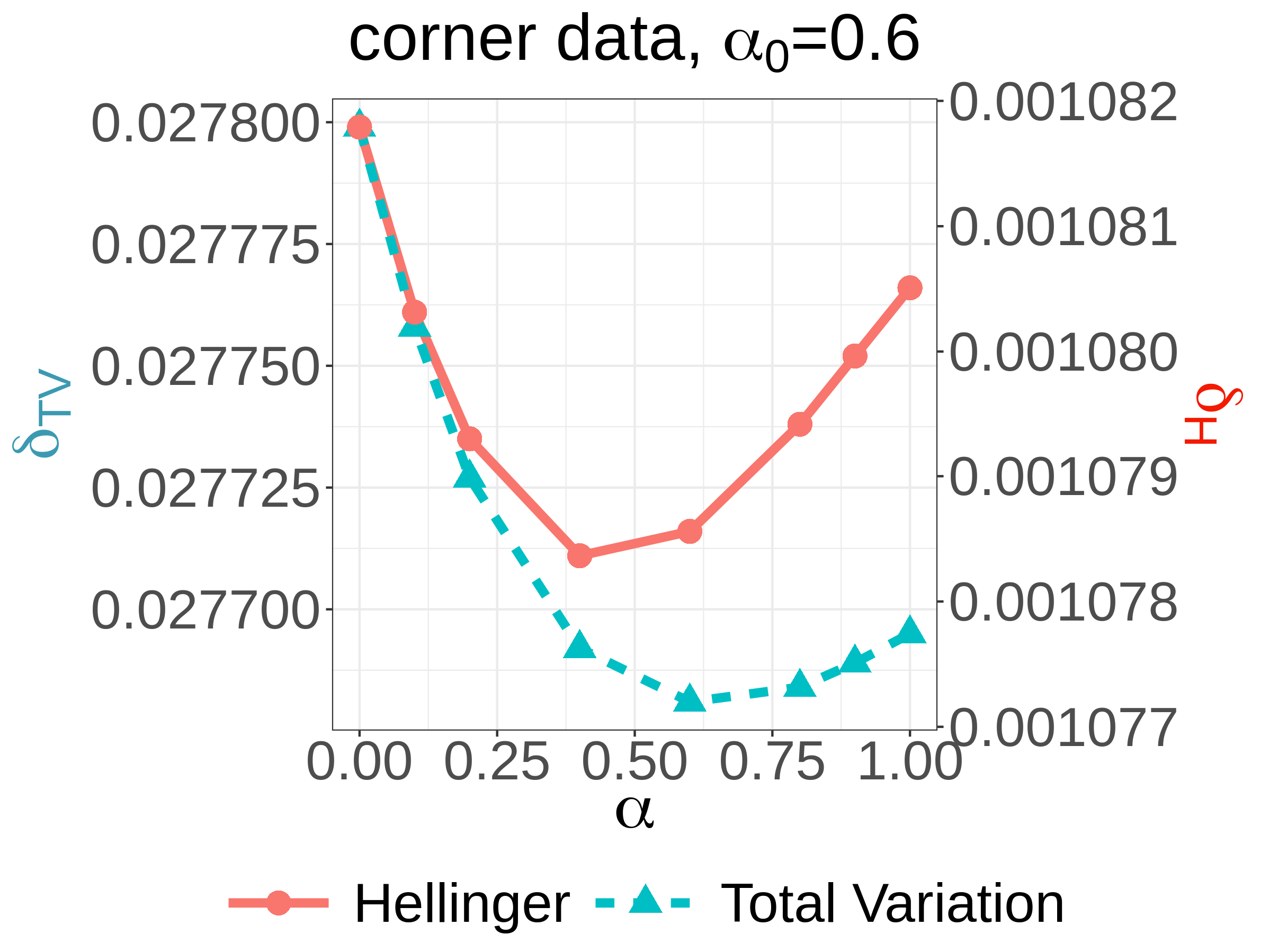}
	\caption{Average error over $B=100$ tests for border (up) and corner (down) data after $\alpha$-IT with different $\alpha$ and $\alpha_0=0.2$ (left) and $\alpha_0=0.6$ (right): Total Variation metric in blue dashed lines and Hellinger metric in red continuous lines.}
	\label{fig:means_border_corner}
\end{figure}

\clearpage

\begin{figure}[!htp]
	\centering
	$\vcenter{\hbox{\includegraphics[width=0.48\linewidth]{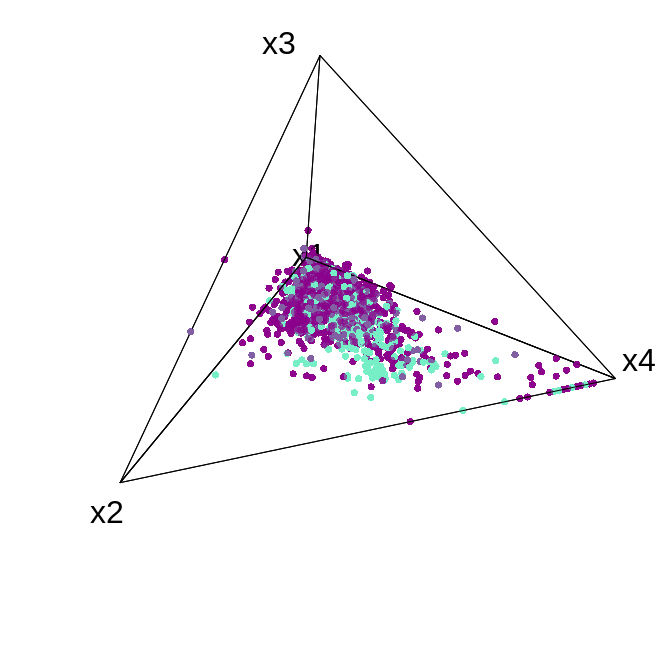}}}$
    $\vcenter{\hbox{\includegraphics[width=0.48\linewidth]{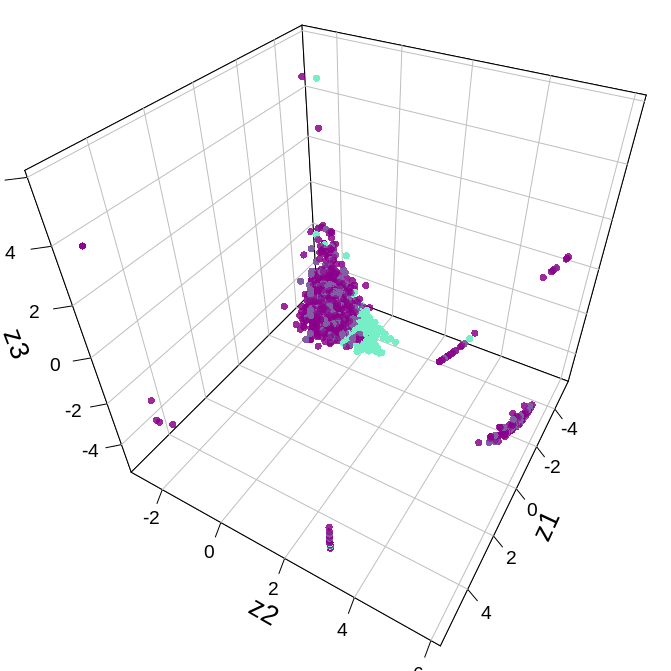}}}$\\
    $\vcenter{\hbox{\includegraphics[width=0.48\linewidth]{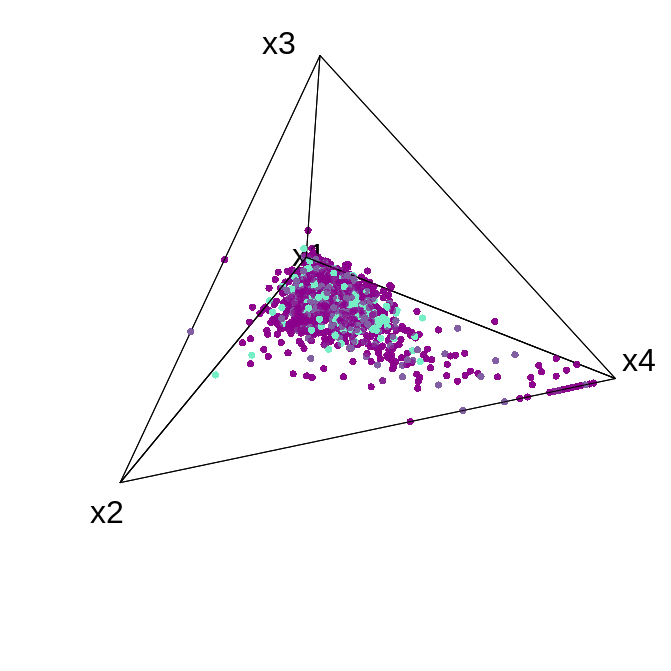}}}$
    $\vcenter{\hbox{\includegraphics[width=0.48\linewidth]{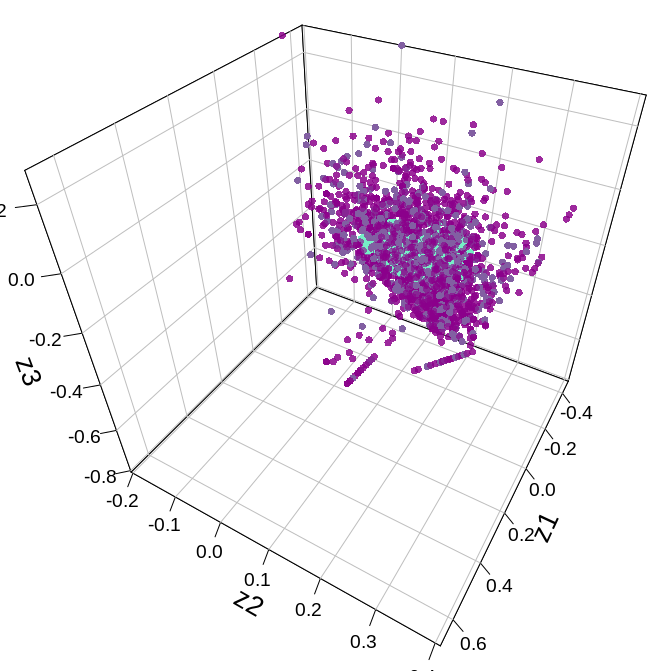}}}$
	\caption{Original data (violet) and kriging predictions (light blue)  with a moderate (10\%) fraction of data with 0s (first setting). Left: simplex. Right: Euclidean space. Top: $\alpha^\star=0.12$. Bottom: $\alpha=1$.}
	\label{fig:euc_cop_zeros10pc_3D}
\end{figure}

\end{document}